\def\SUBMISSION{0}
\def\CONFERENCE{1}
\def\ARXIV{2}
\def\version{\ARXIV}  

\ifnum\version=\SUBMISSION
\documentclass[a4paper,USenglish,autoref,anonymous]{lipics-v2021}
\else
\documentclass[a4paper,USenglish,autoref]{lipics-v2021}
\fi

\usepackage{algorithm}
\usepackage[noend]{algpseudocode}
\usepackage{amsmath,amssymb,amsthm}
\usepackage{booktabs}
\usepackage{cite}
\usepackage{enumitem}
\usepackage{graphicx}
\usepackage{mathtools}
\usepackage{soul}
\usepackage{subcaption} 
\usepackage{tablefootnote}
\usepackage{xcolor}

\newcommand{\bigo}[1]{\ensuremath{\mathcal{O}(#1)}}

\newcommand{\true}{\ensuremath{\textsc{true}}}
\newcommand{\false}{\ensuremath{\textsc{false}}}
\newcommand{\nil}{\ensuremath{\textsc{null}}}
\newcommand{\Gtri}{\ensuremath{G_{\Delta}}}
\newcommand{\Connected}{\ensuremath{\textsc{Connected}}}
\newcommand{\Read}{\ensuremath{\textsc{Read}}}
\newcommand{\Write}{\ensuremath{\textsc{Write}}}
\newcommand{\Contract}{\ensuremath{\textsc{Contract}}}
\newcommand{\Expand}{\ensuremath{\textsc{Expand}}}
\newcommand{\Push}{\ensuremath{\textsc{Push}}}
\newcommand{\Pull}{\ensuremath{\textsc{Pull}}}
\newcommand{\Lock}{\ensuremath{\textsc{Lock}}}
\newcommand{\Unlock}{\ensuremath{\textsc{Unlock}}}
\newcommand{\numAmoebots}{\ensuremath{n}}
\newcommand{\capacity}{\ensuremath{\kappa}}
\newcommand{\demand}{\ensuremath{\delta}}
\newcommand{\tree}{\ensuremath{\mathcal{T}}}
\newcommand{\forest}{\ensuremath{\mathcal{F}}}
\newcommand{\alg}{\ensuremath{\mathcal{A}}}
\newcommand{\energyAlg}{\textsf{Energy-Sharing}}
\newcommand{\forestRepairAlg}{\textsf{Forest-Prune-Repair}}
\newcommand{\erosionAlg}{\textsf{Leader-Election-by-Erosion}}
\newcommand{\hexagonAlg}{\textsf{Hexagon-Formation}}
\newcommand{\xstate}{\ensuremath{\texttt{state}}}
    \newcommand{\source}{\textsc{source}}
    \newcommand{\idle}{\textsc{idle}}
    \newcommand{\xactive}{\textsc{active}}
    \newcommand{\asking}{\textsc{asking}}
    \newcommand{\growing}{\textsc{growing}}
    \newcommand{\pruning}{\textsc{pruning}}
\newcommand{\battery}{\ensuremath{e_{bat}}}
\newcommand{\parent}{\ensuremath{\texttt{parent}}}
\newcommand{\xflag}{\ensuremath{\texttt{flag}}}
\newcommand{\energydist}{\textsc{EnergyDistribution}}
\newcommand{\getpruned}{\textsc{GetPruned}}
\newcommand{\askgrowth}{\textsc{AskGrowth}}
\newcommand{\growforest}{\textsc{GrowForest}}
\newcommand{\harvestenergy}{\textsc{HarvestEnergy}}
\newcommand{\shareenergy}{\textsc{ShareEnergy}}
\newcommand{\algruntime}{\ensuremath{T_\alg(\numAmoebots)}}
\newcommand{\sched}{\ensuremath{\mathcal{S}}}

\makeatletter
\algrenewcommand\ALG@beginalgorithmic{\small}
\algrenewcommand\alglinenumber[1]{\footnotesize #1:}
\makeatother

\algdef{SE}[SUBALG]{Indent}{EndIndent}{}{\algorithmicend\ }%
\algtext*{Indent}
\algtext*{EndIndent}

\newtheorem{convention}{Convention}
\newtheorem{invariant}{Invariant}

\newif\ifcomment
\commenttrue    

\newif\iffigabbrv
\figabbrvfalse   
\newcommand{\figtext}{\iffigabbrv Fig.\else Figure\fi}

\title{\texorpdfstring{Energy-Constrained Programmable Matter\\Under Unfair Adversaries}{Energy-Constrained Programmable Matter Under Unfair Adversaries}}
\titlerunning{Energy-Constrained Programmable Matter Under Unfair Adversaries}

\ifnum\version=\SUBMISSION
\author{Anonymous Author(s)}{Anonymous Affiliation(s)}{}{}{}
\authorrunning{Anonymous Author(s)}
\else
\author{Jamison W. Weber}{School of Computing and Augmented Intelligence, Arizona State University, Tempe, AZ, USA}{jwweber@asu.edu}{https://orcid.org/0000-0002-9573-1783}{}
\author{Tishya Chhabra}{School of Computing and Augmented Intelligence, Arizona State University, Tempe, AZ, USA}{tchhabr2@asu.edu}{https://orcid.org/0000-0002-3555-1078}{}
\author{Andr\'ea W. Richa}{School of Computing and Augmented Intelligence \and Biodesign Center for Biocomputing, Security and Society\\Arizona State University, Tempe, AZ, USA}{aricha@asu.edu}{https://orcid.org/0000-0003-3592-3756}{}
\author{Joshua J. Daymude}{School of Computing and Augmented Intelligence \and Biodesign Center for Biocomputing, Security and Society\\Arizona State University, Tempe, AZ, USA}{jdaymude@asu.edu}{https://orcid.org/0000-0001-7294-5626}{}
\authorrunning{J.\ W.\ Weber, T.\ Chhabra, A.\ W.\ Richa, and J.\ J.\ Daymude}

\Copyright{Jamison W.\ Weber, Tishya Chhabra, Andr\'ea W.\ Richa, and Joshua J.\ Daymude}
\fi



\ifnum\version=\CONFERENCE
\relatedversion{All proofs omitted from this paper due to space constraints can be found in:}
\relatedversiondetails{Full Version}{https://arxiv.org/abs/2309.04898}
\else\fi

\ifnum\version=\SUBMISSION\else
\supplement{Source code for all simulations in this work is openly available as part of AmoebotSim, a visual simulator for the amoebot model of programmable matter.}
\supplementdetails[subcategory={AmoebotSim}]{Software}{https://github.com/SOPSLab/AmoebotSim}

\funding{This work is supported in part by National Science Foundation award CCF-2312537 and by Army Research Office MURI award \#W911NF-19-1-0233.}

\fi

\ifnum\version=\CONFERENCE
\EventEditors{Alysson Bessani, Xavier D\'efago, Yukiko Yamauchi, and Junya Nakamura}
\EventNoEds{4}
\EventLongTitle{27th International Conference on Principles of Distributed Systems (OPODIS 2023)}
\EventShortTitle{OPODIS 2023}
\EventAcronym{OPODIS}
\EventYear{2023}
\EventDate{December 6--8, 2023}
\EventLocation{Tokyo, Japan}
\EventLogo{}
\SeriesVolume{286}
\ArticleNo{7}
\else\fi

\ifnum\version=\CONFERENCE\else
\hideLIPIcs
\fi

\ifnum\version=\SUBMISSION\else
\nolinenumbers
\fi

\begin{document}

\maketitle

\begin{abstract}
    Individual modules of \textit{programmable matter} participate in their system's collective behavior by expending energy to perform actions.
    However, not all modules may have access to the external energy source powering the system, necessitating a local and distributed strategy for supplying energy to modules.
    In this work, we present a general \textit{energy distribution framework} for the \textit{canonical amoebot model} of programmable matter that transforms energy-agnostic algorithms into energy-constrained ones with equivalent behavior and an $\bigo{\numAmoebots^2}$-round runtime overhead---even under an \textit{unfair adversary}---provided the original algorithms satisfy certain conventions.
    We then prove that existing amoebot algorithms for \textit{leader election} (ICDCN 2023) and \textit{shape formation} (Distributed Computing, 2023) are compatible with this framework and show simulations of their energy-constrained counterparts, demonstrating how other unfair algorithms can be generalized to the energy-constrained setting with relatively little effort.
    Finally, we show that our energy distribution framework can be composed with the \textit{concurrency control framework} for amoebot algorithms (Distributed Computing, 2023), allowing algorithm designers to focus on the simpler energy-agnostic, sequential setting but gain the general applicability of energy-constrained, asynchronous correctness.
    \ifnum\version=\SUBMISSION
    
    \smallskip
    \noindent \textbf{This paper is eligible for the best student paper award.}
    \else\fi
\end{abstract}

\ifnum\version=\CONFERENCE
\clearpage
\setcounter{page}{1}
\fi

\section{Introduction}  \label{sec:intro}

\textit{Programmable matter}~\cite{Toffoli1991-programmablematter} is often envisioned as a material composed of simple, homogeneous modules that collectively change the system's physical properties based on environmental stimuli or user input.
These modules participate in the system's overall collective behavior by expending energy to perform internal computation, communicate with their neighbors, and move.
But as the number of modules per collective increases and individual modules are miniaturized from the centimeter/millimeter-scale~\cite{Gilpin2010-robotpebbles,Goldstein2005-programmablematter,Piranda2018-designingquasispherical} to the micro- and nano-scale~\cite{Dolev2016-invivoenergy,Kriegman2020-scalablepipeline,Blackiston2021-cellularplatform}, traditional methods of robotic power supply such as internal battery storage and tethering become infeasible.
Many programmable matter systems instead make use of an external energy source accessible by at least one module and rely on \textit{module-to-module power transfer} to supply the system with energy~\cite{Campbell2005-robottether,Gilpin2010-robotpebbles,Goldstein2009-audiovideo,Piranda2018-designingquasispherical}.
This external energy can be supplied directly to modules in the form of electricity~\cite{Gilpin2010-robotpebbles} or may be ambiently available as light, heat, sound, or chemical energy in the environment~\cite{MacLennan2015-morphogeneticpath,Napp2011-setpointregulation}.
Since energy may not be uniformly accessible to all modules in the system, a strategy for \textit{energy distribution}---sharing energy among modules such that the system can achieve its desired function---is imperative.

Algorithmic theory for programmable matter---including population protocols~\cite{Angluin2006-computationnetworks}, the nubot model~\cite{Woods2013-activeselfassembly}, mobile robots~\cite{Flocchini2019-distributedcomputing}, hybrid programmable matter~\cite{Gmyr2020-formingtile}, and the amoebot model~\cite{Daymude2023-canonicalamoebot,Derakhshandeh2014-amoebotba}---has largely ignored energy constraints, focusing instead on characterizing individual modules' necessary and sufficient capabilities for goal collective behaviors.
Besides a few notable exceptions~\cite{Dolev2016-invivoenergy,Piranda2018-designingquasispherical}, this literature only references energy to justify assumptions (e.g., why a system should remain connected~\cite{Michail2019-transformationcapability}) and ignores the impact of energy usage and distribution on an algorithm's efficiency.
In contrast, both programmable matter practitioners and the modular and swarm robotics literature incorporate energy constraints as influential aspects of algorithm design~\cite{Bartashevich2017-energysavingdecision,Kernbach2013-handbookcollective,Mostaghim2016-energyaware,Pickem2017-robotariumremotely,Wei2012-stayingalivepath}.

This gap motivated the prior \energyAlg\ algorithm for energy distribution~\cite{Daymude2021-bioinspiredenergy} under the \textit{amoebot model} of programmable matter~\cite{Derakhshandeh2014-amoebotba}.
When amoebots do not move and are activated \textit{sequentially and fairly}, \energyAlg\ distributes any necessary energy to all $\numAmoebots$ amoebots within at most $\bigo{\numAmoebots}$ rounds.
Combined with the \forestRepairAlg\ algorithm introduced in the same work to repair energy distribution networks as amoebots move, it was suggested that any amoebot algorithm could be composed with these two to handle energy constraints, though this was only shown for one algorithm in simulation.

In this work, we introduce a general \textit{energy distribution framework} that provably converts any energy-agnostic amoebot algorithm satisfying certain conventions into an \textit{energy-constrained} version that exhibits the same system behavior while also distributing the energy amoebots need to meet the demands of their actions.
In particular, we use the message passing-based \textit{canonical amoebot model}~\cite{Daymude2023-canonicalamoebot} to address the challenges of \textit{unfair} adversarial schedulers---the most general of all fairness assumptions---that can activate any amoebot that is able to perform an action regardless of how long others have been waiting to do the same.
Under an unfair adversary, the prior \forestRepairAlg\ algorithm may not terminate, rendering it unusable for maintaining energy distribution networks.
In contrast, energy-constrained algorithms produced by our framework not only terminate despite unfairness, but do so within an $\bigo{\numAmoebots^2}$-round overhead, where $\numAmoebots$ is the number of amoebots in the system.

\ifnum\version=\SUBMISSION
\medskip
\noindent \textsf{\textbf{Our Contributions.}}
\else
\subparagraph{Our Contributions.}

\fi
We summarize our contributions as follows.
We introduce the \textit{energy distribution framework} that transforms any energy-agnostic amoebot algorithm $\alg$ satisfying some basic conventions and a demand function $\demand$ specifying its energy costs into an energy-constrained algorithm $\alg^\demand$ that provably exhibits equivalent behavior to $\alg$, even under an unfair adversary, while incurring at most an $\bigo{\numAmoebots^2}$-round runtime overhead (Section~\ref{sec:framework}).
We then prove that both the \erosionAlg\ algorithm from~\cite{Briones2023-invitedpaper} and the \hexagonAlg\ algorithm from~\cite{Daymude2023-canonicalamoebot} satisfy the framework's conventions and show simulations of their energy-constrained counterparts produced by the framework (Section~\ref{sec:edfcompatible}).

Finally, we prove that a particular class of ``expansion-corresponding'' algorithms that are compatible with the established \textit{concurrency control framework} for amoebot algorithms~\cite{Daymude2023-canonicalamoebot}---including \erosionAlg\ and \hexagonAlg---remain so after transformation by our energy distribution framework, establishing a general pipeline for lifting energy-agnostic, non-concurrent amoebot algorithms (which are easier to design and analyze) to the more realistic \textit{energy-constrained, asynchronous setting} (Section~\ref{sec:concurrency}).

\section{Preliminaries}  \label{sec:prelims}

We begin with necessary background on the (canonical) amoebot model in Section~\ref{subsec:model} and our extensions for energy constraints in Section~\ref{subsec:energymodel}.

\subsection{The Amoebot Model} \label{subsec:model}

In the \textit{canonical amoebot model}~\cite{Daymude2023-canonicalamoebot}, programmable matter consists of individual, homogeneous computational elements called \textit{amoebots}.
The structure of an amoebot system is represented as a subgraph of an infinite, undirected graph $G = (V,E)$ where $V$ represents all relative positions an amoebot can occupy and $E$ represents all atomic movements an amoebot can make.
Each node in $V$ can be occupied by at most one amoebot at a time.
Here, we adopt the \underbar{geometric space variant} in which $G = \Gtri$, the triangular lattice (\figtext~\ref{fig:model:lattice}).

\begin{figure}
    \centering
    \begin{subfigure}{.3\textwidth}
        \centering
        \includegraphics[width=\textwidth]{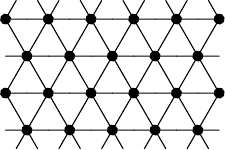}
        \caption{\centering}
        \label{fig:model:lattice}
    \end{subfigure}
    \hfill
    \begin{subfigure}{.3\textwidth}
        \centering
        \includegraphics[width=\textwidth]{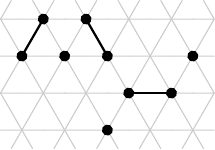}
        \caption{\centering}
        \label{fig:model:particles}
    \end{subfigure}
    \hfill
    \begin{subfigure}{.3\textwidth}
        \centering
        \includegraphics[width=0.8\textwidth]{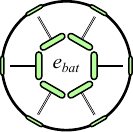}
        \caption{\centering}
        \label{fig:model:energy}
    \end{subfigure}
    \caption{\textit{The Amoebot Model}.
    (a) A section of the triangular lattice $\Gtri$ used in the geometric space variant; nodes of $V$ are shown as black circles and edges of $E$ are shown as black lines.
    (b) Expanded and contracted amoebots; $\Gtri$ is shown in gray and amoebots are shown as black circles.
    Amoebots with a black line between their nodes are expanded.
    (c) When modeling energy, each amoebot $A$ has a battery $A.\battery$ storing energy for its own use and for sharing with its neighbors.}
    \label{fig:model}
\end{figure}

An amoebot has two \textit{shapes}: \textsc{contracted}, meaning it occupies a single node in $V$, and \textsc{expanded}, meaning it occupies a pair of adjacent nodes in $V$ (\figtext~\ref{fig:model:particles}).
Each amoebot keeps a collection of ports---one for each edge incident to the node(s) it occupies---that are labeled consecutively according to its own local, persistent \textit{orientation}.
All results in this work allow for \underbar{assorted orientations}, meaning amoebots may disagree on both direction (which incident edge points ``north'') and chirality (clockwise vs.\ counter-clockwise rotation).
Two amoebots occupying adjacent nodes are said to be \textit{neighbors}.
Although each amoebot is \textit{anonymous}, lacking a unique identifier, an amoebot can locally identify its neighbors using their port labels.
In particular, amoebots $A$ and $B$ connected via ports $p_A$ and $p_B$ know each other's orientations and labels for $p_A$ and $p_B$.

Each amoebot has memory whose size is a model variant; all results in this work assume \underbar{constant-size memories}.
An amoebot’s memory consists of two parts: a persistent \textit{public memory} that is only accessible to an amoebot algorithm via communication operations (defined next) and a volatile \textit{private memory} that is directly accessible by amoebot algorithms for temporary variables, computation, etc.
\textit{Operations} define the programming interface for amoebot algorithms to communicate and move (see~\cite{Daymude2023-canonicalamoebot} for details):
\begin{itemize}
    \item The \Connected\ operation tests the presence of neighbors.
    $\Connected(p)$ returns \true\ if and only if there is a neighbor connected via port $p$.
    
    \item The \Read\ and \Write\ operations exchange information in public memory.
    $\Read(p, x)$ issues a request to read the value of a variable $x$ in the public memory of the neighbor connected via port $p$ while $\Write(p, x, x_{val})$ issues a request to update its value to $x_{val}$.
    If $p = \bot$, an amoebot's own public memory is accessed instead of a neighbor's.

    \item An expanded amoebot can \Contract\ into either node it occupies; a contracted amoebot can \Expand\ into an unoccupied adjacent node.
    Neighboring amoebots can coordinate their movements in a \textit{handover}, which occurs in one of two ways.
    A contracted amoebot $A$ can \Push\ an expanded neighbor $B$ by expanding into a node occupied by $B$, forcing it to contract.
    Alternatively, an expanded amoebot $B$ can \Pull\ a contracted neighbor $A$ by contracting, forcing $A$ to expand into the node it is vacating.
\end{itemize}

Amoebot algorithms are sets of \textit{actions}, each of the form $\langle label\rangle : \langle guard\rangle \to \langle operations\rangle$.
An action's \textit{label} specifies its name.
Its \textit{guard} is a Boolean predicate determining whether an amoebot $A$ can execute it based on the ports $A$ has connections on---i.e., which nodes adjacent to $A$ are (un)occupied---and information from the public memories of $A$ and its neighbors.
An action is \textit{enabled} for an amoebot $A$ if its guard is true for $A$, and an amoebot is \textit{enabled} if it has at least one enabled action.
An action's \textit{operations} specify the finite sequence of operations and computation in private memory to perform if this action is executed.

An amoebot is \textit{active} while executing an action and is \textit{inactive} otherwise.
An \textit{adversary} controls the timing of amoebot activations and the resulting action executions, whose \textit{concurrency} and \textit{fairness} are assumption variants.
In this work, we consider two concurrency variants: \underbar{sequential}, in which at most one amoebot can be active at a time; and \underbar{asynchronous}, in which any set of amoebots can be simultaneously active.
We consider the most general fairness variant: \underbar{unfair}, in which the adversary may activate any enabled amoebot.

An amoebot algorithm's time complexity is evaluated in terms of \textit{rounds} representing the time for the slowest continuously enabled amoebot to execute a single action.
Let $t_i$ denote the time at which round $i \in \{0, 1, 2, \ldots\}$ starts, where $t_0 = 0$, and let $\mathcal{E}_i$ denote the set of amoebots that are enabled or already executing an action at time $t_i$.
Round $i$ completes at the earliest time $t_{i+1} > t_i$ by which every amoebot in $\mathcal{E}_i$ either completed an action execution or became disabled at some time in $(t_i, t_{i+1}]$.
\ifnum\version=\ARXIV
Depending on the adversary's concurrency, action executions may span more than one round.
\else\fi

\subsection{Extensions for Energy Modeling} \label{subsec:energymodel}

In addition to the standard model, we introduce new assumptions and terminology specific to modeling energy in amoebot systems.
We consider amoebot systems that are finite, initially connected, and contain at least one \textit{source amoebot} with access to an external energy source.
Although system connectivity is not generally required by the (canonical) amoebot model, it is necessary for sharing energy from a single source amoebot to the rest of the system via module-to-module power transfer.
Each amoebot $A$ has an \textit{energy battery} denoted $A.\battery$ with capacity $\capacity > 0$ representing energy that $A$ can use to perform actions or share with its neighbors (\figtext~\ref{fig:model:energy}).
In this paper, we assume $\capacity = \Theta(1)$ is a fixed integer constant that does not scale with the number of amoebots $\numAmoebots$, but all results in this paper would hold even if $\capacity = \bigo{\numAmoebots}$.
Source amoebots can harvest energy directly into their batteries while those without access depend on their neighbors to share with them.
In either case, we assume an amoebot transfers at most a single unit of energy per activation.\footnote{One could assume that the battery capacity $\capacity > 0$ is any positive real number and that the energy demands are $\demand : \alg \to (0, \capacity]$.
However, this generality complicates our analysis without meaningfully extending our results, so we make the simplifying assumption that there exists a fundamental unit of energy that divides all action demands $\demand(\alpha_i)$ and the battery capacity $\capacity$.}
For modeling purposes, we treat $A.\battery$ as a variable stored in the public memory of $A$.
An amoebot $A$ harvesting energy from an external source can be expressed as $\Write(\bot, \battery, \Read(\bot, \battery) + 1)$ and likewise an amoebot $A$ transferring energy to a neighbor $B$ connected via a port $p$ is a pair of operations $\Write(\bot, \battery, \Read(\bot, \battery) - 1)$ and $\Write(p, \battery, \Read(p, \battery) + 1)$.

The energy costs for an amoebot algorithm $\alg = \{[\alpha_i : g_i \to ops_i] : i \in \{1, \ldots, m\}\}$ are given by a \textit{demand function} $\demand : \alg \to \{1, 2, \ldots, \capacity\}$; i.e., an amoebot must use $\demand(\alpha_i)$ energy to execute action $\alpha_i$.
Energy is incorporated into actions $\alpha_i \in \alg$ by (1) including $A.\battery \geq \demand(\alpha_i)$ in each guard $g_i$ and (2) setting $\Write(\bot, \battery, \Read(\bot, \battery) - \demand(\alpha_i))$ as the first operation of $ops_i$ to spend the corresponding amount of energy.
\ifnum\version=\ARXIV

Finally, we give two definitions central to our energy distribution results.
The first characterizes amoebots that, due to a lack of energy in their batteries, may be blocked from executing an action.
The second names our algorithm regimes of interest.


\begin{definition} \label{def:deficient}
    An amoebot $A$ is \underbar{deficient} w.r.t.\ an action $\alpha_i \in \alg$ if $A.\battery < \demand(\alpha_i)$.
\end{definition}

\begin{definition} \label{def:agnosticconstrained}
    An amoebot algorithm $\alg$ is \underbar{energy-agnostic} if it is not associated with a demand function $\demand$ and is \underbar{energy-constrained} (w.r.t.\ $\demand$) otherwise.
\end{definition}
\else
An amoebot $A$ is \textit{deficient} w.r.t.\ an action $\alpha_i \in \alg$ if $A.\battery < \demand(\alpha_i)$.
An amoebot algorithm $\alg$ is \textit{energy-agnostic} if it is not associated with a demand function $\demand$ and is \textit{energy-constrained} (w.r.t.\ $\demand$) otherwise.
\fi

The remainder of this paper is dedicated to transforming amoebot algorithms that were designed for the energy-agnostic setting into algorithms with equivalent behavior in the energy-constrained setting w.r.t.\ any valid demand function under an unfair adversary.

\section{A General Framework for Energy-Constrained Algorithms} \label{sec:framework}

Amoebot algorithm designers prove the correctness of their algorithms with respect to a \textit{safety} condition (related to the desired system behavior) and a \textit{liveness} condition (ensuring that until this behavior is achieved, some amoebot can make progress towards it).
Moving from energy-agnosticism to respecting energy constraints does not affect safety, but may threaten liveness.
Some amoebot that was critical to achieving progress in the energy-agnostic setting may now be deficient under the constraints of actions' energy costs, deadlocking the system until it is provided with sufficient energy.
Since not all amoebots have access to an external energy source, simply waiting to recharge is not an option.
There must be an active strategy for energy distribution embedded in any energy-constrained algorithm.

Instead of placing the burden on algorithm designers to create bespoke implementations of energy distribution for each algorithm, we introduce a general \textit{energy distribution framework}.
This framework transforms energy-agnostic algorithms $\alg$ that terminate under an unfair adversary and satisfy certain \textit{conventions} into algorithms $\alg^\demand$ that are energy-constrained w.r.t.\ any valid demand function $\demand$ and retain their unfair correctness.
We give a narrative description and pseudocode for our framework in Section~\ref{subsec:edf} and analyze it in Section~\ref{subsec:edfanalysis}.

\subsection{The Energy Distribution Framework} \label{subsec:edf}

Our \textit{energy distribution framework} (Algorithm~\ref{alg:framework}) takes as input any energy-agnostic amoebot algorithm $\alg = \{[\alpha_i : g_i \to ops_i] : i \in \{1, \ldots, m\}\}$ and demand function $\demand : \alg \to \{1, 2, \ldots, \capacity\}$ and outputs an energy-constrained algorithm
\ifnum\version=\ARXIV
\[\alg^\demand = \{[\alpha_i^\demand : g_i^\demand \to ops_i^\demand] : i \in \{1, \ldots, m\}\} \cup \{\alpha_\energydist\},\]
\else
$\alg^\demand = \{[\alpha_i^\demand : g_i^\demand \to ops_i^\demand] : i \in \{1, \ldots, m\}\} \cup \{\alpha_\energydist\}$,
\fi
where actions $\alpha_i^\demand$ are energy-constrained versions of the original actions and $\alpha_\energydist$ is a new action that handles energy distribution.
\ifnum\version=\ARXIV
Algorithm $\alg^\demand$ will achieve the same system behavior as algorithm $\alg$ so long as $\alg$ satisfies certain conventions.
Formally, we say:
\else
Algorithm $\alg^\demand$ will achieve the same system behavior as algorithm $\alg$ so long as $\alg$ satisfies certain conventions:
\fi

\begin{definition} \label{def:edfcompatible}
    An energy-agnostic amoebot algorithm $\alg$ is \underbar{energy-compatible}---i.e., it is compatible with the energy distribution framework---if every (unfair) sequential execution of $\alg$ terminates and $\alg$ satisfies Conventions~\ref{conv:valid}--\ref{conv:connect} (defined below).
\end{definition}

Our first two conventions are taken directly from the analogous concurrency control framework for amoebot algorithms~\cite{Daymude2023-canonicalamoebot}.
The first convention requires an algorithm's actions to execute successfully in isolation, allowing the framework to ignore invalid actions like attempting to \Read\ on a disconnected port or \Expand\ when already expanded.
Formally, we define a \textit{system configuration} as the mapping of amoebots to the node(s) they occupy and the contents of each amoebot's public memory.
Throughout the remainder of this paper, we assume configurations are \textit{legal}; i.e., they meet the requirements of the amoebot model.

\begin{convention}[Validity] \label{conv:valid}
    All actions $\alpha$ of an amoebot algorithm $\alg$ should be \underbar{valid}, i.e., for all (legal) system configurations in which $\alpha$ is enabled for some amoebot $A$, the execution of $\alpha$ by $A$ should be successful whenever all other amoebots are inactive.
\end{convention}

The second convention defines a common structure for an algorithm's actions by controlling the order and number of their operations, similar to the ``look-compute-move'' paradigm in the mobile robots literature~\cite{Flocchini2019-distributedcomputing}.

\begin{convention}[Phase Structure] \label{conv:phases}
    Each action of an amoebot algorithm $\alg$ should structure its operations as: (1) a \underbar{compute phase}, during which an amoebot performs a finite amount of computation and a finite sequence of \Connected, \Read, and \Write\ operations, and (2) a \underbar{move phase}, during which an amoebot performs at most one movement operation decided upon in the compute phase.
    In particular, no action should use the canonical amoebot model's concurrency control operations, \Lock\ and \Unlock.
\end{convention}

Our third and final convention is specific to the energy distribution framework.
Recall from Section~\ref{subsec:energymodel} that we consider amoebot systems that are initially connected.
This last convention requires an algorithm to maintain system connectivity throughout its execution, ensuring that every amoebot has a path to a source amoebot with access to external energy.

\begin{convention}[Connectivity] \label{conv:connect}
    All system configurations reachable by any sequential execution of an amoebot algorithm $\alg$ starting in a connected configuration must also be connected.
\end{convention}

\begin{algorithm*}[tp]
    \caption{Energy Distribution Framework for Amoebot $A$} \label{alg:framework}
    \begin{algorithmic}[1]
        \Statex \textbf{Input}: An energy-compatible algorithm $\alg = \{[\alpha_i : g_i \to ops_i] : i \in \{1, \ldots, m\}\}$ and a demand function $\demand : \alg \to \{1, 2, \ldots, \capacity\}$.
        \For {each action $[\alpha_i : g_i \to ops_i] \in \alg$} construct action $\alpha_i^\demand : g_i^\demand \to ops_i^\demand$ as:  \label{alg:framework:alphai_start}
            \State Set $g_i^\demand \gets \big(g_i \wedge (A.\battery \geq \demand(\alpha_i)) \wedge (\forall B \in N(A) \cup \{A\} : B.\xstate \not\in \{\idle, \pruning\})\big)$.
            \State Set $ops_i^\demand \gets$ ``Do:
            \Indent
                \State $\Write(\bot, \battery, \Read(\bot, \battery) - \demand(\alpha_i))$.
                \State Execute the compute phase of $ops_i$.
                \If {the movement phase of $ops_i$ contains a movement operation $M_i$}
                    \If {$M_i$ is $\Contract(\,)$ or $\Pull(p)$}
                        \State $\Write(\bot, \parent, \nil)$ and \Call{Prune}{\,}.  \label{alg:framework:alphai_prune}
                    \ElsIf {$M_i$ is $\Push(p)$} \label{alg:framework:push_start}
                        \State $\Write(\bot, \parent, \nil)$ and $\Write(p, \parent, \nil)$.
                        \State $\Write(\bot, \xstate, \pruning)$ and $\Write(p, \xstate, \pruning)$. \label{alg:framework:push_end}
                    \EndIf
                    \State Execute $M_i$.''
                \EndIf
            \EndIndent
        \EndFor  \label{alg:framework:alphai_end}
        \State Construct $\alpha_\energydist : g_\energydist \to ops_\energydist$ as:  \label{alg:framework:energydistaction_start}
        \Indent
            \State Set $g_\energydist \gets \bigvee_{g \in \mathcal{G}} (g)$, where $\mathcal{G} = \{$  \label{alg:framework:energydistguard_start}
            \Indent
                \State \hspace{-1.8mm} $\begin{array}{l@{}l}
                    g_\getpruned &= (A.\xstate = \pruning), \\
                    g_\askgrowth &= (A.\xstate = \xactive) \wedge (A \text{ has an \idle\ neighbor or \asking\ child}), \\
                    g_\growforest &= (A.\xstate = \growing) \;\vee \\
                    &\quad \big((A.\xstate = \source) \wedge (A \text{ has an \idle\ neighbor or \asking\ child})\big), \\
                    g_\harvestenergy &= (A.\xstate = \source) \wedge (A.\battery < \capacity), \\
                    g_\shareenergy &= (A.\xstate \not\in \{\idle, \pruning\}) \;\wedge \\
                    &\quad (A.\battery \geq 1) \wedge (A \text{ has a child } B : B.\battery < \capacity)\}
                \end{array}$  \label{alg:framework:energydistguard_end}
            \EndIndent
            \State Set $ops_\energydist \gets$ ``Do:  \label{alg:framework:energydistops_start}
            \Indent
                \If {$g_\getpruned$} \Call{Prune}{\,}.  \Comment{\getpruned}  \label{alg:framework:getpruned}
                \EndIf
                \If {$g_\askgrowth$} $\Write(\bot, \xstate, \asking)$. \Comment{\askgrowth} \label{alg:framework:askgrowth}
                \EndIf
                \If {$g_\growforest$}  \Comment{\growforest}  \label{alg:framework:growforest_start}
                    \For {each port $p$ for which $\Connected(p) = \true$ and $\Read(p, \xstate) = \idle$}
                        \State $\Write(p, \parent, p')$, where $p'$ is any port of the neighbor on port $p$ facing $A$.
                        \State $\Write(p, \xstate, \xactive)$.  \label{alg:framework:growforest_addchild}
                    \EndFor
                    \For {each port $p \in \Call{Children}{\,} : (\Read(p, \xstate) = \asking)$}
                        \State $\Write(p, \xstate, \growing)$.
                    \EndFor
                    \If {$\Read(\bot, \xstate) = \growing$} $\Write(\bot, \xstate, \xactive)$.  \label{alg:framework:growforest_end}
                    \EndIf
                \EndIf
                \If {$g_\harvestenergy$} $\Write(\bot, \battery, \Read(\bot, \battery) + 1)$.  \Comment{\harvestenergy}  \label{alg:framework:harvestenergy}
                \EndIf
                \If {$g_\shareenergy$}  \Comment{\shareenergy}  \label{alg:framework:shareenergy_start}
                    \State Let port $p \in \Call{Children}{\,}$ be one for which $\Read(p, \battery) < \capacity$.
                    \State $\Write(\bot, \battery, \Read(\bot, \battery) - 1)$.
                    \State $\Write(p, \battery, \Read(p, \battery) + 1)$.''  \label{alg:framework:shareenergy_end}
                \EndIf
            \EndIndent
        \EndIndent  \label{alg:framework:energydistops_end}
        \State \Return $\alg^\demand = \{[\alpha_i^\demand : g_i^\demand \to ops_i^\demand] : i \in \{1, \ldots, m\}\} \cup \{\alpha_\energydist\}$.
        \Statex 
        \Function {Children}{\,}
            \State \Return $\{\text{ports } p : \Connected(p) \wedge (\Read(p, \parent) \text{ points to } A)\}$.
        \EndFunction
        \Function {Prune}{\,}  \label{alg:frameworkhelper:prune_start}
            \For {each port $p \in \Call{Children}{\,}$}  \label{alg:frameworkhelper:prunechildren_start}
                \State $\Write(p, \xstate, \pruning)$.  \label{alg:frameworkhelper:pruneset}
                \State $\Write(p, \parent, \nil)$.  \label{alg:frameworkhelper:prunechildren_end}
            \EndFor
            \If {$\Read(\bot, \xstate) \neq \source$} $\Write(\bot, \xstate, \idle)$.
            \EndIf  \label{alg:frameworkhelper:prunereset_end}
        \EndFunction  \label{alg:frameworkhelper:prune_end}
    \end{algorithmic}
\end{algorithm*}

\begin{table}[t]
    \centering
    \caption{Variables used in the Energy Distribution Framework.}
    \label{tab:frameworkvariables}
    \begin{tabular}{llll}
        \toprule
        \textbf{Variable} & \textbf{Notation} &\textbf{Domain} &\textbf{Initialization} \\
        \midrule
        Forest State & $\xstate$ & \begin{tabular}{@{}l@{}}
            $\{\source, \idle, \xactive,$ \\
            $\asking, \growing, \pruning\}$
        \end{tabular} & $\left\{
        \begin{array}{@{}ll@{}}
            \source & \text{if source amoebot}; \\
            \idle & \text{otherwise}.
        \end{array} \right.$ \\
        Parent Pointer & $\parent$ & $\{\nil, 0, \ldots, 9\}$\tablefootnote{Amoebots maintain one port per incident lattice edge (see Section~\ref{subsec:model}), so an expanded amoebot has ten ports despite having a maximum of eight neighbors.} & \nil \\
        Battery Energy & $\battery$ & $\{0, 1, 2, \ldots, \capacity\}$ & 0 \\
        \bottomrule
    \end{tabular}
\end{table}

\ifnum\version=\SUBMISSION
\noindent \textsf{\textbf{Framework Overview.}}
\else
\subparagraph{Framework Overview.}

\fi
With the conventions defined, we now describe how the energy distribution framework (Algorithm~\ref{alg:framework}) transforms an energy-compatible algorithm $\alg$ and a demand function $\demand : \alg \to \{1, 2, \ldots, \capacity\}$ into an energy-constrained algorithm $\alg^\demand$ with ``equivalent'' behavior (defined formally in Section~\ref{subsec:edfanalysis}).
At a high level, $\alg^\demand$ works as follows.
The amoebot system first self-organizes as a spanning forest $\forest$ rooted at source amoebots with access to external energy sources.
Energy is harvested by source amoebots and transferred from parents to children in $\forest$ as there is need.
Amoebots spend energy on enabled actions of algorithm $\alg$ until they become deficient, when they will once again need to wait to recharge.
This process repeats until termination, which must occur since $\alg$ is energy-compatible.

Algorithm $\alg^\demand$ comprises two types of actions.
First, every action $\alpha_i \in \alg$ is transformed into an energy-constrained version $\alpha_i^\demand \in \alg^\demand$ (Algorithm~\ref{alg:framework}, Lines~\ref{alg:framework:alphai_start}--\ref{alg:framework:alphai_end}).
By including $A.\battery \geq \demand(\alpha_i)$ in its guard $g_i^\demand$ and spending $\demand(\alpha_i)$ energy at the start of its operations $ops_i^\demand$, the transformed action $\alpha_i^\demand$ is only executed if there is sufficient energy to do so and any such execution spends the corresponding energy.
The guard $g_i^\demand$ also ensures any amoebot executing an $\alpha_i^\demand$ action and all of its neighbors are part of the forest structure $\forest$.

Second, there is a singular $\alpha_\energydist$ action that defines how amoebots self-organize as a spanning forest and distribute energy throughout the system (Algorithm~\ref{alg:framework}, Lines~\ref{alg:framework:energydistaction_start}--\ref{alg:framework:energydistops_end}).
Its operations are organized into five blocks---\getpruned, \askgrowth, \growforest, \harvestenergy, and \shareenergy---each of which has a corresponding logical predicate in the set $\mathcal{G}$.
These predicates appear in the guard $\bigvee_{g \in \mathcal{G}} (g)$, which ensures that $\alpha_\energydist$ is only enabled when its execution would progress towards distributing energy to deficient amoebots.
The latter is critical for proving that $\alg^\demand$ achieves energy distribution even under an unfair adversary, which we show in Section~\ref{subsec:edfanalysis}.
The remainder of this section details the five blocks; their local variables are summarized in Table~\ref{tab:frameworkvariables}.

\ifnum\version=\SUBMISSION
\medskip
\noindent \textsf{\textbf{Forming and Maintaining a Spanning Forest.}}
\else
\subparagraph{Forming and Maintaining a Spanning Forest.}

\fi
Recall from Section~\ref{subsec:energymodel} that we consider amoebot systems that are initially connected and contain at least one source amoebot with access to an external energy source.
The \getpruned, \askgrowth, and \growforest\ blocks (Algorithm~\ref{alg:framework}, Lines~\ref{alg:framework:getpruned}--\ref{alg:framework:growforest_end}) continuously organize the amoebot system as a spanning forest $\forest$ of trees rooted at the source amoebot(s).
These trees act as an acyclic resource distribution network for energy transfers, which is important for avoiding non-termination under an unfair adversary.

The well-established \textit{spanning forest primitive}~\cite{Daymude2019-computingprogrammable} and the recent \textit{feather tree formation} algorithm~\cite{Kostitsyna2022-briefannouncement} are both guaranteed to organize an amoebot system as a spanning forest $\forest$ under an unfair sequential adversary, assuming no parent--child relationship in $\forest$ is ever disrupted after it is formed.
However, many amoebot algorithms $\alg$---and by extension, the actions $\alpha_i^\demand$ of algorithms $\alg^\demand$---cause amoebots to move, partitioning $\forest$ into ``unstable'' trees whose connections to source amoebots have been disrupted and ``stable'' trees that remain rooted at sources.
This necessitates a protocol for dynamically repairing $\forest$ as amoebots move.
To this end, the earlier \forestRepairAlg\ algorithm~\cite{Daymude2021-bioinspiredenergy} was designed to ``prune'' unstable trees, allowing their amoebots to rejoin stable trees.
Unfortunately, \forestRepairAlg\ requires fairness for termination, which we do not have here.
In the following, we describe a new algorithm that dynamically maintains $\forest$ under an unfair sequential adversary.

Each amoebot has a $\xstate$ variable that is initialized to \source\ for source amoebots and \idle\ for all others.
Additionally, each amoebot has a $\parent$ pointer indicating the port incident to their parent in the forest $\forest$; these pointers are initially set to $\nil$.
A source amoebot adopts its \idle\ neighbors into its tree by making them \xactive\ and setting their $\parent$ pointers to itself (\growforest, Algorithm~\ref{alg:framework}, Lines~\ref{alg:framework:growforest_start}--\ref{alg:framework:growforest_addchild}).
\xactive\ amoebots, however, must ask the source amoebot at the root of their tree for permission before adopting their \idle\ neighbors (\askgrowth, Algorithm~\ref{alg:framework}, Line~\ref{alg:framework:askgrowth}).
\ifnum\version=\ARXIV
Although indirect, this ensures that \idle\ amoebots only join trees that are (or were recently) stable, stopping the unfair adversary from creating non-terminating executions (see Lemma~\ref{lem:forestblocksfinite}).
\else
Although indirect, this ensures that \idle\ amoebots only join trees that are (or were recently) stable, stopping the unfair adversary from creating non-terminating executions (see Lemma~\ref{lem:energyrunfinite}).
\fi
Specifically, an \xactive\ amoebot with an \idle\ neighbor becomes \asking.
Any \xactive\ amoebot with an \asking\ child also becomes \asking, propagating this ``asking signal'' towards the tree's source amoebot.
When the source amoebot receives this asking signal, it updates all its \asking\ children to \growing, granting them permission to grow the tree.
A \growing\ amoebot adopts its \idle\ neighbors as \xactive\ children, updates its \asking\ children to \growing, and resets its $\xstate$ to \xactive.
This process repeats until no \idle\ amoebots remain.

If an amoebot's movement during an $\alpha_i^\demand$ execution would disrupt $\forest$, it initiates a pruning process to dissolve disrupted subtrees.
Amoebots performing \Contract\ or \Pull\ movements must prune immediately since their movement may disconnect them from their neighbors; \Push\ movements instead make the two involved amoebots \pruning, which will cause them to prune during their next action.
When an amoebot prunes, it makes its children \pruning\ and resets both its own and its children's $\parent$ pointers, severing them from their tree (Algorithm~\ref{alg:framework}, Lines~\ref{alg:framework:alphai_prune} and~\ref{alg:frameworkhelper:prunechildren_start}--\ref{alg:frameworkhelper:prunechildren_end}).
If it is not a source, it also becomes \idle\ (Algorithm~\ref{alg:framework}, Line~\ref{alg:frameworkhelper:prunereset_end}).
The \getpruned\ block ensures that any \pruning\ amoebot does the same, dissolving the unstable tree (Algorithm~\ref{alg:framework}, Line~\ref{alg:framework:getpruned}).
These newly \idle\ amoebots are then collected into stable trees by the \askgrowth\ and \growforest\ blocks as described above.

\ifnum\version=\SUBMISSION
\medskip
\noindent \textsf{\textbf{Sharing Energy.}}
\else
\subparagraph{Sharing Energy.}

\fi
The \harvestenergy\ and \shareenergy\ blocks (Algorithm~\ref{alg:framework}, Lines~\ref{alg:framework:harvestenergy}--\ref{alg:framework:shareenergy_end}) define how source amoebots harvest energy from external energy sources and how all non-\idle, non-\pruning\ amoebots transfer energy to their neighbors, respectively.
If its battery is not already full, a source amoebot harvests a unit of energy from its external energy source into its own battery.
Any non-\idle, non-\pruning\ amoebot with at least one unit of energy to share and a child whose battery is not full will then transfer a unit of energy from its own battery to that of its child.

\subsection{Analysis} \label{subsec:edfanalysis}

\ifnum\version=\ARXIV
In this section, we prove the following theorem.
\else
In this section, we sketch the results building to the following theorem.
\fi
Informally, it states that an energy-constrained algorithm $\alg^\demand$ produced by the energy distribution framework (1) only yields system outcomes that could have been achieved by the original energy-agnostic algorithm $\alg$, provided $\alg$ is energy-compatible, and (2) incurs an $\bigo{\numAmoebots^2}$ runtime overhead.

\begin{theorem} \label{thm:main}
    Consider any energy-compatible amoebot algorithm $\alg$ and demand function $\demand : \alg \to \{1, 2, \ldots, \capacity\}$, and let $\alg^\demand$ be the algorithm produced from $\alg$ and $\demand$ by the energy distribution framework (Algorithm~\ref{alg:framework}).
    Let $C_0$ be any (legal) connected initial configuration for $\alg$ and let $C_0^\demand$ be its extension for $\alg^\demand$ that designates at least one source amoebot and adds the energy distribution variables with their initial values (Table~\ref{tab:frameworkvariables}) to all amoebots.
    Then for any configuration $C^\demand$ in which an unfair sequential execution of $\alg^\demand$ starting in $C_0^\demand$ terminates, there exists an unfair sequential execution of $\alg$ starting in $C_0$ that terminates in a configuration $C$ that is identical to $C^\demand$ modulo the energy distribution variables.
    Moreover, if all unfair sequential executions of $\alg$ on $\numAmoebots$ amoebots terminate after at most $\algruntime$ action executions, then any unfair sequential execution of $\alg^\demand$ on $\numAmoebots$ amoebots terminates in $\bigo{\numAmoebots^2\algruntime}$ rounds.
\end{theorem}

\ifnum\version=\ARXIV
\subparagraph{Analysis Overview.}

We outline our analysis as follows.
We start by considering an arbitrary sequential execution $\sched^\demand$ of $\alg^\demand$ starting in $C_0^\demand$.
One way of conceptualizing $\sched^\demand$ is as a sequence of \textit{energy runs}---i.e., maximal sequences of consecutive $\alpha_\energydist$ executions---that are delineated by sequences of $\alpha_i^\demand$ executions.
In fact, $\sched^\demand$ contains only a finite number of $\alpha_i^\demand$ executions (and thus a finite number of energy runs) because the corresponding sequence of $\alpha_i$ executions forms a possible sequential execution $\sched_\alpha$ of $\alg$ (Lemma~\ref{lem:equivalence}), which must terminate because $\alg$ is energy-compatible.
It is exactly this execution $\sched_\alpha$ of $\alg$ that we will argue terminates in a configuration $C$ corresponding to the final configuration $C^\demand$ of $\sched^\demand$.

Of course, we have not yet shown that $\sched^\demand$ terminates at all under an unfair adversary, let alone in a final configuration corresponding to $\sched_\alpha$.
To do so, we will show that any energy run in $\sched^\demand$ is finite (Lemmas~\ref{lem:forestblocksfinite} and~\ref{lem:energyblocksfinite}); specifically, it either reaches a configuration where $\alpha_\energydist$ is disabled for all $\numAmoebots$ amoebots within $\bigo{\numAmoebots^2}$ rounds, or ends earlier because some $\alpha_i^\demand$ action is executed (Lemmas~\ref{lem:stabletime} and~\ref{lem:rechargetime}).
Since each energy run terminates within $\bigo{\numAmoebots^2}$ rounds and is delineated by a sequence of $\alpha_i^\demand$ executions, each $\alpha_i^\demand$ execution in $\sched^\demand$ can be mapped to an $\alpha_i$ execution in $\sched_\alpha$, and $\sched_\alpha$ contains at most $\algruntime$ action executions, we conclude that $\sched^\demand$ is not only finite, but terminates within $\bigo{\numAmoebots^2\algruntime}$ rounds.

Once it is established that both $\sched^\demand$ and $\sched_\alpha$ terminate, we argue that their respective final configurations $C^\demand$ and $C$ are identical (modulo the energy distribution variables).
Because every $\alpha_i^\demand$ execution in $\sched^\demand$ corresponds to a possible $\alpha_i$ execution in $\sched_\alpha$ (Lemma~\ref{lem:equivalence}), we know that any configuration reachable by $\sched^\demand$ is also reachable by $\sched_\alpha$.
So $\sched_\alpha$ must be able to reach a configuration $C$ corresponding to $C^\demand$, but we need to show that it will also terminate there; i.e., that the energy distribution aspects of $\alg^\demand$ don't impede it from making as much progress as $\alg$.
This will follow from the above energy run arguments, concluding the analysis.

\bigskip

We begin our analysis with two sets of invariants maintained by the energy distribution framework that we will reference repeatedly.
The first set describes useful properties of energy runs, i.e., maximal sequences of consecutive $\alpha_\energydist$ executions.
The second set characterizes all configurations reachable by algorithm $\alg^\demand$.

\begin{invariant} \label{inv:energyrun}
    In any energy run of any sequential execution of $\alg^\demand$ starting in $C_0^\demand$,
    \begin{enumerate}[label=(\alph*),ref=\ref{inv:energyrun}\alph*,leftmargin=1cm]
        \item \label{inv:energyrun:nospend} Energy is only harvested or transferred; it is never spent.

        \item \label{inv:energyrun:nomove} No amoebot ever moves.
    
        \item \label{inv:energyrun:stabletrees} Any amoebot that belongs to a stable tree of forest $\forest$ (i.e., one that is rooted at a source amoebot) will never change its \parent\ pointer.
    \end{enumerate}
\end{invariant}
\begin{proof}
    We prove each part independently.
    \begin{enumerate}[label=(\alph*),leftmargin=1cm]
        \item The only way for an amoebot to spend energy is during an $\alpha_i^\demand$ execution, which never occurs during an energy run by definition.

        \item The only way for an amoebot to move is during an $\alpha_i^\demand$ execution, which never occurs during an energy run by definition.
    
        \item The \parent\ pointer of an amoebot $A$ is only updated if $A$ contracts or is involved in a handover, calls $\textsc{Prune}(\,)$, or is adopted during \growforest.
        No amoebot moves during an energy run (Invariant~\ref{inv:energyrun:nomove}) and stable trees never prune by definition.
        So members of stable trees remain there throughout an energy run.
        \qedhere
    \end{enumerate}
\end{proof}

\begin{invariant} \label{inv:reachable}
    Any configuration reached by any sequential execution of $\alg^\demand$ starting in $C_0^\demand$:
    \begin{enumerate}[label=(\alph*),ref=\ref{inv:reachable}\alph*,leftmargin=1cm]
        \item \label{inv:reachable:connected} is connected.

        \item \label{inv:reachable:sources} contains at least one source amoebot.

        \item \label{inv:reachable:batteries} maintains $A.\battery \in \{0, 1, \ldots, \capacity\}$ for all amoebots $A$.
    \end{enumerate}
\end{invariant}
\begin{proof}
    We prove each part independently.
    \begin{enumerate}[label=(\alph*),leftmargin=1cm]
        \item The initial configuration $C_0^\demand$ is connected by supposition.
        All amoebot movements in $\alg^\demand$ originate from the movement phases of $\alpha_i$ actions from the original algorithm $\alg$.
        Since $\alg$ satisfies the connectivity convention (Convention~\ref{conv:connect}) by supposition, no configuration reachable from $C_0^\demand$ could ever be disconnected.

        \item The initial configuration $C_0^\demand$ contains at least one source amoebot by supposition.
        By inspection of Algorithm~\ref{alg:framework}, a source amoebot never updates its $\xstate$, so any source amoebot in $C_0^\demand$ remains a source amoebot throughout the execution of $\alg^\demand$.

        \item All amoebot batteries are initially empty in $C_0^\demand$.
        The guards $g_i^\demand$ and predicates $g_\harvestenergy$ and $g_\shareenergy$ ensure that $A.\battery \in [0, \capacity]$.
        Moreover, all changes to $A.\battery$ are integral: the $\alpha_i^\demand$ actions spend $\demand(\alpha_i) \in \{1, 2, \ldots, \capacity\}$ energy, \harvestenergy\ always harvests a single unit of energy into a source amoebot's battery, and \shareenergy\ always transfers a single unit of energy from a parent to one of its children.
        Noting that the battery capacity $\capacity$ is an integer, the invariant follows.
        \qedhere
    \end{enumerate}
\end{proof}

With the invariants in place, we can move on to analyzing sequential executions of $\alg^\demand$ representing any sequence of activations the unfair sequential adversary could have chosen.

\begin{lemma} \label{lem:equivalence}
    Consider any sequential execution $\sched^\demand$ of $\alg^\demand$ starting in initial configuration $C_0^\demand$ and let $\sched_\alpha^\demand$ denote its subsequence of $\alpha_i^\demand$ action executions.
    Then the corresponding sequence $\sched_\alpha$ of $\alpha_i$ executions is a valid sequential execution of $\alg$ starting in initial configuration $C_0$.
\end{lemma}
\begin{proof}
    Let $C_r^\demand$ (resp., $C_r$) denote the configuration reached by the first $r$ action executions in $\sched_\alpha^\demand$ starting in $C_0^\demand$ (resp., in $\sched_\alpha$ starting in $C_0$).
    Argue by induction on $r \geq 0$ that $C_r^\demand \cong C_r$; i.e., these configurations are identical with respect to amoebots' positions and the variables of $\alg$.
    This implies that $\sched_\alpha$ is a valid sequential execution of $\alg$ starting in $C_0$, as desired.

    If $r = 0$, then trivially $C_0^\demand \cong C_0$ by definition (see the statement of Theorem~\ref{thm:main}).
    So suppose $r \geq 1$.
    By the induction hypothesis, $C_{r-1}^\demand \cong C_{r-1}$.
    By definition, there is at most one energy run of $\alpha_\energydist$ executions in $\sched^\demand$ between $C_{r-1}^\demand$ and the configuration $C_{r'}^\demand$ in which the $r$-th $\alpha_i^\demand$ execution of $\sched_\alpha^\demand$ is enabled.
    But $\alpha_\energydist$ executions do not move amoebots or modify any variables of algorithm $\alg$, so $C_{r'}^\demand \cong C_{r-1}^\demand \cong C_{r-1}$.
    Also, any amoebot $A$ for which some $\alpha_i^\demand$ action is enabled must also satisfy the guard $g_i$ of action $\alpha_i$, by definition of the guard $g_i^\demand$.
    Thus, if $A$ executes $\alpha_i^\demand$ in $C_{r'}^\demand$, action $\alpha_i$ can also be executed by $A$ in $C_{r-1}$.
    Moreover, any amoebot movements or updates to variables of $\alg$ must be identical in both action executions, since $\alpha_i^\demand$ emulates $\alpha_i$.
    Therefore, $C_r^\demand \cong C_r$.
\end{proof}

Lemma~\ref{lem:equivalence} gives us a handle on the $\alpha_i^\demand$ action executions in any sequential execution of $\alg^\demand$, so it remains to analyze the energy runs between them.
In this first series of lemmas, we show that if $\alpha_\energydist$ is continuously enabled for some amoebot $A$ during an energy run, then within one additional round either $A$ is activated or the energy run is ended by some $\alpha_i^\demand$ action execution (Lemma~\ref{lem:round}).
Formally, we say an execution of $\alpha_\energydist$ by an amoebot $A$ is \textit{$g$-supported} if predicate $g \in \mathcal{G}$ is satisfied when $A$ is activated and executes $\alpha_\energydist$.
To prove eventual execution, we argue that any predicate $g \in \mathcal{G}$ can support at most a finite number of executions per energy run (Lemmas~\ref{lem:forestblocksfinite} and~\ref{lem:energyblocksfinite}).
Combining this with the definition of a round from Section~\ref{subsec:model} yields the one round upper bound on how long an $\alpha_\energydist$ action can remain continuously enabled in an energy run.

We begin with the \getpruned, \askgrowth, and \growforest\ blocks that maintain the spanning forest $\forest$.
Recall from Section~\ref{subsec:edf} that amoebots may move and disrupt the forest structure.
Thus, at the start of any energy run, the amoebot system is partitioned into \textit{stable trees} rooted at source amoebots, \textit{unstable trees} rooted at \pruning\ amoebots, and \idle\ amoebots that do not belong to any tree.
In the following lemma, we argue that amoebots cannot be trapped in an infinite loop of pruning and rejoining the forest $\forest$.

\begin{lemma} \label{lem:constantrejoin}
    In any energy run of $\sched^\demand$, no amoebot is pruned from and adopted into the forest $\forest$ more than eight times.
\end{lemma}
\begin{proof}
    By Invariant~\ref{inv:energyrun:stabletrees}, any amoebot that was already in a stable tree at the start of the energy run or is adopted into a stable tree during the energy run will remain there throughout the energy run.
    So suppose to the contrary that an amoebot $A$ is pruned from and adopted into unstable trees of the forest $\forest$ more than eight times.
    Since amoebot $A$ can have at most eight neighbors (if it is expanded) and none of these neighbors can move during an energy run (Invariant~\ref{inv:energyrun:nomove}), there must exist a neighbor $B$ that adopts $A$ into an unstable tree more than once.
    By the predicate $g_\growforest$ and the fact that $B$ cannot be a source if it is in an unstable tree, this implies that $B$ must become \growing\ multiple times.

    Observe that when a \growing\ amoebot transfers its \xstate\ to its \asking\ children during a $g_\growforest$-supported execution, it excludes any newly adopted child (which is \xactive) and then becomes \xactive.
    Moreover, because unstable trees are severed from source amoebots, no new \growing\ ancestors can be introduced in an unstable tree.
    Thus, the only amoebots that can become \growing\ in an unstable tree are those that had \growing\ ancestors in this tree at the start of the energy run, but even those will become \growing\ at most once.
    So $B$ cannot become \growing\ multiple times to adopt $A$ more than once, a contradiction.
\end{proof}

We next show that all amoebots eventually join and remain in stable trees.

\begin{lemma} \label{lem:forestblocksfinite}
    Any energy run of $\sched^\demand$ contains at most a finite number of $g_\getpruned$-, $g_\askgrowth$-, and $g_\growforest$-supported executions of $\alpha_\energydist$.
\end{lemma}
\begin{proof}
    The predicates $g_\getpruned$, $g_\askgrowth$, and $g_\growforest$ depend only on the $\xstate$ and $\parent$ variables, neither of which are updated by the \harvestenergy\ and \shareenergy\ blocks.
    Thus, we may consider only the \getpruned, \askgrowth, and \growforest\ blocks when analyzing executions of $\alpha_\energydist$ supported by their predicates.

    Suppose to the contrary that an energy run of $\sched^\demand$ contains an infinite number of $g_\getpruned$-supported executions.
    With only a finite number of amoebots in the system, there must exist an amoebot $A$ that performs an infinite number of $g_\getpruned$-supported executions.
    Then an infinite number of times, $A$ must start as \pruning\ to satisfy $g_\getpruned$ and end as \idle\ after executing $\getpruned$.
    But by Lemma~\ref{lem:constantrejoin}, $A$ can only be pruned from and adopted into the forest a constant number of times in an energy run, a contradiction.

    Suppose instead that an energy run of $\sched^\demand$ contains an infinite number of $g_\askgrowth$-supported executions.
    Again, this implies some amoebot $A$ performs an infinite number of $g_\askgrowth$-supported executions.
    Then an infinite number of times, $A$ must be \xactive\ and have either an \idle\ neighbor or \asking\ child to satisfy $g_\askgrowth$ and then become \asking\ after executing $\askgrowth$.
    One way $A$ can return to \xactive\ from \asking\ is via pruning and later readoption into the forest, but Lemma~\ref{lem:constantrejoin} states that this can only happen a constant number of times per energy run.
    The only alternative is for $A$ to become \growing\ during a $g_\growforest$-supported execution by its parent and later reset itself to \xactive\ during its own $g_\growforest$-supported execution.
    So if $A$ performs an infinite number of $g_\askgrowth$-supported executions in this energy run, it must also perform an infinite number of $g_\growforest$-supported executions, which we address in the following final case.

    Suppose to the contrary that an amoebot $A$ executes an infinite number of $g_\growforest$-supported executions in an energy run of $\sched^\demand$.
    At the start of each of these infinite executions, $A$ must either be \growing\ or be a source with an \idle\ neighbor or \asking\ child.
    If $A$ is \growing, then it becomes \xactive\ after executing $\growforest$.
    The only way for $A$ to become \growing\ again is if its parent performs a $g_\growforest$-supported execution, which in turn is only possible if its grandparent performed an earlier $g_\growforest$-supported execution, and so on all the way up to the source amoebot rooting this tree.
    
    So it suffices to analyze the case when $A$ satisfies $g_\growforest$ as a source.
    Each time $A$ performs a $g_\growforest$-supported execution as a source, it adopts all its \idle\ neighbors into its (stable) tree.
    By Invariant~\ref{inv:energyrun:stabletrees}, these adopted amoebots will remain children of $A$ throughout this energy run.
    Thus, $A$ can perform a $g_\growforest$-supported execution as a source with an \idle\ neighbor only as many times as the number of its \idle\ neighbors, which is at most six if $A$ is contracted and at most eight if $A$ is expanded.
    
    The remaining possibility is that $A$ performs an infinite number of $g_\growforest$-supported executions as a source with an \asking\ child.
    The predicate $g_\askgrowth$ ensures that every asking signal that reaches $A$ originates at an \xactive\ amoebot with an \idle\ neighbor.
    Again, because there are only a finite number of amoebots in the system, an infinite number of asking signals reaching $A$ implies the existence of an amoebot $B$ in the stable tree rooted at $A$ that performs an infinite number of $g_\askgrowth$-supported executions as an \xactive\ amoebot with an \idle\ neighbor.
    Because $B$ is in a stable tree, the only way it can return to \xactive\ from \asking\ is to become \growing\ during a $g_\growforest$-supported execution by its parent and later reset itself to \xactive\ during its own $g_\growforest$-supported execution.
    During its own $g_\growforest$-supported execution, $B$ adopts any \idle\ neighbors it has.
    But it is not guaranteed that $B$ will have an \idle\ neighbor at the time of its $g_\growforest$-supported execution, even though it had one earlier: some neighbor could be \idle\ at the time $B$ performs its $g_\askgrowth$-supported execution, get adopted by a different amoebot by the time $B$ performs its $g_\growforest$-supported execution, and then become \idle\ again via pruning before $B$ performs its next $g_\askgrowth$-supported execution.
    However, $B$ can only ask but fail to adopt an \idle\ neighbor a constant number of times by Lemma~\ref{lem:constantrejoin}.    
    With any adoptee remaining in the stable tree throughout the energy run by Invariant~\ref{inv:energyrun:stabletrees} and at most a constant number of \idle\ neighbors to adopt, $B$ can perform at most a constant total number of $g_\askgrowth$-supported executions before adopting all its \idle\ children, a contradiction.

    Therefore, we conclude that the number of $g_\getpruned$-, $g_\askgrowth$-, and $g_\growforest$-supported executions in any energy run is finite, as desired.
\end{proof}

The next lemma is an analogous result for the \harvestenergy\ and \shareenergy\ blocks that move energy throughout the system. 

\begin{lemma} \label{lem:energyblocksfinite}
    Any energy run of $\sched^\demand$ contains at most a finite number of $g_\harvestenergy$- and $g_\shareenergy$-supported executions of $\alpha_\energydist$.
\end{lemma}
\begin{proof}
    Energy is never spent in an energy run (Invariant~\ref{inv:energyrun:nospend}).
    Thus, since every $g_\harvestenergy$-supported execution harvests a single unit of energy into the system, there can be at most $\numAmoebots\capacity$ such executions before the total harvested energy exceeds the total capacity of all $\numAmoebots$ amoebots' batteries.
    Analogously, since every $g_\shareenergy$-supported execution transfers one unit of energy from some parent amoebot to one of its children in $\forest$, any amoebot with $d$ descendants in $\forest$ can perform at most $d\capacity$ such executions before exceeding the total capacity of its descendants' batteries.
    None of the other blocks (\getpruned, \askgrowth, and \growforest) transfer energy, so once all amoebots' batteries are full, $g_\harvestenergy$ and $g_\shareenergy$ will be continuously dissatisfied for the remainder of the energy run.
\end{proof}

Combining Lemmas~\ref{lem:forestblocksfinite} and~\ref{lem:energyblocksfinite} shows that any energy run is finite.
But more importantly, they show that the unfair adversary exhibits weak fairness in an energy run.
Since the total number of $\alpha_\energydist$ executions in an energy run is finite, the unfair adversary will eventually be forced to activate any continuously enabled amoebot.
We formalize this result in the next lemma, concluding our arguments on energy run termination.

\begin{lemma} \label{lem:round}
    Consider any amoebot $A$ for which $\alpha_\energydist$ is enabled and would remain so until execution in some energy run of $\sched^\demand$.
    Then within one additional round, either $A$ executes $\alpha_\energydist$ or this energy run is ended by some $\alpha_i^\demand$ execution.
\end{lemma}
\begin{proof}
    Suppose $\alpha_\energydist$ is enabled for amoebot $A$ in round $r$.
    If an $\alpha_i^\demand$ execution ends this energy run by the completion of round $r+1$, we are done.
    Otherwise, this energy run extends through the remainder of round $r$ and---if round $r$ is finite---all of round $r+1$.
    
    Suppose to the contrary that $A$ is not activated in the remainder of round $r$ or at any time in round $r+1$.
    Recall from Section~\ref{subsec:model} that a (sequential) round ends once every amoebot that was enabled at its start has either completed an action execution or become disabled.
    By supposition, $A$ will remain enabled until its $\alpha_\energydist$ action is executed.
    So at least one of rounds $r$ and $r+1$ must never complete; i.e., at least one of them contains an infinite sequence of $\alpha_\energydist$ executions by enabled amoebots other than $A$.    
    There are only finitely many amoebots, so there must exist an amoebot $B \neq A$ that performs an infinite number of $\alpha_\energydist$ executions.
    Moreover, there are only five predicates that could support these executions, so there must exist a predicate $g \in \mathcal{G}$ such that $B$ performs an infinite number of $g$-supported executions of $\alpha_\energydist$.
    But Lemmas~\ref{lem:forestblocksfinite} and~\ref{lem:energyblocksfinite} show that any predicate can support at most a finite number of $\alpha_\energydist$ executions per energy run of $\sched^\demand$, a contradiction.
\end{proof}

With Lemma~\ref{lem:round} in place, we now argue about the progress and runtime of energy runs towards their overall goal of distributing energy to deficient amoebots in the system.
This next series of lemmas proves an $\bigo{\numAmoebots^2}$ upper bound on the number of rounds any energy run can take before all $\numAmoebots$ amoebots belong to stable trees (Lemma~\ref{lem:stabletime}).
Of course, an energy run could be ended by an $\alpha_i^\demand$ execution before all amoebots join stable trees, but this only helps our overall progress argument.
In the following lemmas, we prove our upper bound for \textit{uninterrupted energy runs} that continue until $\alpha_\energydist$ is disabled for all amoebots.
We first upper bound the time for any unstable tree to be dissolved by pruning.

\begin{lemma} \label{lem:prunetime}
    In an uninterrupted energy run of $\sched^\demand$, any amoebot $A$ at depth $d$ of an unstable tree $\tree$ will be pruned (i.e., set its children to \pruning, reset their $\parent$ pointers, and become \idle) within at most $d + 1$ rounds.\footnote{The \textit{depth} of a amoebot $A$ in a tree $\tree$ rooted at an amoebot $R$ is the number of nodes in the $(R,A)$-path in $\tree$ (i.e., the root $R$ is at depth $1$, and so on). The depth of a tree $\tree$ is $\max_{A \in \tree}\{\text{depth of } A\}$.}
\end{lemma}
\begin{proof}
    Argue by induction on $d$, the depth of $A$ in $\tree$.
    If $d = 1$, $A$ is the root of the unstable tree $\tree$ and thus must be \pruning\ by definition.
    So $A$ continuously satisfies $g_\getpruned$ since only a \pruning\ amoebot can change its own \xstate.
    By Lemma~\ref{lem:round}, $A$ will be activated and perform a $g_\getpruned$-supported execution within $d = 1$ additional round.
    Now suppose $d > 1$ and that every amoebot at depth at most $d - 1$ in $\tree$ is pruned within $d$ rounds.
    If $A$ is also pruned by round $d$, we are done.
    Otherwise, $A$ has been \pruning\ since at least the end of round $d$ when its parent in $\tree$ performed its own $g_\getpruned$-supported execution.
    So $A$ again continuously satisfies $g_\getpruned$ and must be activated by the end of round $d + 1$ by Lemma~\ref{lem:round}.
    Thus, in all cases, $A$ is pruned in at most $d + 1$ rounds.
\end{proof}

Once all unstable trees are dissolved, the newly \idle\ amoebots need to be adopted into stable trees.
Recall that members of stable trees must become \asking\ and then \growing\ before they can adopt their \idle\ neighbors as \xactive\ children.

\begin{lemma} \label{lem:growtime}
    In an uninterrupted energy run of $\sched^\demand$, any \asking\ amoebot $A$ at depth $d$ of a stable tree $\tree$ will become \growing\ within at most $2d - 2$ rounds.
\end{lemma}
\begin{proof}
    Recall that asking signals are propagated to the source root of a stable tree by \xactive\ parents performing $g_\askgrowth$-supported executions when they have \asking\ children.
    In the worst case, all non-source ancestors of $A$ are \xactive; i.e., no progress has been made towards propagating this asking signal.
    Since $A$ is in a stable tree and thus can't become \pruning, $A$ remains \asking\ until it becomes \growing.
    Thus, the \xactive\ parent of $A$ continuously satisfies $g_\askgrowth$ and will become \asking\ within one additional round by Lemma~\ref{lem:round}.
    Any \xactive\ ancestor of $A$ with an \asking\ child also continuously satisfies $g_\askgrowth$ and thus will become \asking\ within one additional round by Lemma~\ref{lem:round}.
    There are $d - 2$ \xactive\ ancestors strictly between $A$ and the source amoebot rooting this stable tree, so within at most $d - 2$ rounds the source amoebot will have an \asking\ child.
    The source amoebot will continuously satisfy $g_\growforest$ because of its \asking\ child, so it will make all its \asking\ children \growing\ within one additional round by Lemma~\ref{lem:round}.
    Similarly, \growing\ amoebots continuously satisfy $g_\growforest$ and pass their \growing\ \xstate\ to their \asking\ children within one additional round by Lemma~\ref{lem:round}.
    So $A$ must become \growing\ within another $d - 1$ additional rounds, for a total of at most $(d - 2) + 1 + (d - 1) = 2d - 2$ rounds.
\end{proof}

Combining Lemmas~\ref{lem:prunetime} and~\ref{lem:growtime} yields an upper bound on the time an uninterrupted energy run requires to organize all amoebots into stable trees.

\begin{lemma} \label{lem:stabletime}
    After at most $\bigo{\numAmoebots^2}$ rounds of any uninterrupted energy run of $\sched^\demand$, all $\numAmoebots$ amoebots belong to stable trees.
\end{lemma}
\begin{proof}
    If all amoebots already belong to stable trees, we are done.
    So suppose at least one amoebot is \idle\ or in an unstable tree.
    The system always contains at least one source amoebot (Invariant~\ref{inv:reachable:sources}), so the depth of any unstable tree is at most $\numAmoebots - 1$.
    By Lemma~\ref{lem:prunetime}, all members of unstable trees will be pruned and become \idle\ within at most $\numAmoebots$ rounds.
    
    Since the system remains connected (Invariant~\ref{inv:reachable:connected}) and always contains a source amoebot (Invariant~\ref{inv:reachable:sources}), there must exist an \idle\ amoebot $A$ that has at least one neighbor in a stable tree.
    \idle\ amoebots do not execute any actions, so at least one of its \xactive\ neighbors will continuously satisfy $g_\askgrowth$ and become \asking\ within one additional round by Lemma~\ref{lem:round}.
    The depth of any of these \asking\ neighbors of $A$ in their respective stable trees can be at most $\numAmoebots - 1$, counting all amoebots except $A$.
    So by Lemma~\ref{lem:growtime}, at least one of these \asking\ neighbors of $A$ will become \growing\ within at most $2(\numAmoebots - 1) - 2 \leq 2\numAmoebots$ rounds.
    \growing\ amoebots continuously satisfy $g_\growforest$, so within one additional round a \growing\ neighbor of $A$ will attempt to adopt an \idle\ neighbor by Lemma~\ref{lem:round}.
    The first such \growing\ neighbor must succeed in an adoption because $A$ is in its neighborhood.
    
    Thus, at least one \idle\ amoebot is adopted into a stable tree every $\bigo{\numAmoebots}$ rounds.
    There can be at most $\numAmoebots - 1$ amoebots initially outside stable trees, so we conclude that all amoebots are adopted into stable trees within $\numAmoebots + (\numAmoebots - 1) \cdot \bigo{\numAmoebots} = \bigo{\numAmoebots^2}$ rounds.
\end{proof}

Lemma~\ref{lem:stabletime} shows that after at most $\bigo{\numAmoebots^2}$ rounds of any energy run, all amoebots will belong to stable trees.
By Invariant~\ref{inv:energyrun:stabletrees}, they will remain there throughout the energy run; in particular, no amoebot will execute $g_\getpruned$-, $g_\askgrowth$-, or $g_\growforest$-supported executions after this point of the energy run.
For convenience, we refer to these sub-runs as \textit{stabilized energy runs}.
This next series of lemmas proves an $\bigo{n}$ upper bound on the \textit{recharge time}, i.e., the worst case number of rounds any stabilized energy run can take to fully recharge all $\numAmoebots$ amoebots, i.e., $A.\battery = \capacity$ for all amoebots $A$ (Lemma~\ref{lem:rechargetime}).

We make four observations that simplify this analysis, w.l.o.g.
First, we again consider uninterrupted energy runs as it only helps our overall progress argument if some $\alpha_i^\demand$ execution ends an energy run earlier.
Second, we assume all amoebots have initially empty batteries as this can only increase the recharge time.
Third, it suffices to analyze the recharge time of any one stable tree $\tree$ since trees are not reconfigured and do not interact in stabilized energy runs.
Fourth and finally, we show in the following lemma that the recharge time for $\tree$ is at most the recharge time for a simple path of the same number of amoebots.

\begin{lemma} \label{lem:pathrecharge}
    Suppose $\tree$ is a (stable) tree of $k$ amoebots rooted at a source amoebot $A_1$.
    If all amoebots in $\tree$ have initially empty batteries, then the recharge time for $\tree$ is at most the recharge time for a simple path $\mathcal{L} = (A_1, \ldots, A_k)$ in which $A_1$ is a source amoebot, $A_i.\parent = A_{i-1}$ for all $1 < i \leq k$, and all $k$ amoebots have initially empty batteries.
\end{lemma}
\begin{proof}
    Consider any tree $\mathcal{U}$ of $k$ amoebots rooted at a source amoebot $A_1$ and any sequence of amoebot activations $S$ representing an uninterrupted, stabilized energy run in which all amoebots' batteries are initially empty.
    Let $t_S(\mathcal{U})$ denote the number of rounds required to fully recharge all amoebots in $\mathcal{U}$ with respect to $S$ and let $t(\mathcal{U}) = \max_S\{t_S(\mathcal{U})\}$ denote the worst-case recharge time for $\mathcal{U}$.
    With this notation, our goal is to show that $t(\tree) \leq t(\mathcal{L})$.

    The \textit{maximum non-branching path} of a tree $\mathcal{U}$ is the longest directed path $(A_1, \ldots, A_\ell)$ starting at the source amoebot such that $A_{i+1}$ is the only child of $A_i$ in $\mathcal{U}$ for all $1 \leq i < \ell$.
    We argue by (reverse) induction on $\ell$, the length of the maximum non-branching path of $\tree$.
    If $\ell = k$, then $\tree$ and $\mathcal{L}$ are both simple paths of $k$ amoebots with initially empty batteries and thus $t(\tree) = t(\mathcal{L})$.
    So suppose $\ell < k$ and $t(\mathcal{U}) \leq t(\mathcal{L})$ for any tree $\mathcal{U}$ that comprises the same $k$ amoebots as $\tree$ with initially empty batteries, is rooted at amoebot $A_1$, and has at least $\ell + 1$ amoebots in its maximum non-branching path.
    Our goal is to modify the $\parent$ pointers in $\tree$ to form another tree $\tree'$ that has exactly one more amoebot in its maximum non-branching path and satisfies $t(\tree) \leq t(\tree')$.
    Since $\tree'$ has exactly $\ell + 1$ amoebots in its maximum non-branching path, the induction hypothesis implies that $t(\tree) \leq t(\tree') \leq t(\mathcal{L})$.

    We construct $\tree'$ from $\tree$ as follows.
    Let $(A_1, \ldots, A_\ell)$ be a maximum non-branching path of $\tree$, where $A_\ell$ is the ``closest'' amoebot to $A_1$ with multiple children, say $B_1, \ldots, B_c$ for some $c \geq 2$.
    Note that such an $A_\ell$ must exist because $\ell < k$.
    We form $\tree'$ by reassigning $B_i.\parent$ from $A_\ell$ to $B_1$ for each $2 \leq i \leq c$.
    Then $B_1$ is the only child of $A_\ell$ in $\tree'$, and thus $(A_1, \ldots, A_\ell, B_1)$ is the maximum non-branching path of $\tree'$ which has length $\ell + 1$.
    By the induction hypothesis, $t(\tree') \leq t(\mathcal{L})$.
    So it suffices to show that $t(\tree) \leq t(\tree')$.

    Consider any activation sequence $S = (s_1, \ldots, s_f)$ representing an uninterrupted, stabilized energy run where $s_f$ is the first amoebot activation after which all amoebots in $\tree$ have fully recharged batteries.
    Note that Lemma~\ref{lem:energyblocksfinite} implies $S$ has finite length and hence $s_f$ exists.
    We must show that there exists an activation sequence $S'$ such that $t_S(\tree) \leq t_{S'}(\tree')$.
    We construct $S'$ from $S$ so that the flow of energy through $\tree'$ mimics that of $\tree$.
    For each $s_i \in S$, we append a corresponding subsequence of activations $s_i'$ to the end of $S'$ that activates the same amoebot as $s_i$ and possibly some others as well, if needed.

    In almost all cases, $s_i$ is valid and has the same effect in both $\tree$ and $\tree'$, so we simply add $s_i' = (s_i)$ to $S'$.
    However, any activations $s_i$ in which $A_\ell$ passes energy to a child $B_j$, for $2 \leq j \leq c$, cannot be performed directly in $\tree'$ since $B_j$ is a child of $B_1$---not of $A_\ell$---in $\tree'$.
    We instead add a pair of activations $s_i' = (s_i^1, s_i^2)$ to $S'$ that have the effect of passing energy from $A_\ell$ to $B_j$ but use $B_1$ as an intermediary.
    There are two cases.
    If the battery of $B_1$ is not full (i.e., $B_1.\battery < \capacity$) just before $s_i$, then $s_i^1$ is a $g_\shareenergy$-supported execution of $\alpha_\energydist$ by $A$ passing a unit of energy to $B_1$ and $s_i^2$ is a $g_\shareenergy$-supported execution of $\alpha_\energydist$ by $B_1$ passing a unit of energy to $B_j$.
    Otherwise, these executions are reversed: $B_1$ passes a unit of energy to $B_j$ in $s_i^1$ and $A$ passes a unit of energy to $B_1$ in $s_i^2$.
    In any case, these activations are valid as their respective amoebots satisfy $g_\shareenergy$.
    
    Since all amoebots start with empty batteries and no energy is ever spent in an energy run (Invariant~\ref{inv:energyrun:nospend}), this construction of $S'$ ensures all amoebots' battery levels in $\tree$ and $\tree'$ are the same after each $s_i \in S$ and $s_i' \in S'$, respectively, for all $1 \leq i \leq f$.
    Thus, amoebots in $\tree$ and $\tree'$ only finish recharging after $s_f$ and $s_f'$, respectively.
    Each $s_i'$ activates the same amoebot as $s_i$ does and possibly one additional amoebot, so the number of rounds in $S'$ must be at least that in $S$.
    Therefore, we have $t_S(\tree) \leq t_{S'}(\tree')$, and since the choice of $S$ was arbitrary, we have $t(\tree) \leq t(\tree')$, as desired.
\end{proof}

By Lemma~\ref{lem:pathrecharge}, it suffices to analyze the case where $\tree$ is a simple path of $k$ amoebots with initially empty batteries.
To bound the recharge time, we use a \textit{dominance argument} between the sequential setting of stabilized energy runs and a parallel setting that is easier to analyze.
First, we prove that for any stabilized energy run, there exists a parallel version that makes at most as much progress towards recharging the system in the same number of rounds (Lemma~\ref{lem:dominance}).
We then upper bound the recharge time in parallel rounds (Lemma~\ref{lem:paralleltime}).
Combining these results gives an upper bound on the recharge time in sequential rounds.

Let an \textit{energy configuration} $E$ of the path $\mathcal{L} = (A_1, \ldots, A_k)$ encode the battery values of each amoebot $A_i$ as $E(A_i)$.
An \textit{energy schedule} is a sequence of energy configurations $(E_1, \ldots, E_t)$.
Given any sequence of amoebot activations $S$ representing a stabilized energy run, we define a \textit{sequential energy schedule} $(E_1^S, \ldots, E_t^S)$ where $E_r^S$ is the energy configuration of the path $\mathcal{L}$ at the start of sequential round $r$ in $S$.
Our dominance argument compares these schedules to parallel energy schedules, defined below.

\begin{definition} \label{def:parallelschedule}
    A \underline{parallel energy schedule} $(E_1, \ldots, E_t)$ is a schedule such that for all energy configurations $E_r$ and amoebots $A_i$ we have $E_r(A_i) \in [0, \capacity]$ and, for every $1 \leq r < t$, $E_{r+1}$ is reached from $E_r$ using the following for each amoebot $A_i$:
    \begin{itemize}
        \item $E_r(A_1) < \capacity$, so the source amoebot $A_1$ harvests energy from the external source with:
        \[E_{r+1}(A_1) = E_r(A_1) + 1\]
        
        \item $E_r(A_i) \geq 1$ and $E_r(A_{i+1}) < \capacity$, so $A_i$ passes energy to its child $A_{i+1}$ with:
        \[E_{r+1}(A_i) = E_r(A_i) - 1, \quad
        E_{r+1}(A_{i+1}) = E_r(A_{i+1}) + 1\]
    \end{itemize}
    Such a schedule is \underline{greedy} if the above actions are taken in parallel whenever possible.
\end{definition}

For an amoebot $A_i$ in an energy configuration $E$, let $\Delta_E(A_i) = \sum_{j=i}^k E(A_j)$ denote the total amount of energy in the batteries of amoebots $A_i, \ldots, A_k$ in $E$.
For any two battery configurations $E$ and $E'$, we say $E$ \textit{dominates} $E'$---denoted $E \succeq E'$---if and only if $\Delta_E(A_i) \geq \Delta_{E'}(A_i)$ for all amoebots $A_i \in \mathcal{L}$.

\begin{lemma} \label{lem:dominance}
    Given any activation sequence $S$ representing an uninterrupted, stabilized energy run on a simple path $\mathcal{L}$ of $k$ amoebots starting in an energy configuration $E_1^S$ in which all amoebots have empty batteries, there exists a greedy parallel energy schedule $(E_1, \ldots, E_t)$ with $E_1 = E_1^S$ such that $E_r^S \succeq E_r$ for all $1 \leq r \leq t$.
\end{lemma}
\begin{proof}
    The activation sequence $S$ and initial energy configuration $E_1^S$ yield a unique sequential energy schedule $(E_1^S, \ldots, E_t^S)$.
    Construct a corresponding parallel energy schedule $(E_1, \ldots, E_t)$ as follows.
    First, set $E_1 = E_1^S$.
    Then, for $1 < r \leq t$, obtain $E_r$ from $E_{r-1}$ by performing one \textit{parallel round} in which each amoebot greedily performs the actions of Definition~\ref{def:parallelschedule} if possible.
    We will show $E_r^S \succeq E_r$ for all $1 \leq r \leq t$ by induction on $r$.
    
    Since $E_1 = E_1^S$, we trivially have $E_1^S \succeq E_1$.
    So suppose $r \geq 1$ and for all rounds $1 \leq r' \leq r$ we have $E_{r'}^S \succeq E_{r'}$.
    Considering any amoebot $A_i$, we have $\Delta_{E_r^S}(A_i) \geq \Delta_{E_r}(A_i)$ by the induction hypothesis and want to show that $\Delta_{E_{r+1}^S}(A_i) \geq \Delta_{E_{r+1}}(A_i)$.    
    First suppose the inequality from the induction hypothesis is strict---i.e., $\Delta_{E_r^S}(A_i) > \Delta_{E_r}(A_i)$---meaning strictly more energy has been passed into $A_i, \ldots, A_k$ in the sequential setting than in the parallel one by the start of round $r$.
    No energy is spent in an energy run (Invariant~\ref{inv:energyrun:nospend}), so we know $\Delta_{E_{r+1}^S}(A_i) \geq \Delta_{E_r^S}(A_i)$.
    Because all energy transfers pass one unit of energy either from the external energy source to the source amoebot $A_1$ or from a parent $A_i$ to its child $A_{i+1}$, we have that $\Delta_{E_r^S}(A_i) \geq \Delta_{E_r}(A_i) + 1$.
    But by Definition~\ref{def:parallelschedule}, an amoebot can receive at most one unit of energy per parallel round, so we have:
    \[\Delta_{E_{r+1}^S}(A_i) \geq \Delta_{E_r^S}(A_i) \geq \Delta_{E_r}(A_i) + 1 \geq \Delta_{E_{r+1}}(A_i).\]
    
    Thus, it remains to consider when $\Delta_{E_r^S}(A_i) = \Delta_{E_r}(A_i)$, meaning the amount of energy passed into $A_i, \ldots, A_k$ is exactly the same in the sequential and parallel settings by the start of round $r$.
    It suffices to show that if $A_i$ receives an energy unit in parallel round $r$, then it also does so in the sequential round $r$.
    We first prove that if $A_i$ receives an energy unit in parallel round $r$, then there is at least one unit of energy for $A_i$ to receive in sequential round $r$.
    If $A_i$ is the source amoebot, this is trivial: the external source of energy is its infinite supply.
    Otherwise, $i > 1$ and we must show $E_r^S(A_{i-1}) \geq 1$.
    We have $\Delta_{E_r^S}(A_i) = \Delta_{E_r}(A_i)$ by supposition and $\Delta_{E_r^S}(A_{i-1}) \geq \Delta_{E_r}(A_{i-1})$ by the induction hypothesis, so
    \begin{align*}
        E_r^S(A_{i-1}) &= \sum_{j=i-1}^k E_r^S(A_j) - \sum_{j=i}^k E_r^S(A_j) \\
        &= \Delta_{E_r^S}(A_{i-1}) - \Delta_{E_r^S}(A_i) \\
        &\geq \Delta_{E_r}(A_{i-1}) - \Delta_{E_r}(A_i) \\
        &= \sum_{j=i-1}^k E_r(A_j) - \sum_{j=i}^k E_r(A_j) \\
        &= E_r(A_{i-1}) \geq 1,
    \end{align*}
    where the final inequality follows from the fact that we presumed $A_i$ receives one energy unit in parallel round $r$ which must come from its parent $A_{i-1}$ since $A_i$ is not a source amoebot.
    
    Next, we show that if $A_i$ receives an energy unit in parallel round $r$, then $E_r^S(A_i) \leq \capacity - 1$; i.e., $A_i$ has enough room in its battery to receive an energy unit during sequential round $r$.
    By supposition we have $\Delta_{E_r^S}(A_i) = \Delta_{E_r}(A_i)$ and by the induction hypothesis we have $\Delta_{E_r^S}(A_{i+1}) \geq \Delta_{E_r}(A_{i+1})$.
    Combining these facts, we have
    \begin{align*}
        E_r^S(A_i) &= \sum_{j=i}^k E_r^S(A_j) - \sum_{j=i+1}^k E_r^S(A_j) \\
        &= \Delta_{E_r^S}(A_i) - \Delta_{E_r^S}(A_{i+1}) \\
        &\leq \Delta_{E_r}(A_i) - \Delta_{E_r}(A_{i+1}) \\
        &= \sum_{j=i}^k E_r(A_j) - \sum_{j=i+1}^k E_r(A_j) \\
        &= E_r(A_i) \leq \capacity - 1,
    \end{align*}
    where the final inequality follows from the following observation about how energy is transferred in a parallel schedule.
    It is easy to see from Definition~\ref{def:parallelschedule} that if $j > i$, then $E_{r-1}(A_i) \leq E_{r-1}(A_j)$; i.e., an amoebot can only have as much energy as any one of its descendants in a greedy parallel schedule.
    So if $A_i$ is receiving energy, it cannot have a full battery; otherwise, all of its descendants' batteries must also be full, leaving $A_i$ unable to simultaneously transfer energy to make room for the new energy it is receiving.
    Thus, $A_i$ must have capacity for at least one energy unit at the start of sequential round $r$, as desired.

    Thus, we have shown that if $A_i$ receives a unit of energy in parallel round $r$, then (1) either $i = 1$ or $E_r^S(A_{i-1}) \geq 1$, and (2) $E_r^S(A_i) \leq \capacity - 1$, meaning that at the start of sequential round $r$, there is both an energy unit available to pass to $A_i$ and $A_i$ has sufficient capacity to receive it.
    In other words, either $A_i$ is a source and continuously satisfies $g_\harvestenergy$ or its parent $A_{i-1}$ continuously satisfies $g_\shareenergy$.
    Since no energy is spent in an energy run (Invariant~\ref{inv:energyrun:nospend}), additional activations in sequential round $r$ can only increase the amount of energy available to pass to $A_i$ and increase the space available in $A_i.\battery$.
    Thus, by Lemma~\ref{lem:round}, $A_i$ must receive at least one energy unit in sequential round $r$, proving that $\Delta_{E_{r+1}^S}(A_i) \geq \Delta_{E_{r+1}}(A_i)$ in all cases.
    Since the choice of $A_i$ was arbitrary, we have shown $E_{r+1}^S \succeq E_{r+1}$.
\end{proof}

To conclude the dominance argument, we bound the number of parallel rounds needed to recharge a path of $k$ amoebots.
Combined with Lemma~\ref{lem:dominance}, this gives an upper bound on the worst case number of sequential rounds for any stabilized energy run to do the same.

\begin{lemma} \label{lem:paralleltime}
    Let $(E_1, \ldots, E_t)$ be the greedy parallel energy schedule on a simple path $\mathcal{L}$ of $k$ amoebots where $E_1(A_i) = 0$ and $E_t(A_i) = \capacity$ for all amoebots $A_i \in \mathcal{L}$.
    Then $t = k\capacity = \bigo{k}$.
\end{lemma}
\begin{proof}
    Argue by induction on $k$, the number of amoebots in path $\mathcal{L}$.
    If $k = 1$, then $A_1 = A_k$ is the source amoebot that harvests one unit of energy per parallel round from the external energy source by Definition~\ref{def:parallelschedule}.
    Since $A_1$ has no children to which it may pass energy, it is easy to see that it will harvest $\capacity$ energy in exactly $\capacity  = \Theta(1)$ parallel rounds.
    
    Now suppose $k > 1$ and that any path of $j \in \{1, \ldots, k-1\}$ amoebots fully recharges in $j\capacity$ parallel rounds.
    Once an amoebot $A_i$ has received energy for the first time, it follows from Definition~\ref{def:parallelschedule} that $A_i$ will receive a unit of energy from $A_{i-1}$ (or the external energy source, in the case that $i = 1$) in every subsequent parallel round until $A_i.\battery = \capacity$.
    Similarly, Definition~\ref{def:parallelschedule} ensures that $A_i$ will pass a unit of energy to $A_{i+1}$ in every subsequent parallel round until $A_{i+1}.\battery = \capacity$.
    Thus, once $A_i$ receives energy for the first time, $A_i$ effectively acts as an external energy source for the remaining amoebots $A_{i+1}, \ldots, A_k$.
    
    The source amoebot $A_1$ first harvests energy from the external energy source in parallel round $1$ and thus acts as a continuous energy source for $A_2, \ldots, A_k$ in all subsequent rounds.
    By the induction hypothesis, we know $A_2, \ldots, A_k$ will fully recharge in $(k-1)\capacity$ parallel rounds, after which $A_1$ will no longer pass energy to $A_2$.
    The source amoebot $A_1$ harvests one energy unit from the external energy source per parallel round and already has $A_1.\battery = 1$, so in an additional $\capacity - 1$ parallel rounds we have $A_1.\battery = \capacity$.
    Therefore, the path $A_1, \ldots, A_k$ fully recharges in $1 + (k-1)\capacity + \capacity - 1 = k\capacity = \bigo{k}$ parallel rounds, as required.
\end{proof}

Combining the lemmas of this section yields the following bound on the recharge time.

\begin{lemma} \label{lem:rechargetime}
    After at most $\bigo{\numAmoebots}$ rounds of any uninterrupted, stabilized energy run of $\sched^\demand$, all $\numAmoebots$ amoebots have full batteries.
\end{lemma}
\begin{proof}
    Consider any stabilized energy run of $\sched^\demand$.
    By definition, this energy run starts in a configuration where all amoebots belong to stable trees, and by Invariant~\ref{inv:energyrun:stabletrees} the structure of $\forest$ will not change throughout this energy run.
    So consider any (stable) tree $\tree \in \forest$ and suppose, in the worst-case, that all amoebots have initially empty batteries.
    By Lemma~\ref{lem:pathrecharge}, the recharge time for $\tree$ is at most the recharge time for a path $\mathcal{L}$ of $|\tree|$ amoebots.
    Any activation sequence representing a recharge process for $\mathcal{L}$ runs at least as fast as a greedy parallel energy schedule for $\mathcal{L}$ (Lemma~\ref{lem:dominance}), and the latter must fully recharge $\mathcal{L}$ in $\bigo{|\mathcal{L}|} = \bigo{|\tree|}$ rounds (Lemma~\ref{lem:paralleltime}).
    Since $\tree$ contains at most $\numAmoebots$ amoebots, the lemma follows.
\end{proof}

We can now prove Theorem~\ref{thm:main}, concluding our analysis.

\begin{proof}[Proof of Theorem~\ref{thm:main}]
    As in the statement of Theorem~\ref{thm:main}, consider any energy-compatible amoebot algorithm $\alg$ and demand function $\demand : \alg \to \{1, 2, \ldots, \capacity\}$, and let $\alg^\demand$ be the algorithm produced from $\alg$ and $\demand$ by the energy distribution framework.
    Let $C_0$ be any (legal) connected initial configuration for $\alg$ and let $C_0^\demand$ be its extension for $\alg^\demand$ that designates at least one source amoebot and adds the energy distribution variables with their initial values (Table~\ref{tab:frameworkvariables}) to all amoebots.
    Finally, consider any sequential execution $\sched^\demand$ of $\alg^\demand$ starting in $C_0^\demand$.
    Let $\sched^\demand_\alpha$ be its subsequence of $\alpha_i^\demand$ action executions and $\sched_\alpha$ be the corresponding sequence of $\alpha_i$ action executions.
    By Lemma~\ref{lem:equivalence}, $\sched_\alpha$ is a valid sequential execution of the original algorithm $\alg$.
    Since $\alg$ is assumed to be energy-compatible, its sequential executions always terminate.
    Thus, $\sched_\alpha$ is finite and, by extension, so is $\sched^\demand_\alpha$.
    This implies that the overall execution $\sched^\demand$ contains at most a finite number of distinct energy runs.
    Each of these energy runs is finite by Lemmas~\ref{lem:forestblocksfinite} and~\ref{lem:energyblocksfinite}, so we conclude that $\sched^\demand$ in total is finite.
    
    Let $C^\demand$ be the terminating configuration of $\sched^\demand$, but suppose to the contrary that there does not exist a sequential execution of $\alg$ starting in $C_0$ that terminates in the configuration $C$ obtained from $C^\demand$ by removing the energy distribution variables.
    We have already shown that $\sched_\alpha$ is a valid sequential execution of $\alg$ starting in $C_0$.
    Moreover, $\alg^\demand$ only moves amoebots and modifies variables of algorithm $\alg$ during $\alpha_i^\demand$ executions, so all amoebot movements and updates to variables of algorithm $\alg$ are identical in $\sched_\alpha$ and $\sched^\demand$.
    Thus, $\sched_\alpha$ must reach configuration $C$ but---for the sake of contradiction---cannot terminate there; i.e., there must exist an amoebot $A$ for which some action $\alpha_i$ is enabled in $C$ but all amoebots are disabled in $C^\demand$; in particular, the corresponding action $\alpha_i^\demand$ is disabled for $A$ in $C^\demand$.

    The guard $g_i^\demand$ of action $\alpha_i^\demand$ requires three properties: $A$ satisfies guard $g_i$ of action $\alpha_i$, $A$ and its neighbors are not \idle\ or \pruning, and $A$ has at least $\demand(\alpha_i)$ energy.
    We know $A$ satisfies $g_i$ in $C^\demand$ because $\alpha_i$ is enabled for $A$ in $C$.
    No amoebot in $C^\demand$ can be \idle, since the connectivity of $C^\demand$ (Invariant~\ref{inv:reachable:connected}) implies that some amoebot would satisfy $g_\askgrowth$ or $g_\growforest$ and thus be enabled by $\alpha_\energydist$, contradicting $C^\demand$ as a terminating configuration.
    Similarly, no amoebot can be \pruning\ in $C^\demand$ since this amoebot would satisfy $g_\getpruned$.
    So suppose that in $C^\demand$, $A.\battery < \demand(\alpha_i) \leq \capacity$.
    Then $A$ cannot be a source, since it would satisfy $g_\harvestenergy$.
    So $A$ must be \xactive, \asking, or \growing, all of which imply $A$ has a parent in forest $\forest$.
    The connectivity of $C^\demand$ (Invariant~\ref{inv:reachable:connected}) implies that some ancestor of $A$ satisfies $g_\harvestenergy$ or $g_\shareenergy$: either the parent of $A$ satisfies $g_\shareenergy$, or the parent of $A$ has insufficient energy to share but the grandparent of $A$ satisfies $g_\shareenergy$, and so on up to the source root of the tree which, if it does not have sufficient energy to share, must satisfy $g_\harvestenergy$.
    Therefore, we reach a contradiction in all cases, proving that if $C^\demand$ is a terminating configuration for $\sched^\demand$, then $C$ is a terminating configuration for $\sched_\alpha$ and thus there exists a sequential execution of $\alg$ starting in $C_0$ that terminates in $C$.

    We conclude by proving the runtime overhead bound.
    Let $\algruntime$ be the maximum number of action executions in any sequential execution of $\alg$ on $\numAmoebots$ amoebots.
    We know $\algruntime$ is finite because $\alg$ is energy-compatible.
    By Lemma~\ref{lem:equivalence}, any sequential execution of $\alg^\demand$ contains at most $\algruntime + 1$ energy runs, and each energy run terminates in at most $\bigo{\numAmoebots^2}$ rounds by Lemmas~\ref{lem:stabletime} and~\ref{lem:rechargetime}.
    Therefore, we conclude that any sequential execution of $\alg^\demand$ terminates in at most $\bigo{\numAmoebots^2} \cdot (\algruntime + 1) = \bigo{\numAmoebots^2\algruntime}$ rounds.
\end{proof}
\else
Due to space constraints, we highlight only the most important supporting results of this analysis.
All omitted lemmas, invariants, and proofs can be found in Appendix~\ref{app:edfproofs}.

\begin{lemma}\label{lem:equivalence}
    Consider any sequential execution $\sched^\demand$ of $\alg^\demand$ starting in initial configuration $C_0^\demand$ and let $\sched_\alpha^\demand$ denote its subsequence of $\alpha_i^\demand$ action executions.
    Then the corresponding sequence $\sched_\alpha$ of $\alpha_i$ executions is a valid sequential execution of $\alg$ starting in initial configuration $C_0$.
\end{lemma}

This lemma implies that any sequential execution $\sched^\demand$ of $\alg^\demand$ contains a finite number of $\alpha_i^\demand$ executions, since the corresponding sequence of $\alpha_i$ executions forms a possible sequential execution of $\alg$, which must terminate because $\alg$ is energy-compatible.
It remains to analyze the \textit{energy runs} in $\sched^\demand$, i.e., the maximal sequences of consecutive $\alpha_\energydist$ executions delineated by $\alpha_i^\demand$ executions.
Formally, an execution of $\alpha_\energydist$ by an amoebot $A$ is \textit{$g$-supported} if predicate $g \in \mathcal{G}$ is satisfied when $A$ is activated and executes $\alpha_\energydist$.
We argue that any predicate $g \in \mathcal{G}$ can support at most a finite number of executions per energy run, implying that all energy runs, and thus all sequential executions of $\alg^\demand$, are finite:

\begin{lemma} \label{lem:energyrunfinite}
    Any energy run of $\sched^\demand$ contains at most a finite number of $g$-supported $\alpha_\energydist$ executions, for any $g \in \mathcal{G}$.
\end{lemma}

Let $C^\demand$ be the terminating configuration of $\sched^\demand$.
We must show that there exists a sequential execution of $\alg$ starting in $C_0$ that terminates in the configuration $C$ obtained from $C^\demand$ by removing the energy distribution variables.
An obvious candidate is the sequence $\sched_\alpha$ of $\alpha_i$ executions corresponding to the $\alpha_i^\demand$ executions in $\sched^\demand$.
Lemma~\ref{lem:equivalence} already implies that $\sched_\alpha$ reaches $C$, and a careful argument involving the guard of $\alpha_\energydist$ shows that it must also terminate there.
The remainder of the analysis characterizes the time required for an \textit{uninterrupted} energy run---i.e., one that is not ended early by an $\alpha_i^\demand$ execution, which only helps the overall progress argument---to collect all amoebots into stable trees rooted at source amoebots and, once this is achieved, to fully recharge all amoebots' batteries.

\begin{lemma} \label{lem:stabletime}
    After at most $\bigo{\numAmoebots^2}$ rounds of any uninterrupted energy run of $\sched^\demand$, all $\numAmoebots$ amoebots belong to stable trees.
\end{lemma}
\begin{lemma} \label{lem:rechargetime}
    After at most $\bigo{\numAmoebots}$ rounds of any uninterrupted, stabilized energy run of $\sched^\demand$, all $\numAmoebots$ amoebots have full batteries.
\end{lemma}

These lemmas imply that every energy run terminates in at most $\bigo{\numAmoebots^2}$ rounds.
The theorem supposes that any sequential execution of $\alg$ terminates in $\algruntime$ action executions, so we know by Lemma~\ref{lem:equivalence} that any sequential execution of $\alg^\demand$ contains at most $\algruntime + 1$ energy runs.
Combining these facts yields the $\bigo{\numAmoebots^2\algruntime}$ runtime bound for $\alg^\demand$.
\fi

\section{Energy-Constrained Leader Election and Shape Formation}  \label{sec:edfcompatible}

With the energy distribution framework defined and its properties analyzed, we now apply it to existing energy-agnostic algorithms for leader election and shape formation and show simulations of their energy-constrained counterparts.
We first make a straightforward observation about \textit{stationary} amoebot algorithms, i.e., those in which amoebots do not move.
These include simple primitives like spanning forest formation~\cite{Daymude2019-computingprogrammable} and binary counters~\cite{Porter2018-collaborativecomputation,Daymude2020-convexhull} as well as the majority of existing algorithms for leader election~\cite{Derakhshandeh2015-leaderelection,Daymude2017-improvedleader,Bazzi2019-stationarydeterministic,Gastineau2019-distributedleader,DiLuna2020-shapeformation,Gastineau2022-leaderelection,Briones2023-invitedpaper}.
It is easily seen that an algorithm that never moves cannot disconnect an initially connected system, and its actions never involve a ``move phase''.
Thus,

\begin{observation}  \label{obs:stationary}
    All stationary amoebot algorithms satisfy Convention~\ref{conv:connect}, and those that do not use \Lock\ or \Unlock\ operations also satisfy Convention~\ref{conv:phases}.
\end{observation}

Observation~\ref{obs:stationary} immediately implies the following about stationary algorithms' compatibility with the energy distribution framework.

\begin{corollary}  \label{cor:stationary}
    Any stationary amoebot algorithm that terminates under every (unfair) sequential execution, comprises only valid actions (i.e., those whose executions always succeed in isolation), and does not use \Lock\ or \Unlock\ operations is energy-compatible.
\end{corollary}
 
One such algorithm is \erosionAlg, a deterministic leader election algorithm for hole-free, connected amoebot systems introduced by Di Luna et al.~\cite{DiLuna2020-shapeformation} and extended to the canonical amoebot model and three-dimensional space by Briones et al.~\cite{Briones2023-invitedpaper}.
All amoebots first become leader candidates.
When activated, a candidate uses certain rules regarding the number and relative positions of its neighbors to decide whether to ``erode'', revoking its candidacy without disconnecting or introducing a hole into the remaining set of candidates.
The last remaining candidate is necessarily unique and thus declares itself the leader.

\begin{lemma} \label{lem:leaderelection}
    \erosionAlg\ is energy-compatible.
\end{lemma}
\begin{proof}
    \erosionAlg\ is clearly stationary---no movement is involved in checking neighbors' positions or revoking candidacy---so it suffices to check the conditions of Corollary~\ref{cor:stationary}.    
    Briones et al.~\cite{Briones2023-invitedpaper} have already shown that any unfair sequential execution of this algorithm elects a leader---and thus terminates---in $\bigo{n}$ rounds.
    This correctness analysis also confirms that no actions of \erosionAlg\ are invalid; otherwise, some action executions would fail.
    Finally, it is easy to verify from the algorithm's pseudocode in~\cite{Briones2023-invitedpaper} that \Lock\ and \Unlock\ are not used, so we are done.
\end{proof}

Combining this lemma, the energy distribution framework's guarantees (Theorem~\ref{thm:main}), and \erosionAlg's correctness and runtime guarantees (Theorem~6.3 of~\cite{Briones2023-invitedpaper}) immediately implies the following theorem.

\begin{theorem} \label{thm:leaderelection}
    For any demand function $\demand : \erosionAlg \to \{1, 2, \ldots, \capacity\}$, the algorithm $\erosionAlg^\demand$ produced by the energy distribution framework deterministically solves the leader election problem for hole-free, connected systems of $\numAmoebots$ amoebots in $\bigo{\numAmoebots^3}$ rounds assuming geometric space, assorted orientations, constant-size memory, and an unfair sequential adversary.
\end{theorem}

\begin{figure}[t]
    \centering
    \begin{subfigure}{.24\textwidth}
        \centering
        \includegraphics[width=\textwidth,trim={0 2cm 0 2cm},clip]{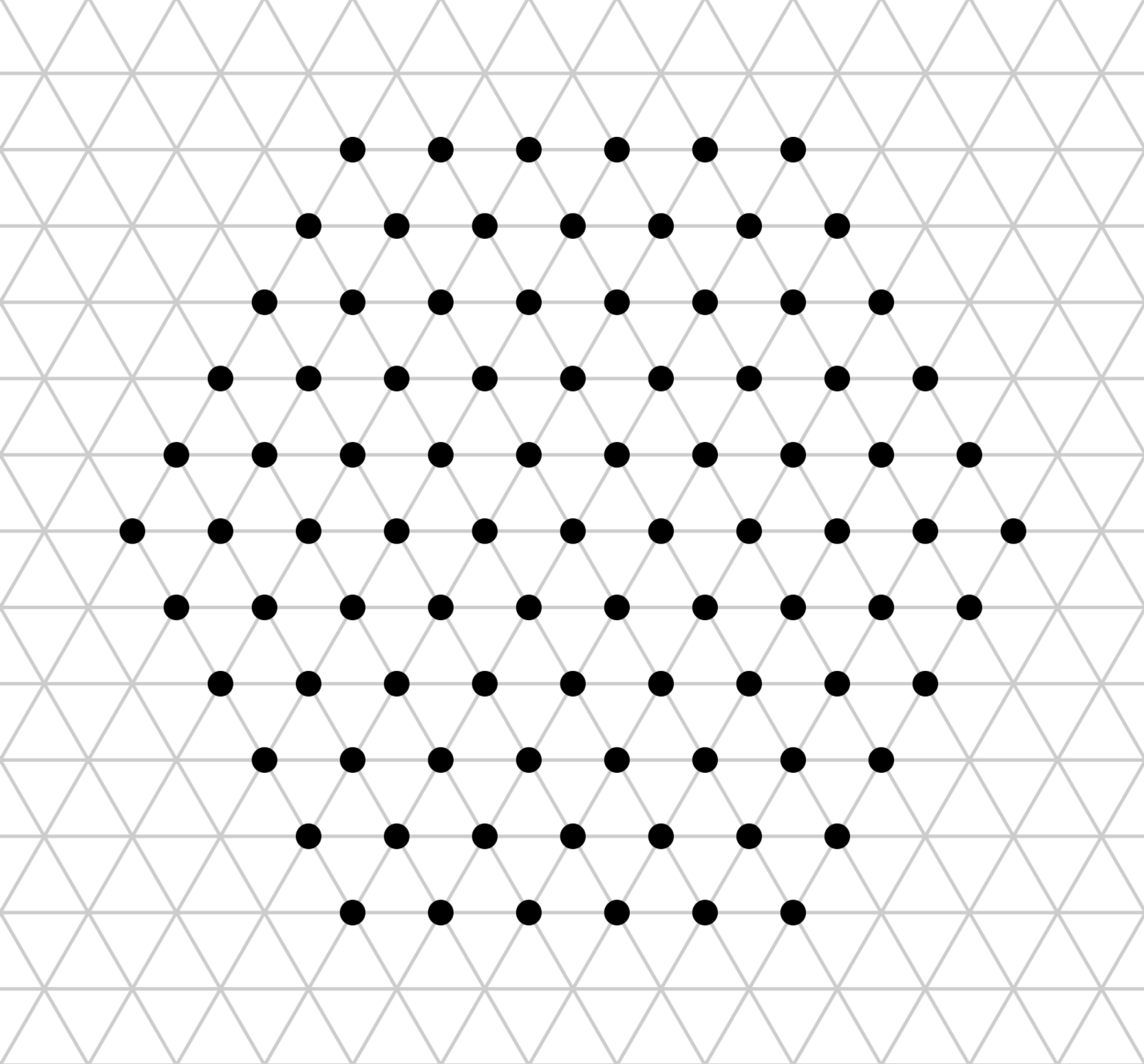}
    \end{subfigure}
    \hfill
    \begin{subfigure}{.24\textwidth}
        \centering
        \includegraphics[width=\textwidth,trim={0 2cm 0 2cm},clip]{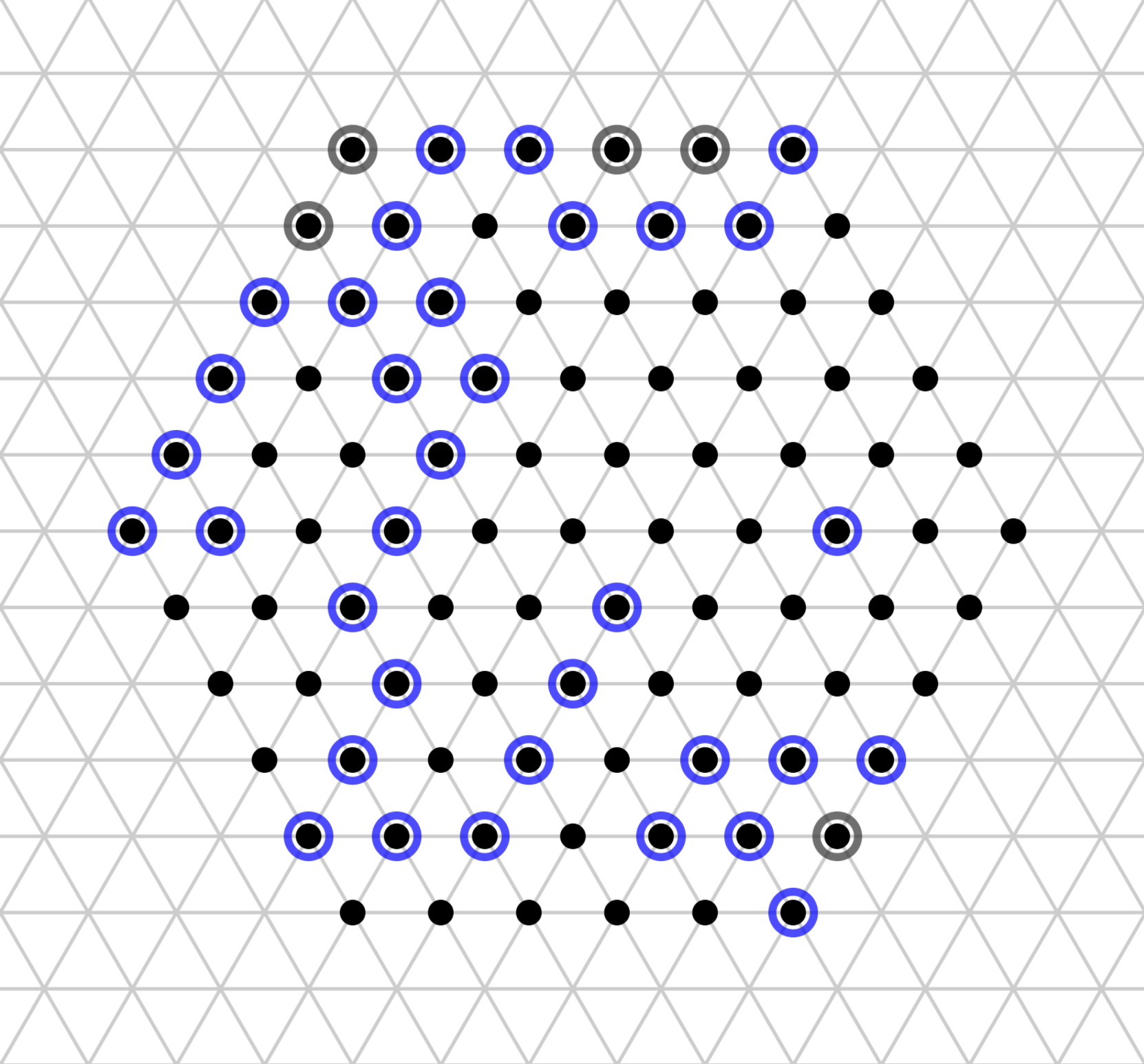}
    \end{subfigure}
    \hfill
    \begin{subfigure}{.24\textwidth}
        \centering
        \includegraphics[width=\textwidth,trim={0 2cm 0 2cm},clip]{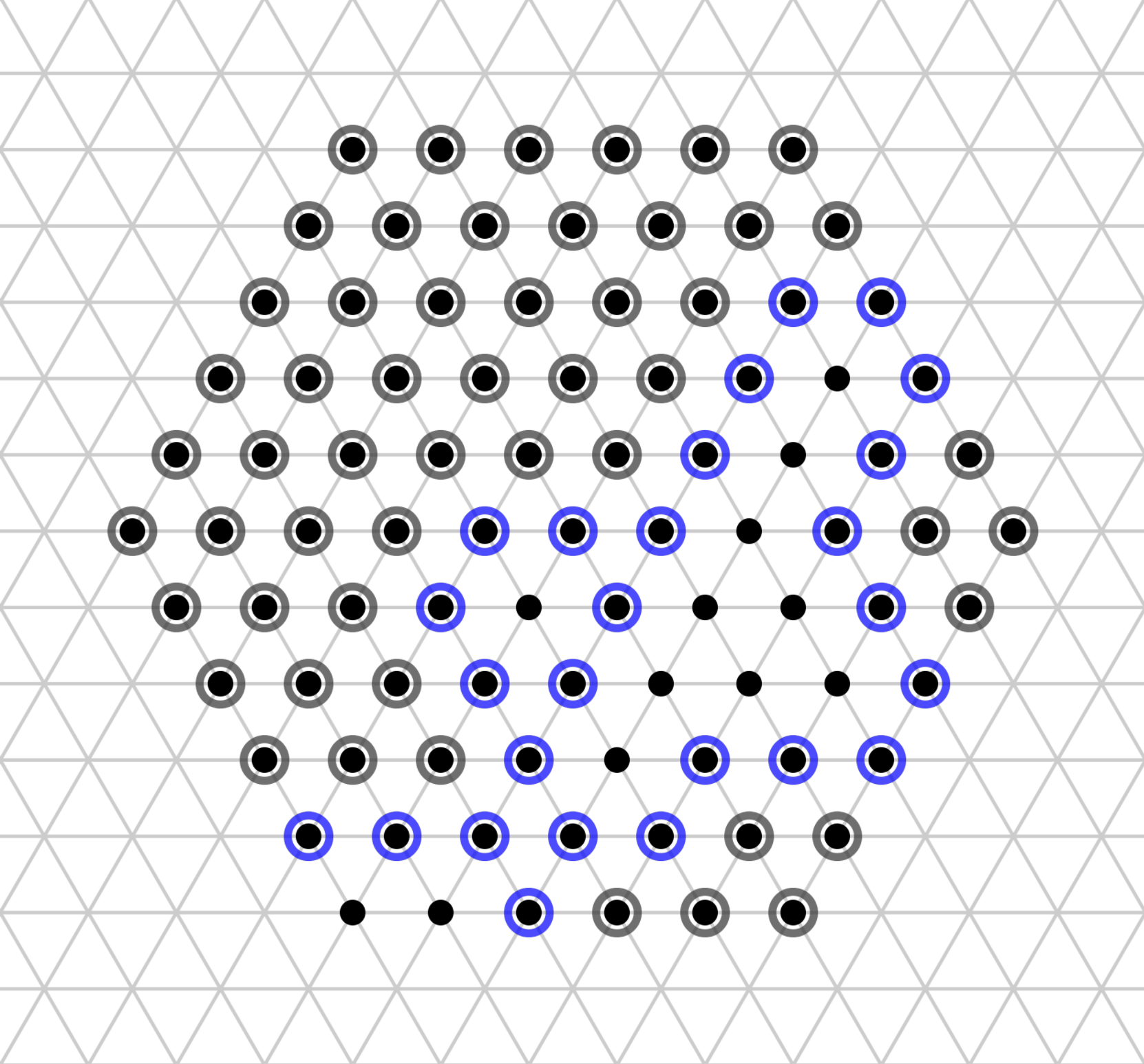}
    \end{subfigure}
    \begin{subfigure}{.24\textwidth}
        \centering
        \includegraphics[width=\textwidth,trim={0 2cm 0 2cm},clip]{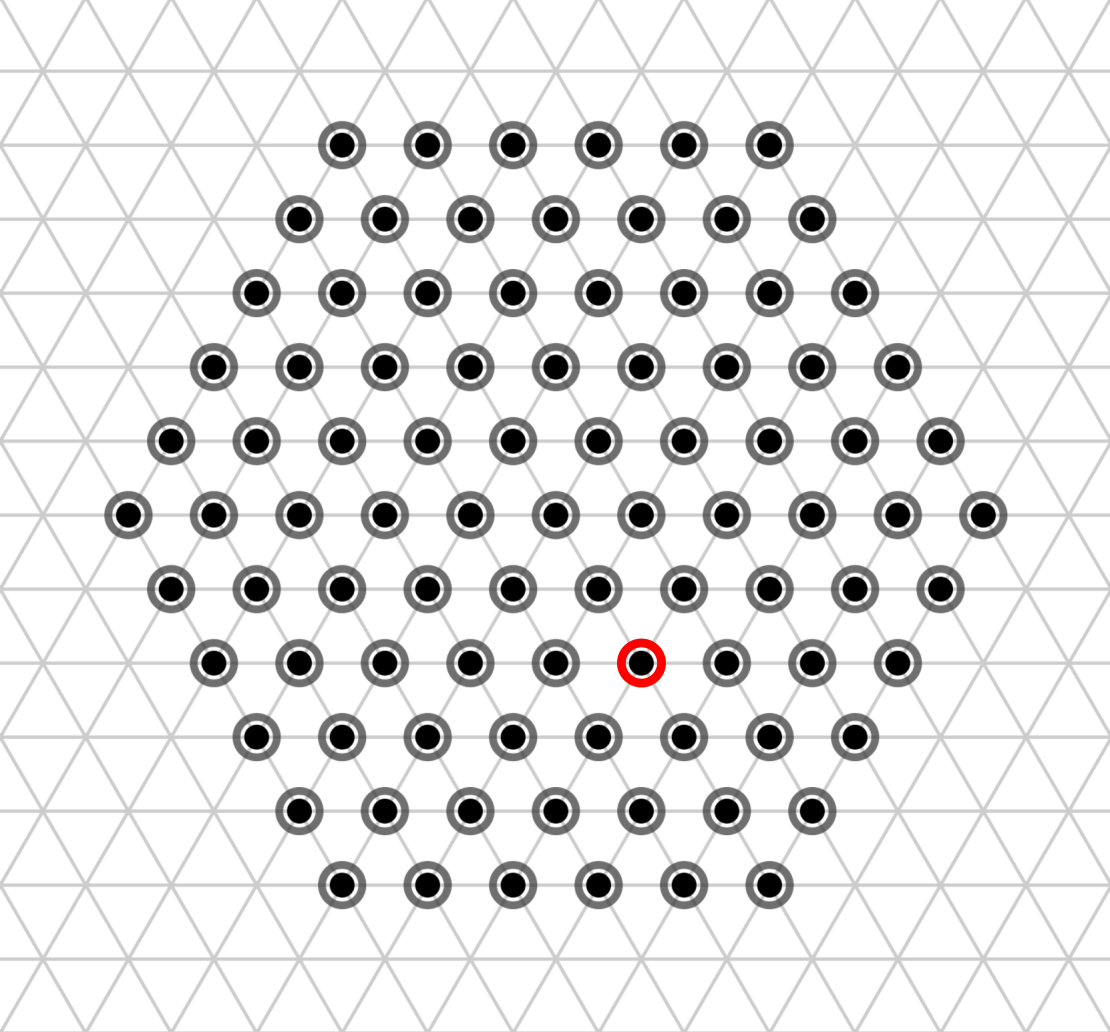}
    \end{subfigure}\\
    \medskip
    \begin{subfigure}{.24\textwidth}
        \centering
        \includegraphics[width=\textwidth,trim={0 2cm 0 2cm},clip]{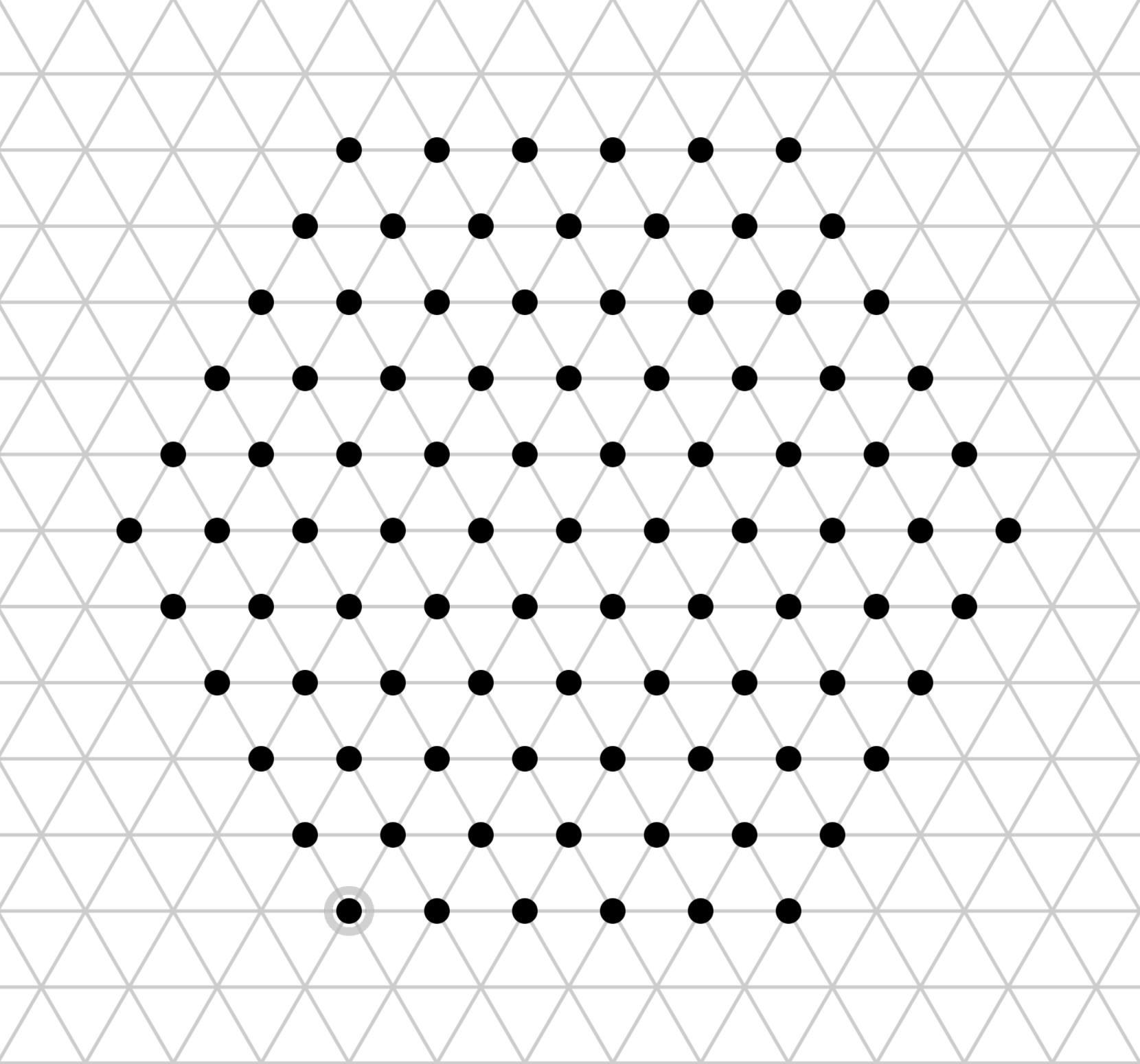}
        \caption{\centering $t = 0$ rounds}
        \label{fig:leadersim:a}
    \end{subfigure}
    \hfill
    \begin{subfigure}{.24\textwidth}
        \centering
        \includegraphics[width=\textwidth,trim={0 2cm 0 2cm},clip]{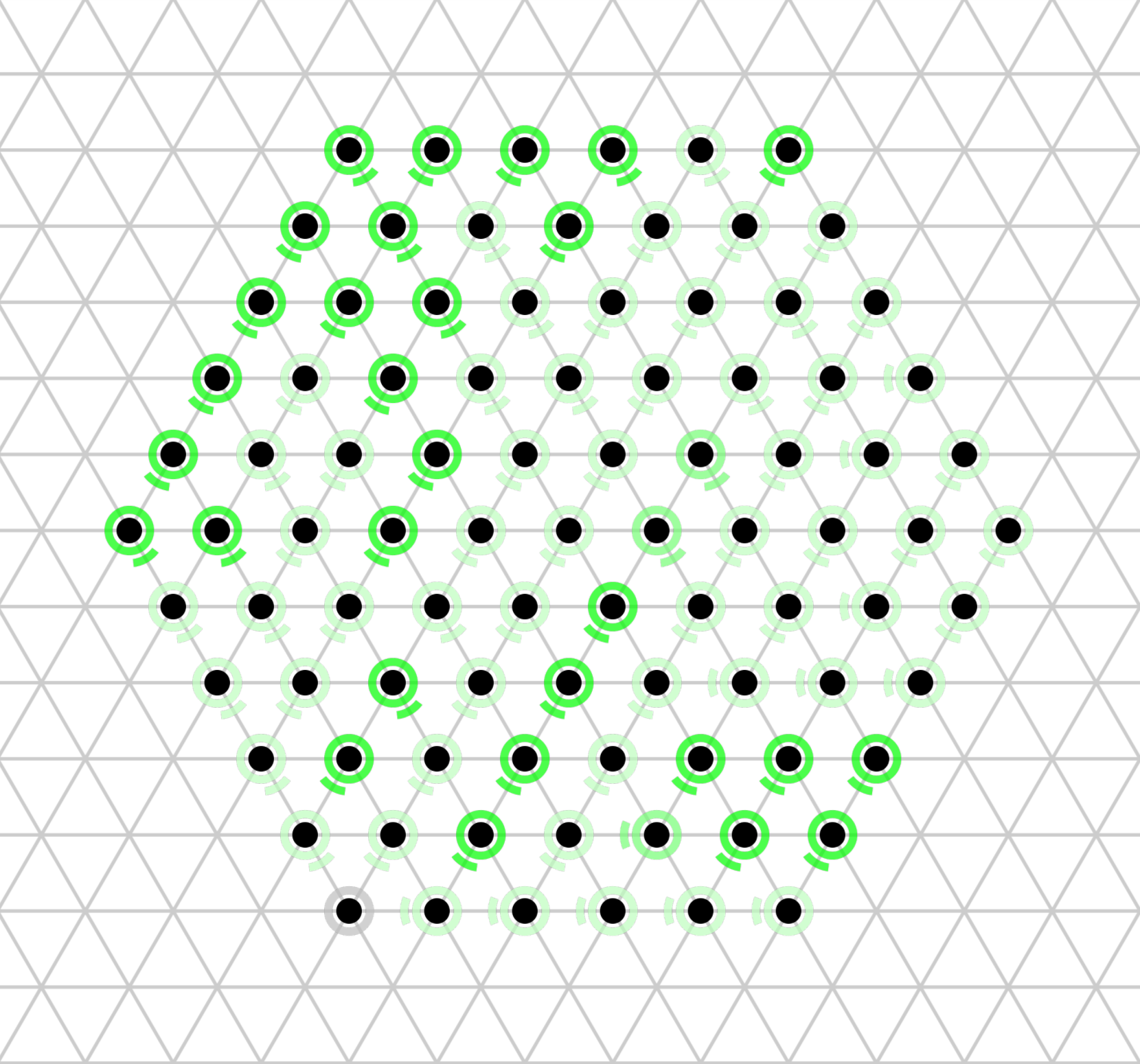}
        \caption{\centering $t = 100$}
        \label{fig:leadersim:b}
    \end{subfigure}
    \hfill
    \begin{subfigure}{.24\textwidth}
        \centering
        \includegraphics[width=\textwidth,trim={0 2cm 0 2cm},clip]{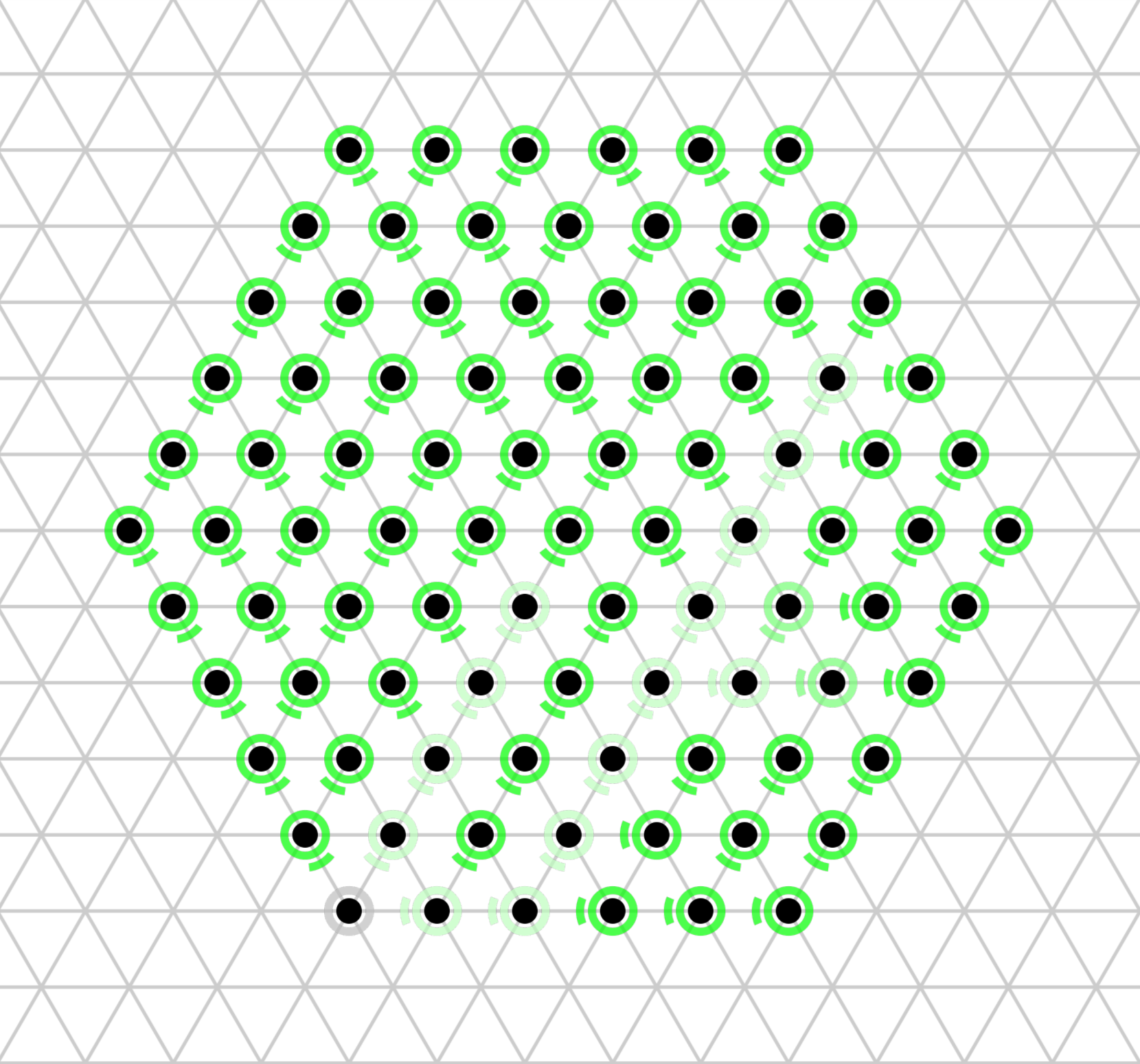}
        \caption{\centering $t = 250$}
        \label{fig:leadersim:c}
    \end{subfigure}
    \hfill
    \begin{subfigure}{.24\textwidth}
        \centering
        \includegraphics[width=\textwidth,trim={0 2cm 0 2cm},clip]{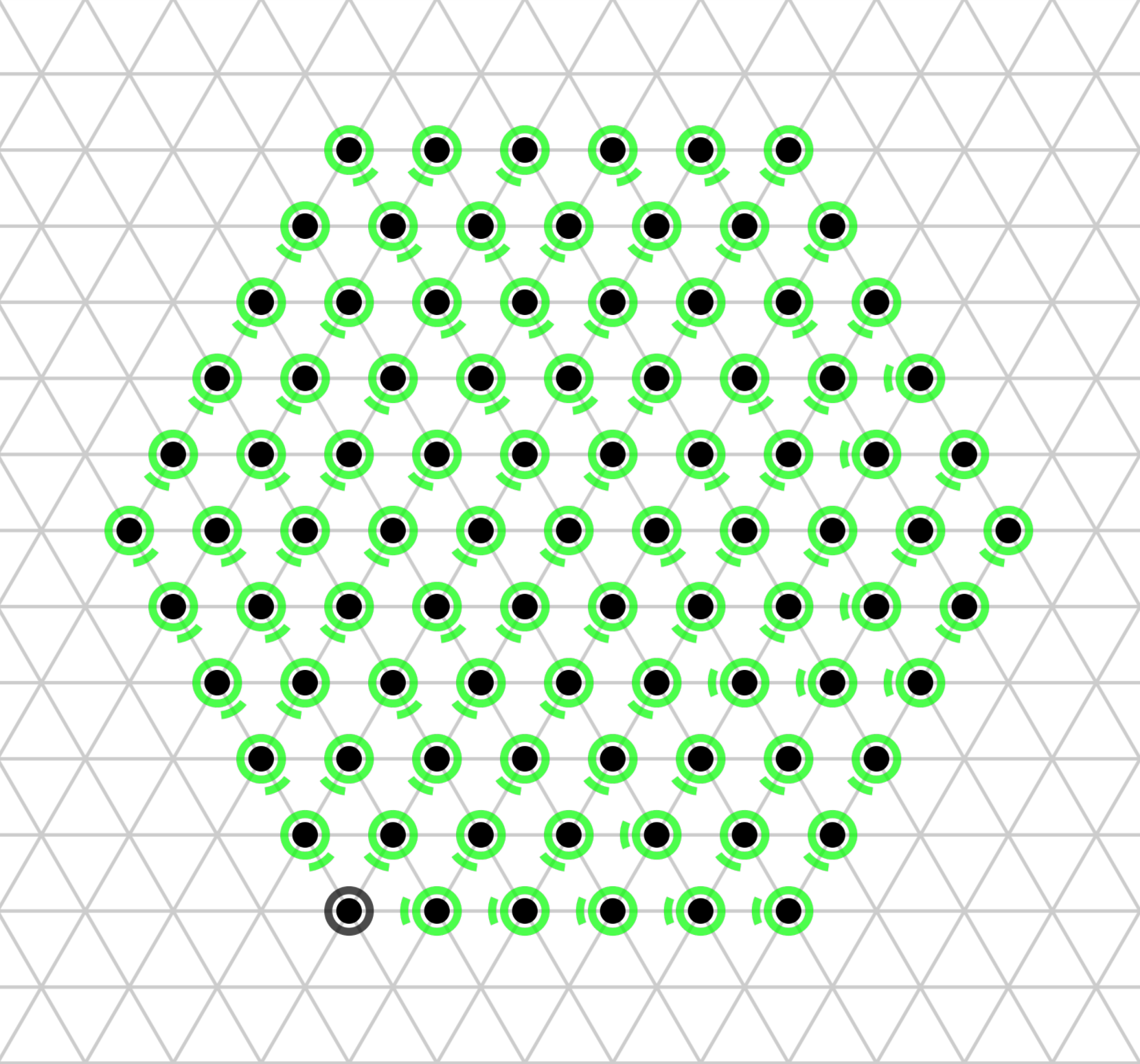}
        \caption{\centering $t = 350$}
        \label{fig:leadersim:d}
    \end{subfigure}
    \caption{\textit{Simulating $\erosionAlg^\demand$.}
    A simulation of $\erosionAlg^\demand$ on $\numAmoebots = 91$ amoebots with one source amoebot, capacity $\capacity = 10$, and demand $\demand(\alpha) = 5$ for all actions $\alpha$.
    Both rows show the same simulation.
    Top: For \erosionAlg, amoebots are initially ``null candidates'' (no color) and eventually declare candidacy (blue); candidates then either erode (dark gray) or become the unique leader (red).
    Bottom: For energy distribution, color opacity indicates energy levels.
    All amoebots are initially \idle\ (no color) except the source (gray/black); amoebots eventually join the forest $\forest$ (green) and distribute energy.}
    \label{fig:leadersim}
\end{figure}

A simulation of $\erosionAlg^\demand$ successfully electing a unique leader under energy constraints is shown in \figtext~\ref{fig:leadersim}.
As the proof of Lemma~\ref{lem:leaderelection} shows, Corollary~\ref{cor:stationary} sets a very low bar for proving stationary algorithms are energy-compatible.
Almost all existing amoebot algorithms are designed to terminate after achieving a desired system behavior, and this property is typically proven as part of their correctness analyses.
Invalid actions are avoided, as their executions would always fail.\footnote{The canonical amoebot model introduced error handling for amoebot algorithm design to deal with operation executions that fail due to concurrency (see Section 2.2 of~\cite{Daymude2023-canonicalamoebot}).
Although error handling could be used to deal with failed executions of invalid actions, no existing amoebot algorithms have taken such a convoluted approach to designing functional algorithms.}
Finally, no existing algorithms use the concurrency control operations \Lock\ and \Unlock\ directly; these are typically reserved for use by the ``concurrency control framework''~\cite{Daymude2023-canonicalamoebot} discussed in the next section.
The only remaining obstacle is that many existing stationary algorithms predate the canonical amoebot model and have not yet been reformulated in guarded action semantics or analyzed under an unfair adversary.
Supposing this obstacle can be overcome without significantly affecting the algorithms' previously proven guarantees, the above discussion shows it is likely that most---if not all---existing stationary amoebot algorithms are energy-compatible.

What about non-stationary amoebot algorithms whose movements make satisfying the phase structure and connectivity conventions (Conventions~\ref{conv:phases} and~\ref{conv:connect}) non-trivial?
Here our example is the \hexagonAlg\ algorithm for basic shape formation, originally introduced by Derakhshandeh et al.~\cite{Derakhshandeh2015-algorithmicframework} and carefully reformulated and analyzed under the canonical amoebot model by Daymude et al.~\cite{Daymude2023-canonicalamoebot}.
The basic idea of this algorithm is to form a hexagon---or as close to one as is possible with the number of amoebots in the system---by extending a spiral that begins at a (pre-defined or elected) seed amoebot.
Thanks to the analysis in~\cite{Daymude2023-canonicalamoebot}, it is easy to show \hexagonAlg\ is compatible with the energy distribution framework.

\begin{lemma} \label{lem:hexagonformation}
    \hexagonAlg\ is energy-compatible.
\end{lemma}
\begin{proof}
    Every sequential execution of \hexagonAlg\ must terminate since Lemma~7 of~\cite{Daymude2023-canonicalamoebot} guarantees that any execution of this algorithm---sequential or concurrent---terminates with the amoebot system forming a hexagon.
    Theorem~10 of~\cite{Daymude2023-canonicalamoebot} guarantees that \hexagonAlg\ satisfies the validity and phase structure conventions (Conventions~\ref{conv:valid} and~\ref{conv:phases}), as these were the two conventions borrowed directly from that paper's concurrency control framework.
    Finally, \hexagonAlg\ is guaranteed to maintain the connectivity of an initially connected system configuration by Lemma~3 of~\cite{Daymude2023-canonicalamoebot}, satisfying Convention~\ref{conv:connect}.
\end{proof}

Combining this lemma, the energy distribution framework's guarantees (Theorem~\ref{thm:main}), \hexagonAlg's correctness guarantees (Theorem~8 of~\cite{Daymude2023-canonicalamoebot}), and \hexagonAlg's $\Theta(\numAmoebots^2)$ worst-case work bound~\cite{Derakhshandeh2015-algorithmicframework}, we have:

\begin{theorem} \label{thm:hexagonformation}
    For any demand function $\demand : \hexagonAlg \to \{1, 2, \ldots, \capacity\}$, the algorithm $\hexagonAlg^\demand$ produced by the energy distribution framework deterministically solves the hexagon formation problem for connected systems of $\numAmoebots$ amoebots in $\bigo{\numAmoebots^4}$ rounds assuming geometric space, assorted orientations, constant-size memory, and an unfair sequential adversary.
\end{theorem}

\begin{figure}[t]
    \centering
    \begin{subfigure}{.24\textwidth}
        \centering
        \includegraphics[width=\textwidth]{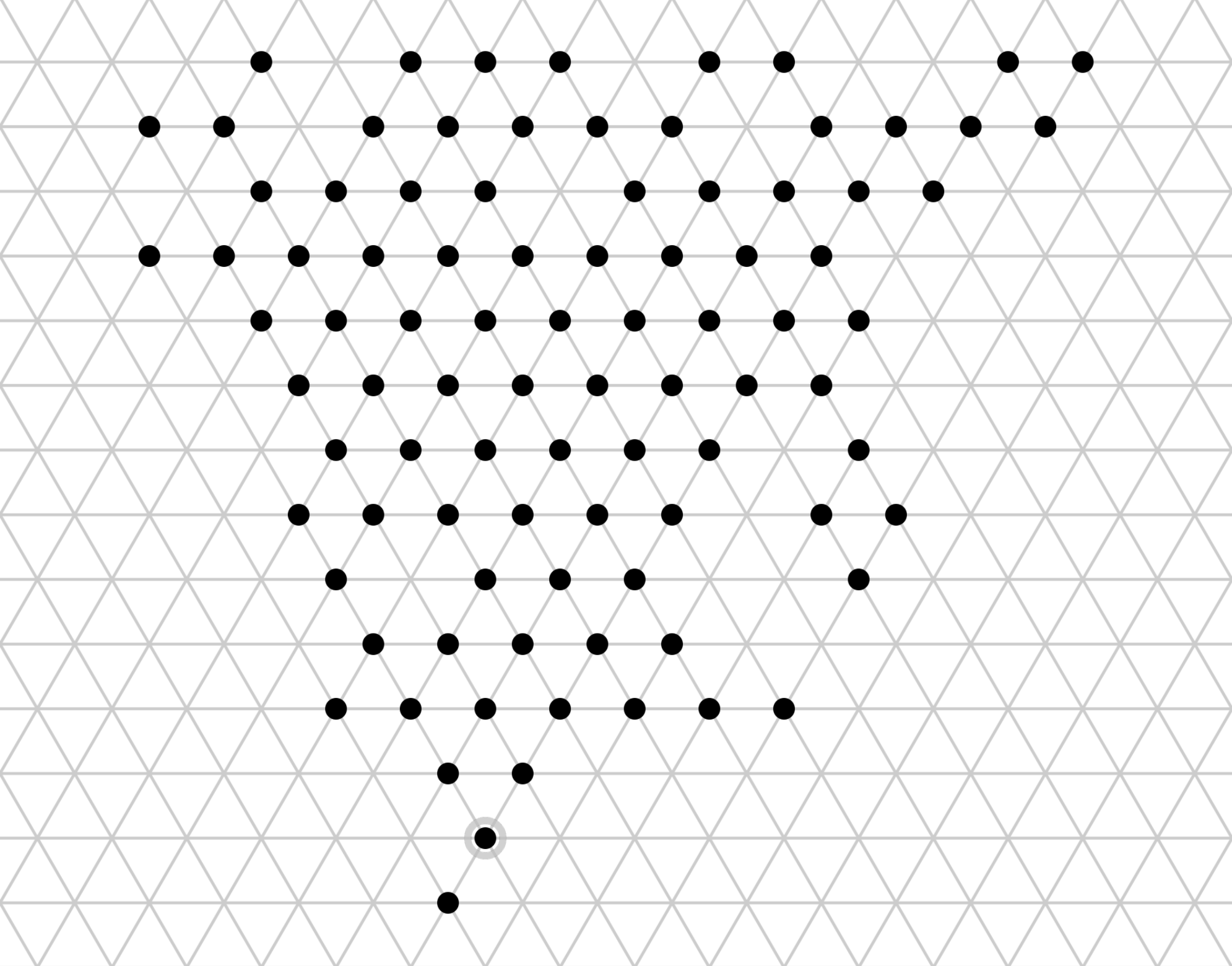}
        \caption{\centering $t = 0$ rounds}
        \label{fig:hexagonsim:a}
    \end{subfigure}
    \hfill
    \begin{subfigure}{.24\textwidth}
        \centering
        \includegraphics[width=\textwidth]{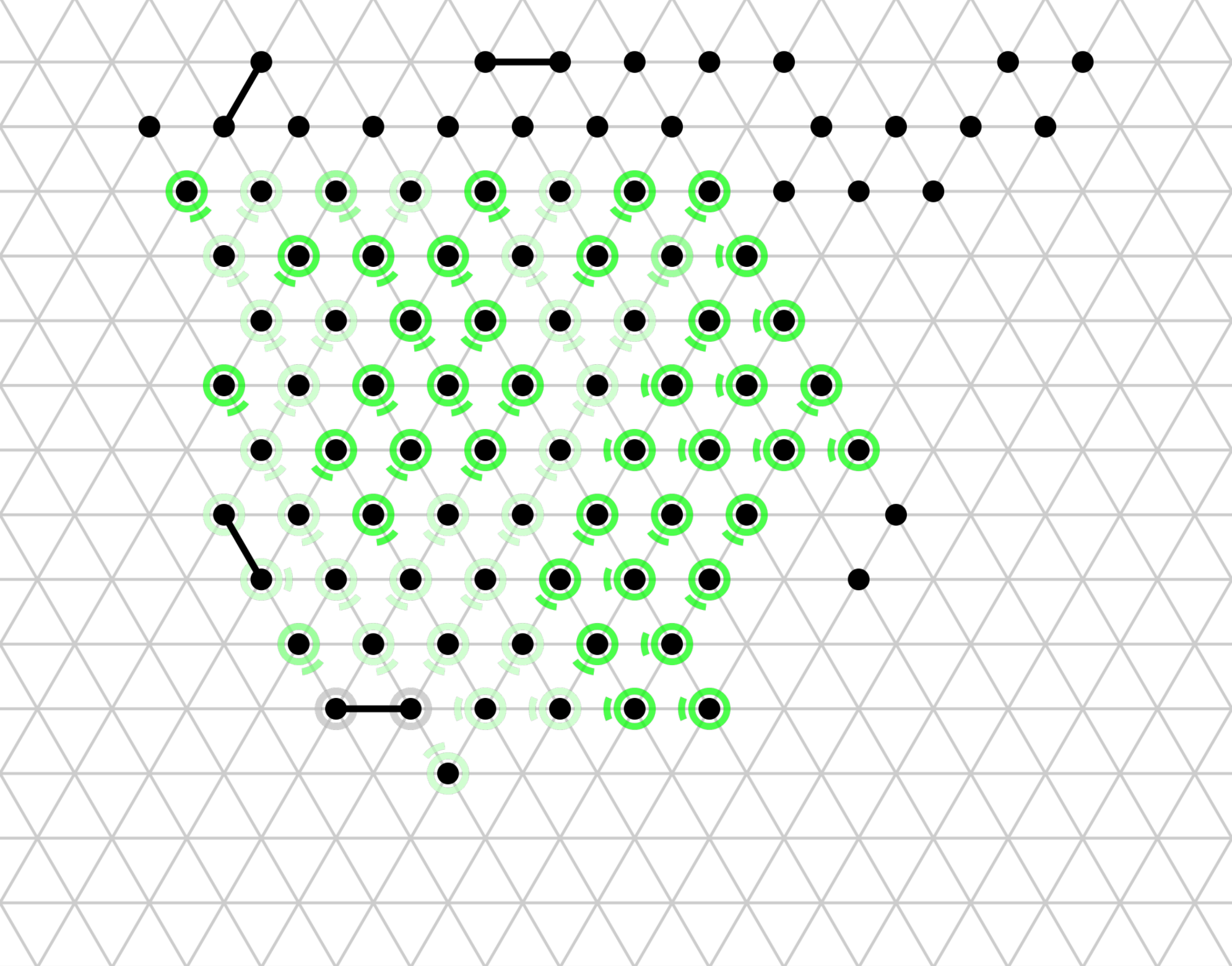}
        \caption{\centering $t = 400$}
        \label{fig:hexagonsim:b}
    \end{subfigure}
    \hfill
    \begin{subfigure}{.24\textwidth}
        \centering
        \includegraphics[width=\textwidth]{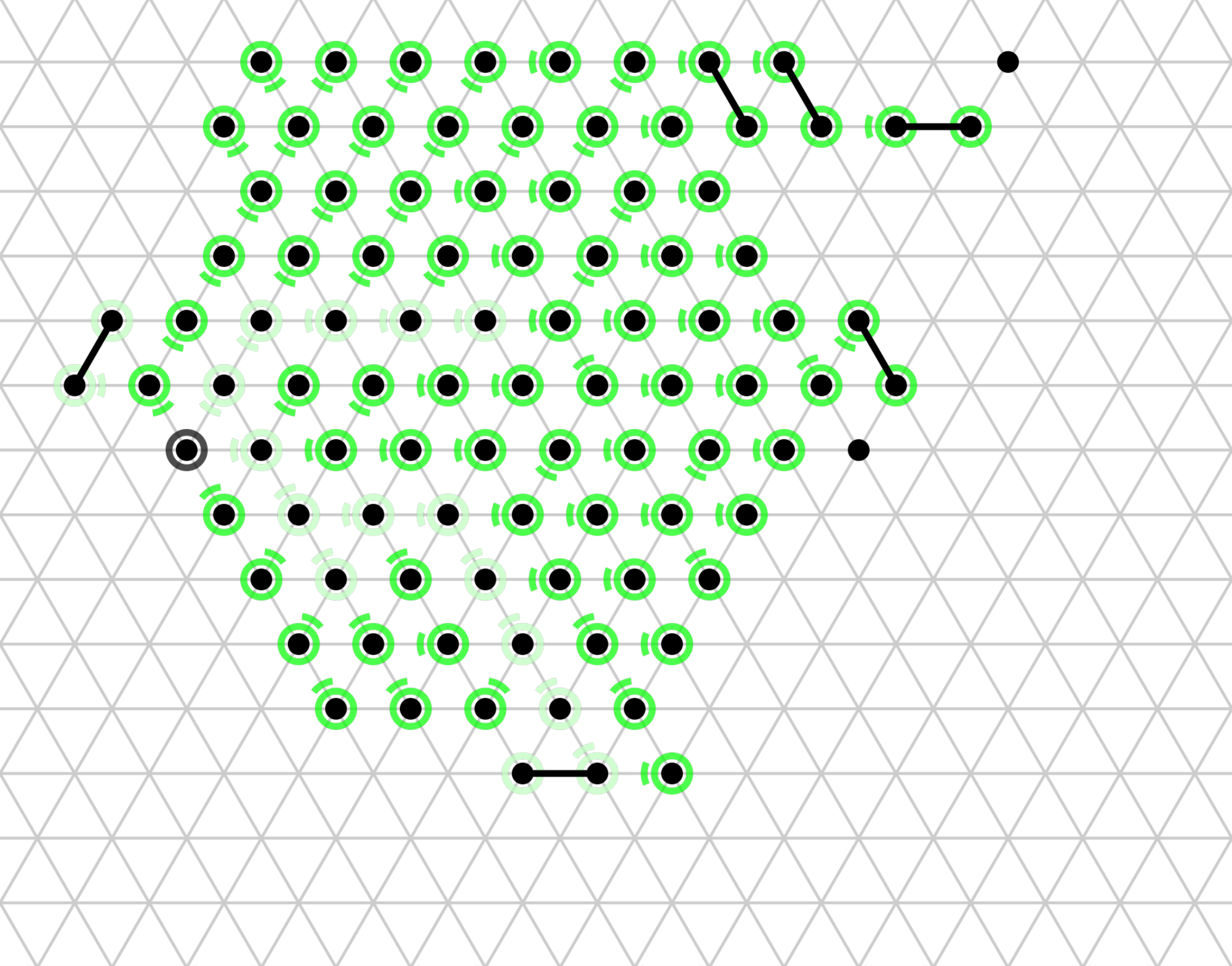}
        \caption{\centering $t = 900$}
        \label{fig:hexagonsim:c}
    \end{subfigure}
    \hfill
    \begin{subfigure}{.24\textwidth}
        \centering
        \includegraphics[width=\textwidth]{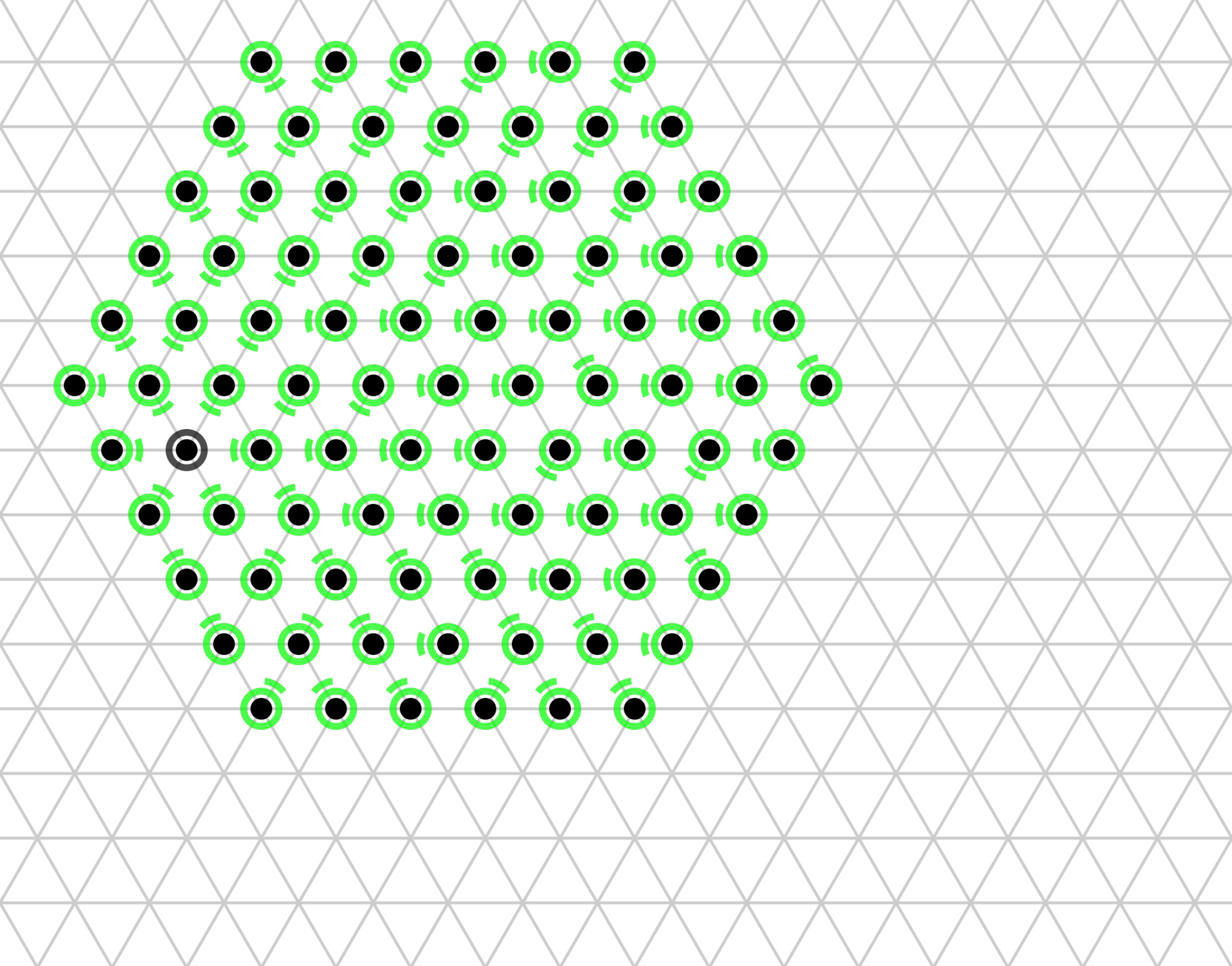}
        \caption{\centering $t = 1200$}
        \label{fig:hexagonsim:d}
    \end{subfigure}
    \caption{\textit{Simulating $\hexagonAlg^\demand$.}
    A simulation of $\hexagonAlg^\demand$ on $\numAmoebots = 91$ amoebots with one source amoebot, capacity $\capacity = 10$, and demand $\demand(\alpha) = 5$ for all actions $\alpha$.
    States from \hexagonAlg\ are not visualized.
    For energy distribution, color opacity indicates energy levels.
    All amoebots are initially \idle\ (no color) except the source (gray/black); amoebots eventually join the forest $\forest$ (green) and distribute energy.}
    \label{fig:hexagonsim}
\end{figure}

\figtext~\ref{fig:hexagonsim} depicts a simulation of $\hexagonAlg^\demand$ forming a hexagon under energy constraints.
We emphasize that \erosionAlg\ and \hexagonAlg\ are not cherry-picked examples with particularly straightforward proofs of energy-compatibility.
On the contrary, we expect that like our two examples, many algorithms already have the ingredients of energy-compatibility proven in their existing correctness analyses.

\begin{figure}[t]
    \centering
    \begin{subfigure}{0.48\textwidth}
        \centering
        \includegraphics[width=\textwidth]{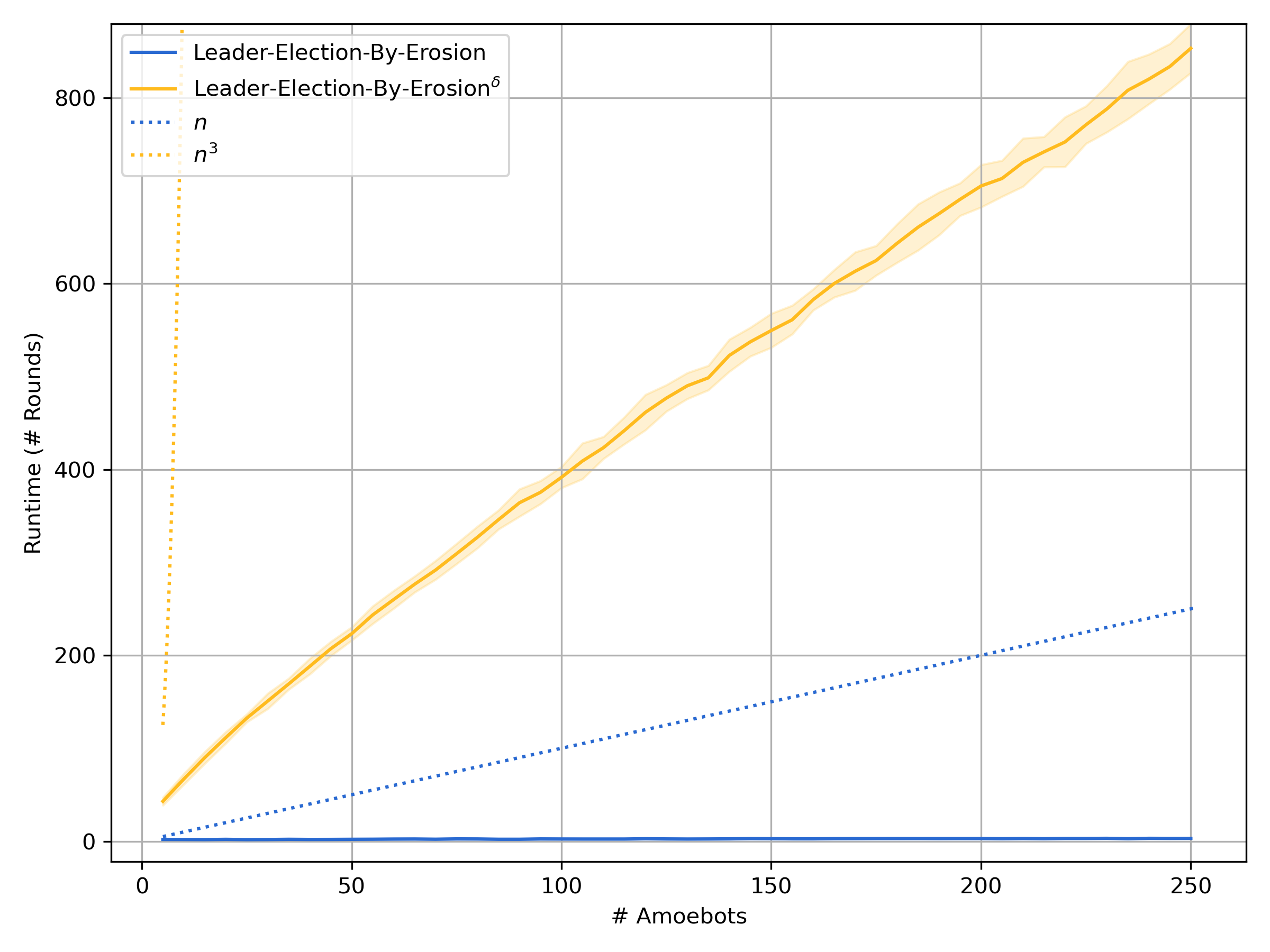}
        \caption{\centering \erosionAlg}
        \label{fig:comptime:erosion}
    \end{subfigure}
    \hfill
    \begin{subfigure}{0.48\textwidth}
        \centering
        \includegraphics[width=\textwidth]{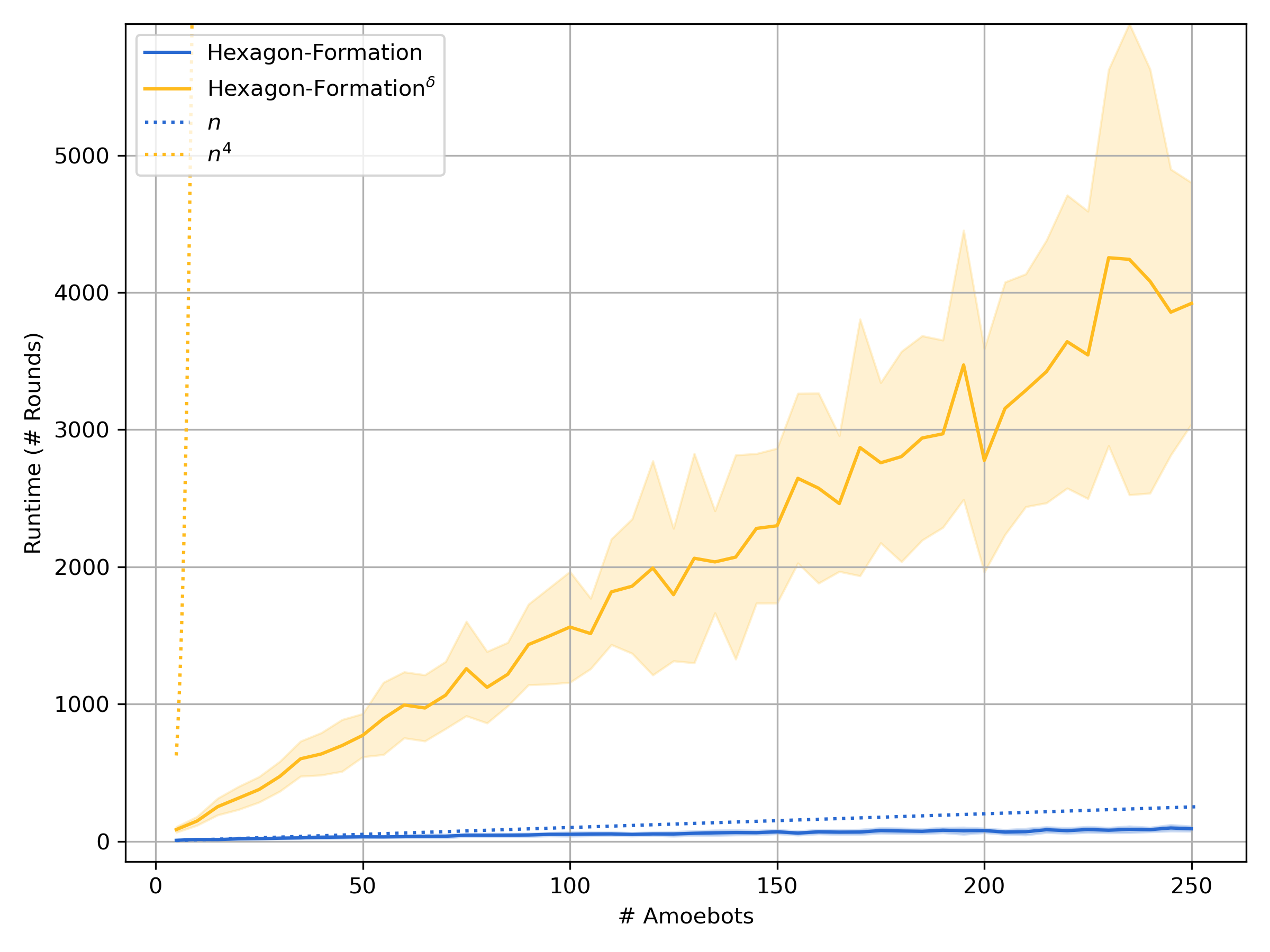}
        \caption{\centering \hexagonAlg}
        \label{fig:comptime:hexagon}
    \end{subfigure}
    \caption{\textit{Runtime Comparisons.}
    The energy-constrained (a) $\erosionAlg^\demand$ and (b) $\hexagonAlg^\demand$ algorithms' runtimes (yellow) and their energy-agnostic counterparts (blue) in terms of sequential rounds.
    Each algorithm was simulated in 25 independent trials per system size $\numAmoebots \in \{5, 10, \ldots, 250\}$; average runtimes are shown as solid lines and one standard deviation is shown as an error tube.
    Relevant asymptotic runtime bounds are shown as dotted lines: the energy-agnostic algorithms both terminate in linear rounds (blue) and the energy-constrained algorithms' bounds are given by Theorems~\ref{thm:leaderelection} and~\ref{thm:hexagonformation} (yellow).}
    \label{fig:comptime}
\end{figure}

We validate the runtime bounds for $\erosionAlg^\demand$ and $\hexagonAlg^\demand$ given in Theorems~\ref{thm:leaderelection} and~\ref{thm:hexagonformation}, respectively, by simulating these algorithms and their energy-agnostic counterparts for a range of system sizes $\numAmoebots$.
\figtext~\ref{fig:comptime} reports their empirical runtimes.
Both energy-constrained algorithms well outperform their theoretical bounds, with $\erosionAlg^\demand$ achieving a near-linear runtime and $\hexagonAlg^\demand$ remaining sub-quadratic.
This suggests that our overhead bound can be optimized further or describes only some pessimistic worst-case scenarios.
\ifnum\version=\ARXIV
In Section~\ref{sec:conclude}, we suggest an open problem whose solution would improve our overhead bound from $\bigo{\numAmoebots^2}$ rounds to $\bigo{\numAmoebots D}$ rounds, where $\sqrt{\numAmoebots} \leq D \leq \numAmoebots$ is the diameter of the amoebot system.
\else\fi

\section{Asynchronous Energy-Constrained Algorithms} \label{sec:concurrency}

Our energy distribution results thus far consider sequential concurrency, in which at most one amoebot can be active at a time (Section~\ref{subsec:model}).
This section details a useful extension of these results to \textit{asynchronous concurrency}, in which arbitrary amoebots can be simultaneously active and their action executions can overlap arbitrarily in time.

There are many hazards of asynchrony that complicate amoebot algorithm design, with concurrent movements and memory updates potentially causing operations to fail or action executions to exhibit unintended behaviors.
To reduce this complexity, one can use the \textit{concurrency control framework} for amoebot algorithms that---analogous to our own energy distribution framework for energy-agnostic/constrained algorithms---transforms any algorithm $\alg$ that terminates under every (unfair) sequential execution and satisfies certain conventions into an algorithm $\alg'$ that achieves equivalent behavior under any asynchronous execution~\cite{Daymude2023-canonicalamoebot}.
Formally, an amoebot algorithm $\alg$ is \textit{concurrency-compatible} if every (unfair) sequential execution of $\alg$ terminates and it satisfies the validity, phase structure, and expansion-robustness conventions.
The first two conventions are identical to Conventions~\ref{conv:valid} and~\ref{conv:phases} of the energy distribution framework.
The third convention, \textit{expansion-robustness}, requires actions to be resilient to concurrent expansions into their neighborhood.


We originally aimed to prove that the energy distribution framework preserves any input algorithm's concurrency-compatibility---i.e., if an algorithm $\alg$ is concurrency-compatible, then so is $\alg^\demand$---and thus the two frameworks can be composed to obtain energy-constrained, asynchronous versions of all energy-compatible, concurrency-compatible algorithms.
But as will become clearer after we formally define expansion-robustness (Definition~\ref{def:expandrobust}), knowing that $\alg$ is expansion-robust is seemingly insufficient for proving that $\alg^\demand$ is also expansion-robust: the former only describes terminating configurations for $\alg$ while the latter requires analyzing possible amoebot movements in all intermediate configurations reached by $\alg^\demand$.
Instead, we focus on a special case of expansion-robustness called \textit{expansion-correspondence} (Definition~\ref{def:expandcorrespond}) that we can prove is preserved by the energy distribution framework (Lemma~\ref{lem:edfconcurrency}).
Although this restriction may appear limiting, the only algorithm known to be non-trivially expansion-robust (\hexagonAlg\ of~\cite{Daymude2023-canonicalamoebot}) was proven to be expansion-robust via expansion-correspondence.
Thus, until an algorithm is discovered to be expansion-robust but not expansion-corresponding, our present focus covers all known concurrency-compatible algorithms.

\ifnum\version=\ARXIV
\begin{algorithm}[t]
    \caption{Expansion-Robust Variant $\alg^E$ of Algorithm $\alg$ for Amoebot $A$} \label{alg:expandrobust}
    \begin{algorithmic}[1]
        \Statex \textbf{Input}: Algorithm $\alg = \{[\alpha_i : g_i \to ops_i] : i \in \{1, \ldots, m\}\}$ satisfying Conventions~\ref{conv:valid} and~\ref{conv:phases}.
        \State Set $\alpha_0^E : (\exists$ port $p$ of $A : A.\xflag_p = \true) \to$ \Write$(\bot, \xflag_p, \false)$.
        \For {each action $[\alpha_i : g_i \to ops_i] \in \alg$}
            \State Set $g_i^E \gets g_i$ with $N(A)$ replaced by $N^E(A)$ and connections defined w.r.t.\ $N^E(A)$.
            \State Set $ops_i^E \gets$ ``Do:
            \Indent
                \For {each port $p$ of $A$} \Write$(\bot, \xflag_p, \false)$.  \Comment{Reset own expand flags.}  \label{alg:expandrobust:resetown}
                \EndFor
                \For {each unique neighbor $B \in \Connected()$}
                    \For{each port $p$ of $B$} \Write$(B, \xflag_p, \false)$.  \Comment{Reset neighbors' expand flags.}  \label{alg:expandrobust:resetnbr}
                    \EndFor
                \EndFor
                \State Execute each operation of $ops_i$ with connections defined w.r.t.\ $N^E(A)$.
                \If {a \Pull\ or \Push\ operation was executed with neighbor $B$}
                    \For {each new port $p$ of $A$ not connected to $B$} \Write$(\bot, \xflag_p, \true)$.  \label{alg:expandrobust:handover1}
                    \EndFor
                    \For {each new port $p$ of $B$ not connected to $A$} \Write$(B, \xflag_p, \true)$.  \label{alg:expandrobust:handover2}
                    \EndFor
                \ElsIf {an \Expand\ operation was successfully executed}
                    \For {each new port $p$ of $A$} \Write$(\bot, \xflag_p, \true)$.  \label{alg:expandrobust:expand}
                    \EndFor
                \ElsIf {an \Expand\ operation failed in its execution} undo $ops_i$.''
                \label{alg:expandrobust:undo}
                \EndIf
            \EndIndent
        \EndFor
        \State \Return $\alg^E = \{[\alpha_i^E : g_i^E \to ops_i^E] : i \in \{0, \ldots, m\}\}$.
    \end{algorithmic}
\end{algorithm}
\else\fi

Formally, let $\alg$ be any amoebot algorithm satisfying Conventions~\ref{conv:valid} and~\ref{conv:phases} and consider its expansion-robust variant $\alg^E$ defined as follows.
Each amoebot $A$ executing $\alg^E$ additionally stores in public memory an \textit{expand flag} $A.\xflag_p$ for each of its ports $p$ that is initially \false, becomes \true\ whenever $A$ expands to reveal a new port $p$, and is reset to \false\ whenever $A$ or one of its neighbors executes a later action.
These expand flags communicate when an amoebot has newly expanded into another amoebot's neighborhood.
\ifnum\version=\ARXIV
Each action $\alpha_i : g_i \to ops_i$ in $\alg$ becomes an action $\alpha_i^E : g_i^E \to ops_i^E$ in $\alg^E$, as detailed in Algorithm~\ref{alg:expandrobust} (reproduced from~\cite{Daymude2023-canonicalamoebot}).\footnote{For the sake of clarity and brevity, we abuse \Connected, \Read, and \Write\ notation slightly by referring directly to the neighboring amoebots and not to the ports which they are connected to.}
\else
Each action $\alpha_i : g_i \to ops_i$ in $\alg$ becomes an action $\alpha_i^E : g_i^E \to ops_i^E$ in $\alg^E$ (see Algorithm~\ref{alg:expandrobust} in Appendix~\ref{app:asyncproofs} for details).
\fi
The main difference is that while an amoebot $A$ executes actions with respect to its full neighborhood $N(A)$ in $\alg$, it does so only with respect to its \textit{established neighborhood} $N^E(A) = \{B \in N(A) : \exists \text{ port $p$ of $B$ connected to $A$ s.t.\ } B.\xflag_p = \false\}$ in $\alg^E$, effectively ignoring its newly expanded neighbors until its next action execution.

\begin{definition} \label{def:expandrobust}
    An amoebot algorithm $\alg$ is \underbar{expansion-robust} if for any (legal) initial system configuration $C_0$ of $\alg$, the following conditions hold:
    \begin{enumerate}
        \item If all sequential executions of $\alg$ starting in $C_0$ terminate, all sequential executions of $\alg^E$ starting in $C_0^E$ (i.e., $C_0$ with all \false\ expand flags) also terminate.

        \item If a sequential execution of $\alg^E$ starting in $C_0^E$ terminates in a configuration $C^E$, some sequential execution of $\alg$ starting in $C_0$ terminates in $C$ (i.e., $C^E$ without expand flags).
    \end{enumerate}
\end{definition}

As alluded to earlier, expansion-robustness only guarantees that sequential executions of $\alg^E$ terminate and do so in a configuration that is reachable by a sequential execution of $\alg$.
This appears to be insufficient to prove $\alg^\demand$ is expansion-robust.
\ifnum\version=\ARXIV
We instead focus on the following property, which we prove is a special case of expansion-robustness in Lemma~\ref{lem:expandcorrespond}.
\else
We instead focus on the following special case of expansion-robustness.
\fi

\begin{definition} \label{def:expandcorrespond}
    An amoebot algorithm $\alg$ is \underbar{expansion-corresponding} if for any (legal) initial system configuration $C_0$ of $\alg$, the following conditions hold:
    \begin{enumerate}
        \item If an action $\alpha_{i \neq 0}^E \in \alg^E$ is enabled for some amoebot $A$ w.r.t.\ $N^E(A)$, then action $\alpha_i \in \alg$ is enabled for $A$ w.r.t.\ $N(A)$.

        \item The executions of $\alpha_{i \neq 0}^E$ w.r.t.\ $N^E(A)$ and $\alpha_i$ w.r.t.\ $N(A)$ by an amoebot $A$ are identical, except the handling of expand flags.
    \end{enumerate}
\end{definition}

\ifnum\version=\ARXIV
\begin{lemma} \label{lem:expandcorrespond}
    If amoebot algorithm $\alg$ is expansion-corresponding, it is also expansion-robust.
\end{lemma}
\begin{proof}
    To prove termination, suppose to the contrary that all sequential executions of $\alg$ starting in $C_0$ terminate, but there exists some infinite sequential execution $\sched^E$ of $\alg^E$ starting in $C_0^E$.
    Algorithm $\alg$ is expansion-corresponding, so there is a sequential execution $\sched$ that is identical to $\sched^E$, modulo executions of $\alpha_0^E$.
    Execution $\sched$ terminates by supposition, so $\sched^E$ must contain an infinite number of $\alpha_0^E$ executions after its final $\alpha_{i \neq 0}^E$ execution.
    But $\alpha_0^E$ executions only reset expand flags, and there are only a finite number of amoebots and a constant number of expand flags per amoebot to reset, a contradiction.

    Correctness follows from the same observation.
    Only $\alpha_{i \neq 0}^E$ executions move amoebots and modify variables of $\alg$.
    Since every sequential execution $\sched^E$ of $\alg^E$ starting in $C_0^E$ represents an identical sequential execution $\sched$ of $\alg$ starting in $C_0$ (after removing the $\alpha_0^E$ executions), and since $\sched^E$ terminates whenever $\sched$ terminates by the above argument, we conclude that they must terminate in configurations that are identical, modulo expand flags.
\end{proof}

Before proving that the energy distribution framework preserves expansion-correspondence, we need one helper lemma characterizing established neighbors in $\alg^\demand$.

\begin{lemma} \label{lem:established}
    During an execution of $(\alg^\demand)^E$, if an amoebot $A$ has a neighbor $B \in N(A)$ that is \idle, \pruning, or a child of $A$, then $B \in N^E(A)$.
\end{lemma}
\begin{proof}
    Any neighbor $B \in N(A) \setminus N^E(A)$ expanded into $N(A)$ during an \Expand\ operation by $B$, a \Push\ operation by $B$, or a \Pull\ operation by some other amoebot pulling $B$.
    Any movement in $(\alg^\demand)^E$ occurs in an $(\alpha_i^\demand)^E$ execution, whose guard requires that both the executing amoebot and all its established neighbors are not \idle\ or \pruning.
    Thus, regardless of whether $B$ is initiating the movement (an \Expand\ or \Push) or is participating in it (a \Pull), $B$ cannot be \idle\ or \pruning\ when it enters $N(A)$.
    Any subsequent action execution that could make $B$ \idle\ or \pruning\ must also reset its expand flags (Algorithm~\ref{alg:expandrobust}, Line~\ref{alg:expandrobust:resetnbr}).
    So there are never \idle\ or \pruning\ neighbors in $N(A) \setminus N^E(A)$.

    Next consider any child $B$ of $A$.
    Amoebot $B$ became a child of $A$ when $A$ adopted it during a $g_\growforest$-supported execution of $\alpha_\energydist^E$.
    During this execution, $A$ reset all expand flags of $B$ (Algorithm~\ref{alg:expandrobust}, Line~\ref{alg:expandrobust:resetnbr}).
    As long as $B$ is a child of $A$, its expand flags facing $A$ remain reset.
    Thus, $B \in N^E(A)$.
\end{proof}

We can now prove the main lemma of this section.

\begin{lemma} \label{lem:edfconcurrency}
    For any energy-compatible, expansion-corresponding algorithm $\alg$ and demand function $\demand : \alg \to \{1, 2, \ldots, \capacity\}$, the algorithm $\alg^\demand$ produced from $\alg$ and $\demand$ by the energy distribution framework is concurrency-compatible.
\end{lemma}
\begin{proof}
    By Theorem~\ref{thm:main}, we know that every sequential execution of $\alg^\demand$ terminates.
    It remains to show that $\alg^\demand$ satisfies the validity, phase structure, and expansion-robustness conventions.
    
    By supposition, every action $\alpha_i \in \alg$ in the original algorithm is valid, i.e., its execution is successful whenever it is enabled and all other amoebots are inactive.
    Since the guard $g_i$ of $\alpha_i$ is a necessary condition for the energy-constrained version $\alpha_i^\demand$ to be enabled, we know this validity carries over to the compute and movement phases of $\alpha_i$.
    The only new operations added by the energy distribution framework in the $\alpha_i^\demand$ and $\alpha_\energydist$ actions are \Connected\ operations (which never fail) and \Read\ and \Write\ operations involving existing neighbors.
    All of these must succeed, so every action of $\alg^\demand$ is valid.

    It is easy to see that $\alg^\demand$ satisfies the phase structure convention.
    Its only movements are in the $\alpha_i^\demand$ actions, each of which has at most one movement operation that it executes last.
    Moreover, the energy distribution framework does not add any \Lock\ or \Unlock\ operations.

    It remains to show $\alg^\demand$ is expansion-robust, and by Lemma~\ref{lem:expandcorrespond}, it suffices to show $\alg^\demand$ is expansion-corresponding.
    We first show that if some action of $(\alg^\demand)^E$ is enabled for an amoebot $A$ w.r.t.\ $N^E(A)$, then the corresponding action of $\alg^\demand$ is enabled for $A$ w.r.t.\ $N(A)$.
    We may safely consider only the guard conditions that depend on an amoebot's neighborhood; all others evaluate identically regardless of neighborhood.
    \begin{itemize}
        \item If $(\alpha_i^\demand)^E$ is enabled for an amoebot $A$, then $A$ must satisfy $g_i^E$---i.e., $A$ satisfies the guard $g_i$ of $\alpha_i \in \alg$ w.r.t.\ $N^E(A)$---and neither $A$ nor its established neighbors can be \idle\ or \pruning.
        Algorithm $\alg$ is expansion-corresponding by supposition, so this implies that $A$ must satisfy $g_i$ w.r.t.\ $N(A)$ as well.
        Moreover, Lemma~\ref{lem:established} ensures that if there are no \idle\ or \pruning\ neighbors in $N^E(A)$, there are none in $N(A)$ either.

        \item Suppose $\alpha_\energydist^E$ is enabled for an amoebot $A$ because $A$ has an \idle\ neighbor or an \asking\ child $B \in N^E(A)$, a condition in both $g_\askgrowth$ and $g_\growforest$.
        We know $N^E(A) \subseteq N(A)$, so $\alpha_\energydist$ must be enabled for $A$ w.r.t.\ $N(A)$ as well.

        \item Suppose $\alpha_\energydist^E$ is enabled for an amoebot $A$ because $A$ has a child $B \in N^E(A)$ whose battery is not full, a condition in $g_\shareenergy$.
        By the same argument as above, we have $N^E(A) \subseteq N(A)$, so $\alpha_\energydist$ must be enabled for $A$ w.r.t.\ $N(A)$ as well.
    \end{itemize}

    Finally, we show that the executions of any action of $(\alg^\demand)^E$ w.r.t.\ $N^E(A)$ and the corresponding action of $\alg^\demand$ w.r.t.\ $N(A)$ by the same amoebot $A$ are identical.
    We may safely focus only on the parts of action executions that depend on or interact with an amoebot's neighbors; all others execute identically regardless of neighborhood.
    \begin{itemize}
        \item If $A$ executes an $(\alpha_i^\demand)^E$ action, it emulates the operations of $\alpha_i \in \alg$ w.r.t.\ $N^E(A)$.
        But algorithm $\alg$ is expansion-corresponding by supposition, which immediately implies that an execution of $\alpha_i$ w.r.t.\ $N(A)$ is identical.

        \item If $A$ executes an $(\alpha_i^\demand)^E$ action or the \getpruned\ block of $\alpha_\energydist^E$, it may update its children's \xstate\ and \parent\ variables during $\textsc{Prune}(\,)$.
        By Lemma~\ref{lem:established}, any child of $A$ in $N(A)$ is also in $N^E(A)$, so the same children are pruned.

        \item If $A$ executes the \growforest\ block of $\alpha_\energydist^E$, it adopts all its \idle\ neighbors as an \xactive\ children.
        Any \idle\ neighbor $B \in N^E(A)$ that $A$ adopts must also be adopted when $A$ executes $\alpha_\energydist$ since $N^E(A) \subseteq N(A)$.
        But if there are no \idle\ neighbors in $N^E(A)$ for $A$ to adopt, there cannot be any in $N(A)$ either by Lemma~\ref{lem:established}.
        Thus, either the same \idle\ neighbors or no neighbors are adopted.

        \item If $A$ executes the \growforest\ block of $\alpha_\energydist^E$, it updates any \asking\ children to \growing.
        By Lemma~\ref{lem:established}, any child of $A$ in $N(A)$ is also in $N^E(A)$, so the same children are updated in $\alpha_\energydist$.

        \item If $A$ executes the \shareenergy\ block of $\alpha_\energydist^E$, it transfers an energy unit to one of its children $B \in N^E(A)$ whose battery is not full.
        We know $N^E(A) \subseteq N(A)$, so $B$ is also a possible recipient of this energy in $\alpha_\energydist$. \qedhere
    \end{itemize}
\end{proof}
\else
The main lemma of this section proves that the energy distribution framework preserves expansion-correspondence.
Its proof and supporting results can be found in Appendix~\ref{app:asyncproofs}.

\begin{lemma} \label{lem:edfconcurrency}
    For any energy-compatible, expansion-corresponding algorithm $\alg$ and demand function $\demand : \alg \to \{1, 2, \ldots, \capacity\}$, the algorithm $\alg^\demand$ produced from $\alg$ and $\demand$ by the energy distribution framework is concurrency-compatible.
\end{lemma}
\fi

Lemma~\ref{lem:edfconcurrency} shows that the energy distribution and concurrency control frameworks can be composed to obtain the benefits of both.
Specifically, an amoebot algorithm designer should first design their algorithm without energy constraints and perform the usual safety and liveness analyses with respect to an unfair sequential adversary.
If the algorithm always terminates, then they need only prove their algorithm satisfies the validity, phase structure, and connectivity conventions and argue that their algorithm is expansion-corresponding to automatically obtain an energy-constrained, asynchronous version of their algorithm with equivalent behavior, courtesy of the two frameworks.
The following theorem states this result formally by combining the energy distribution framework's guarantees (Theorem~\ref{thm:main}), the concurrency control framework's guarantees (Theorem~11 of~\cite{Daymude2023-canonicalamoebot}), and Lemma~\ref{lem:edfconcurrency}.
Note that because the runtime overhead of the concurrency control framework is not known, this theorem does not give any overhead bounds.

\begin{theorem} \label{thm:async}
    Consider any energy-compatible, expansion-corresponding amoebot algorithm $\alg$ and demand function $\demand : \alg \to \{1, 2, \ldots, \capacity\}$.
    Let $\alg^\demand$ be the algorithm produced from $\alg$ and $\demand$ by the energy distribution framework (Algorithm~\ref{alg:framework}) and let $(\alg^\demand)'$ be the algorithm produced from $\alg^\demand$ by the concurrency control framework (Algorithm~4 of~\cite{Daymude2023-canonicalamoebot}).
    Let $C_0$ be any (legal) connected initial configuration for $\alg$ and let $(C_0^\demand)'$ be its extension for $(\alg^\demand)'$ that designates at least one source amoebot and adds the energy distribution and concurrency control variables with their initial values (Table~\ref{tab:frameworkvariables} and \texttt{act} and \texttt{awaken} of~\cite{Daymude2023-canonicalamoebot}) to all amoebots.
    Then every asynchronous execution of $(\alg^\demand)'$ starting in $(C_0^\demand)'$ terminates.
    Moreover, if $(C^\demand)'$ is the final configuration of some asynchronous execution of $(\alg^\demand)'$ starting in $(C_0^\demand)'$, then there exists a sequential execution of $\alg$ starting in $C_0$ that terminates in a configuration $C$ that is identical to $(C^\demand)'$ modulo the energy distribution and concurrency control variables.
\end{theorem}

We conclude this section by applying Theorem~\ref{thm:async} to the \erosionAlg\ and \hexagonAlg\ algorithms from Section~\ref{sec:edfcompatible}.
Those algorithms were shown to be energy-compatible in Lemmas~\ref{lem:leaderelection} and~\ref{lem:hexagonformation} and expansion-corresponding in Lemma~7.1 of~\cite{Briones2023-invitedpaper} and Theorem~10 of~\cite{Daymude2023-canonicalamoebot}, respectively.
Therefore,

\begin{corollary} \label{cor:energyasyncalgs}
    There exist energy-constrained amoebot algorithms that deterministically solve the leader election problem (for hole-free, connected systems) and the hexagon formation problem (for connected systems) assuming geometric space, assorted orientations, constant-size memory, and an unfair asynchronous adversary---the most general of all adversaries.
\end{corollary}

\section{Conclusion}  \label{sec:conclude}

In this work, we introduced the energy distribution framework for amoebot algorithms which transforms any energy-agnostic algorithm into an energy-constrained one with equivalent behavior, provided the original algorithm terminates under an unfair sequential adversary, maintains system connectivity, and follows some basic structural conventions (Theorem~\ref{thm:main}).
We then proved that both the \erosionAlg\ and \hexagonAlg\ algorithms are energy-compatible (Theorems~\ref{thm:leaderelection} and~\ref{thm:hexagonformation}).
Perhaps surprisingly, these proofs were not difficult.
The algorithms' existing correctness and runtime analyses under an unfair sequential adversary provided nearly all that was needed for energy-compatibility, and we expect this would be true for other algorithms as well.
Finally, we proved that if an energy-compatible algorithm is also expansion-corresponding, then its energy-constrained counterpart produced by our framework can be extended to asynchronous concurrency using the concurrency control framework for amoebot algorithms (Theorem~\ref{thm:async}).

The energy-constrained algorithms produced by our framework have an $\bigo{\numAmoebots^2}$ round runtime overhead, though our simulations of $\erosionAlg^\demand$ and $\hexagonAlg^\demand$ suggest that the overhead is much lower in practice.
Comparing Lemmas~\ref{lem:stabletime} and~\ref{lem:rechargetime} reveals the spanning forest maintenance algorithm as the performance bottleneck, which uses $\bigo{\numAmoebots^2}$ rounds in the worst case to prune and rebuild a forest of stable trees.
\ifnum\version=\ARXIV
In particular, amoebots getting permission from their (source) root before adopting children is critical for avoiding non-termination under an unfair adversary (Lemma~\ref{lem:forestblocksfinite}), but requires a number of rounds that is linear in the depth of the tree (Lemma~\ref{lem:growtime}).
\else
In particular, amoebots getting permission from their (source) root before adopting children is critical for avoiding non-termination under an unfair adversary (Lemma~\ref{lem:energyrunfinite}), but requires a number of rounds that is linear in the depth of the tree.
\fi
Improving this bound either requires a new approach to acyclic resource distribution or an optimization of stable tree membership detection.
A shortest-path tree---i.e., one that maintains equality between the in-tree and in-system distances from any amoebot to its root---would bound the depth of any tree by the diameter $D$ of the system.
This would reduce the overall overhead to $\bigo{nD}$ rounds, which is still $\bigo{\numAmoebots^2}$ in the worst case (e.g., a line) but could achieve up to $\bigo{\numAmoebots^{3/2}}$ in the best case (e.g., a regular hexagon).
However, the recent feather tree algorithm~\cite{Kostitsyna2022-briefannouncement} for forming shortest-path forests in amoebot systems only works in stationary systems.
Achieving an algorithm for shortest-path forest maintenance---not just formation---would both improve our present overhead bound and be an interesting contribution in its own right.

\ifnum\version=\SUBMISSION
\clearpage
\else\fi

\bibliographystyle{plainurl}
\bibliography{ref}

\ifnum\version=\SUBMISSION
\clearpage
\appendix

\section{Omitted Analysis of the Energy Distribution Framework} \label{app:edfproofs}

This section contains the technical material omitted from Section~\ref{sec:framework} due to space constraints.

\subparagraph{Analysis Overview.}

We outline our analysis as follows.
We start by considering an arbitrary sequential execution $\sched^\demand$ of $\alg^\demand$ starting in $C_0^\demand$.
One way of conceptualizing $\sched^\demand$ is as a sequence of \textit{energy runs}---i.e., maximal sequences of consecutive $\alpha_\energydist$ executions---that are delineated by sequences of $\alpha_i^\demand$ executions.
In fact, $\sched^\demand$ contains only a finite number of $\alpha_i^\demand$ executions (and thus a finite number of energy runs) because the corresponding sequence of $\alpha_i$ executions forms a possible sequential execution $\sched_\alpha$ of $\alg$ (Lemma~\ref{app:lem:equivalence}), which must terminate because $\alg$ is energy-compatible.
It is exactly this execution $\sched_\alpha$ of $\alg$ that we will argue terminates in a configuration $C$ corresponding to the final configuration $C^\demand$ of $\sched^\demand$.

Of course, we have not yet shown that $\sched^\demand$ terminates at all under an unfair adversary, let alone in a final configuration corresponding to $\sched_\alpha$.
To do so, we will show that any energy run in $\sched^\demand$ is finite (Lemmas~\ref{app:lem:forestblocksfinite} and~\ref{app:lem:energyblocksfinite}); specifically, it either reaches a configuration where $\alpha_\energydist$ is disabled for all $\numAmoebots$ amoebots within $\bigo{\numAmoebots^2}$ rounds, or ends earlier because some $\alpha_i^\demand$ action is executed (Lemmas~\ref{app:lem:stabletime} and~\ref{app:lem:rechargetime}).
Since each energy run terminates within $\bigo{\numAmoebots^2}$ rounds and is delineated by a sequence of $\alpha_i^\demand$ executions, each $\alpha_i^\demand$ execution in $\sched^\demand$ can be mapped to an $\alpha_i$ execution in $\sched_\alpha$, and $\sched_\alpha$ contains at most $\algruntime$ action executions, we conclude that $\sched^\demand$ is not only finite, but terminates within $\bigo{\numAmoebots^2\algruntime}$ rounds.

Once it is established that both $\sched^\demand$ and $\sched_\alpha$ terminate, we argue that their respective final configurations $C^\demand$ and $C$ are identical (modulo the energy distribution variables).
Because every $\alpha_i^\demand$ execution in $\sched^\demand$ corresponds to a possible $\alpha_i$ execution in $\sched_\alpha$ (Lemma~\ref{app:lem:equivalence}), we know that any configuration reachable by $\sched^\demand$ is also reachable by $\sched_\alpha$.
So $\sched_\alpha$ must be able to reach a configuration $C$ corresponding to $C^\demand$, but we need to show that it will also terminate there; i.e., that the energy distribution aspects of $\alg^\demand$ don't impede it from making as much progress as $\alg$.
This will follow from the above energy run arguments, concluding the analysis.

\bigskip

We begin our analysis with two sets of invariants maintained by the energy distribution framework that we will reference repeatedly.
The first set describes useful properties of energy runs, i.e., maximal sequences of consecutive $\alpha_\energydist$ executions.
The second set characterizes all configurations reachable by algorithm $\alg^\demand$.

\begin{invariant} \label{app:inv:energyrun}
    In any energy run of any sequential execution of $\alg^\demand$ starting in $C_0^\demand$,
    \begin{enumerate}[label=(\alph*),ref=\ref{app:inv:energyrun}\alph*,leftmargin=1cm]
        \item \label{app:inv:energyrun:nospend} Energy is only harvested or transferred; it is never spent.

        \item \label{app:inv:energyrun:nomove} No amoebot ever moves.
    
        \item \label{app:inv:energyrun:stabletrees} Any amoebot that belongs to a stable tree of forest $\forest$ (i.e., one that is rooted at a source amoebot) will never change its \parent\ pointer.
    \end{enumerate}
\end{invariant}
\begin{proof}
    We prove each part independently.
    \begin{enumerate}[label=(\alph*),leftmargin=1cm]
        \item The only way for an amoebot to spend energy is during an $\alpha_i^\demand$ execution, which never occurs during an energy run by definition.

        \item The only way for an amoebot to move is during an $\alpha_i^\demand$ execution, which never occurs during an energy run by definition.
    
        \item The \parent\ pointer of an amoebot $A$ is only updated if $A$ contracts or is involved in a handover, calls $\textsc{Prune}(\,)$, or is adopted during \growforest.
        No amoebot moves during an energy run (Invariant~\ref{app:inv:energyrun:nomove}) and stable trees never prune by definition.
        So members of stable trees remain there throughout an energy run.
        \qedhere
    \end{enumerate}
\end{proof}

\begin{invariant} \label{app:inv:reachable}
    Any configuration reached by any sequential execution of $\alg^\demand$ starting in $C_0^\demand$:
    \begin{enumerate}[label=(\alph*),ref=\ref{app:inv:reachable}\alph*,leftmargin=1cm]
        \item \label{app:inv:reachable:connected} is connected.

        \item \label{app:inv:reachable:sources} contains at least one source amoebot.

        \item \label{app:inv:reachable:batteries} maintains $A.\battery \in \{0, 1, \ldots, \capacity\}$ for all amoebots $A$.
    \end{enumerate}
\end{invariant}
\begin{proof}
    We prove each part independently.
    \begin{enumerate}[label=(\alph*),leftmargin=1cm]
        \item The initial configuration $C_0^\demand$ is connected by supposition.
        All amoebot movements in $\alg^\demand$ originate from the movement phases of $\alpha_i$ actions from the original algorithm $\alg$.
        Since $\alg$ satisfies the connectivity convention (Convention~\ref{conv:connect}) by supposition, no configuration reachable from $C_0^\demand$ could ever be disconnected.

        \item The initial configuration $C_0^\demand$ contains at least one source amoebot by supposition.
        By inspection of Algorithm~\ref{alg:framework}, a source amoebot never updates its $\xstate$, so any source amoebot in $C_0^\demand$ remains a source amoebot throughout the execution of $\alg^\demand$.

        \item All amoebot batteries are initially empty in $C_0^\demand$.
        The guards $g_i^\demand$ and predicates $g_\harvestenergy$ and $g_\shareenergy$ ensure that $A.\battery \in [0, \capacity]$.
        Moreover, all changes to $A.\battery$ are integral: the $\alpha_i^\demand$ actions spend $\demand(\alpha_i) \in \{1, 2, \ldots, \capacity\}$ energy, \harvestenergy\ always harvests a single unit of energy into a source amoebot's battery, and \shareenergy\ always transfers a single unit of energy from a parent to one of its children.
        Noting that the battery capacity $\capacity$ is an integer, the invariant follows.
        \qedhere
    \end{enumerate}
\end{proof}

With the invariants in place, we can move on to analyzing sequential executions of $\alg^\demand$ representing any sequence of activations the unfair sequential adversary could have chosen.

\begin{lemma} \label{app:lem:equivalence}
    Consider any sequential execution $\sched^\demand$ of $\alg^\demand$ starting in initial configuration $C_0^\demand$ and let $\sched_\alpha^\demand$ denote its subsequence of $\alpha_i^\demand$ action executions.
    Then the corresponding sequence $\sched_\alpha$ of $\alpha_i$ executions is a valid sequential execution of $\alg$ starting in initial configuration $C_0$.
\end{lemma}
\begin{proof}
    Let $C_r^\demand$ (resp., $C_r$) denote the configuration reached by the first $r$ action executions in $\sched_\alpha^\demand$ starting in $C_0^\demand$ (resp., in $\sched_\alpha$ starting in $C_0$).
    Argue by induction on $r \geq 0$ that $C_r^\demand \cong C_r$; i.e., these configurations are identical with respect to amoebots' positions and the variables of $\alg$.
    This implies that $\sched_\alpha$ is a valid sequential execution of $\alg$ starting in $C_0$, as desired.

    If $r = 0$, then trivially $C_0^\demand \cong C_0$ by definition (see the statement of Theorem~\ref{thm:main}).
    So suppose $r \geq 1$.
    By the induction hypothesis, $C_{r-1}^\demand \cong C_{r-1}$.
    By definition, there is at most one energy run of $\alpha_\energydist$ executions in $\sched^\demand$ between $C_{r-1}^\demand$ and the configuration $C_{r'}^\demand$ in which the $r$-th $\alpha_i^\demand$ execution of $\sched_\alpha^\demand$ is enabled.
    But $\alpha_\energydist$ executions do not move amoebots or modify any variables of algorithm $\alg$, so $C_{r'}^\demand \cong C_{r-1}^\demand \cong C_{r-1}$.
    Also, any amoebot $A$ for which some $\alpha_i^\demand$ action is enabled must also satisfy the guard $g_i$ of action $\alpha_i$, by definition of the guard $g_i^\demand$.
    Thus, if $A$ executes $\alpha_i^\demand$ in $C_{r'}^\demand$, action $\alpha_i$ can also be executed by $A$ in $C_{r-1}$.
    Moreover, any amoebot movements or updates to variables of $\alg$ must be identical in both action executions, since $\alpha_i^\demand$ emulates $\alpha_i$.
    Therefore, $C_r^\demand \cong C_r$.
\end{proof}

Lemma~\ref{app:lem:equivalence} gives us a handle on the $\alpha_i^\demand$ action executions in any sequential execution of $\alg^\demand$, so it remains to analyze the energy runs between them.
In this first series of lemmas, we show that if $\alpha_\energydist$ is continuously enabled for some amoebot $A$ during an energy run, then within one additional round either $A$ is activated or the energy run is ended by some $\alpha_i^\demand$ action execution (Lemma~\ref{app:lem:round}).
Formally, we say an execution of $\alpha_\energydist$ by an amoebot $A$ is \textit{$g$-supported} if predicate $g \in \mathcal{G}$ is satisfied when $A$ is activated and executes $\alpha_\energydist$.
To prove eventual execution, we argue that any predicate $g \in \mathcal{G}$ can support at most a finite number of executions per energy run (Lemmas~\ref{app:lem:forestblocksfinite} and~\ref{app:lem:energyblocksfinite}).
Combining this with the definition of a round from Section~\ref{subsec:model} yields the one round upper bound on how long an $\alpha_\energydist$ action can remain continuously enabled in an energy run.

We begin with the \getpruned, \askgrowth, and \growforest\ blocks that maintain the spanning forest $\forest$.
Recall from Section~\ref{subsec:edf} that amoebots may move and disrupt the forest structure.
Thus, at the start of any energy run, the amoebot system is partitioned into \textit{stable trees} rooted at source amoebots, \textit{unstable trees} rooted at \pruning\ amoebots, and \idle\ amoebots that do not belong to any tree.
In the following lemma, we argue that amoebots cannot be trapped in an infinite loop of pruning and rejoining the forest $\forest$.

\begin{lemma} \label{app:lem:constantrejoin}
    In any energy run of $\sched^\demand$, no amoebot is pruned from and adopted into the forest $\forest$ more than eight times.
\end{lemma}
\begin{proof}
    By Invariant~\ref{app:inv:energyrun:stabletrees}, any amoebot that was already in a stable tree at the start of the energy run or is adopted into a stable tree during the energy run will remain there throughout the energy run.
    So suppose to the contrary that an amoebot $A$ is pruned from and adopted into unstable trees of the forest $\forest$ more than eight times.
    Since amoebot $A$ can have at most eight neighbors (if it is expanded) and none of these neighbors can move during an energy run (Invariant~\ref{app:inv:energyrun:nomove}), there must exist a neighbor $B$ that adopts $A$ into an unstable tree more than once.
    By the predicate $g_\growforest$ and the fact that $B$ cannot be a source if it is in an unstable tree, this implies that $B$ must become \growing\ multiple times.

    Observe that when a \growing\ amoebot transfers its \xstate\ to its \asking\ children during a $g_\growforest$-supported execution, it excludes any newly adopted child (which is \xactive) and then becomes \xactive.
    Moreover, because unstable trees are severed from source amoebots, no new \growing\ ancestors can be introduced in an unstable tree.
    Thus, the only amoebots that can become \growing\ in an unstable tree are those that had \growing\ ancestors in this tree at the start of the energy run, but even those will become \growing\ at most once.
    So $B$ cannot become \growing\ multiple times to adopt $A$ more than once, a contradiction.
\end{proof}

We next show that all amoebots eventually join and remain in stable trees.

\begin{lemma} \label{app:lem:forestblocksfinite}
    Any energy run of $\sched^\demand$ contains at most a finite number of $g_\getpruned$-, $g_\askgrowth$-, and $g_\growforest$-supported executions of $\alpha_\energydist$.
\end{lemma}
\begin{proof}
    The predicates $g_\getpruned$, $g_\askgrowth$, and $g_\growforest$ depend only on the $\xstate$ and $\parent$ variables, neither of which are updated by the \harvestenergy\ and \shareenergy\ blocks.
    Thus, we may consider only the \getpruned, \askgrowth, and \growforest\ blocks when analyzing executions of $\alpha_\energydist$ supported by their predicates.

    Suppose to the contrary that an energy run of $\sched^\demand$ contains an infinite number of $g_\getpruned$-supported executions.
    With only a finite number of amoebots in the system, there must exist an amoebot $A$ that performs an infinite number of $g_\getpruned$-supported executions.
    Then an infinite number of times, $A$ must start as \pruning\ to satisfy $g_\getpruned$ and end as \idle\ after executing $\getpruned$.
    But by Lemma~\ref{app:lem:constantrejoin}, $A$ can only be pruned from and adopted into the forest a constant number of times in an energy run, a contradiction.

    Suppose instead that an energy run of $\sched^\demand$ contains an infinite number of $g_\askgrowth$-supported executions.
    Again, this implies some amoebot $A$ performs an infinite number of $g_\askgrowth$-supported executions.
    Then an infinite number of times, $A$ must be \xactive\ and have either an \idle\ neighbor or \asking\ child to satisfy $g_\askgrowth$ and then become \asking\ after executing $\askgrowth$.
    One way $A$ can return to \xactive\ from \asking\ is via pruning and later readoption into the forest, but Lemma~\ref{app:lem:constantrejoin} states that this can only happen a constant number of times per energy run.
    The only alternative is for $A$ to become \growing\ during a $g_\growforest$-supported execution by its parent and later reset itself to \xactive\ during its own $g_\growforest$-supported execution.
    So if $A$ performs an infinite number of $g_\askgrowth$-supported executions in this energy run, it must also perform an infinite number of $g_\growforest$-supported executions, which we address in the following final case.

    Suppose to the contrary that an amoebot $A$ executes an infinite number of $g_\growforest$-supported executions in an energy run of $\sched^\demand$.
    At the start of each of these infinite executions, $A$ must either be \growing\ or be a source with an \idle\ neighbor or \asking\ child.
    If $A$ is \growing, then it becomes \xactive\ after executing $\growforest$.
    The only way for $A$ to become \growing\ again is if its parent performs a $g_\growforest$-supported execution, which in turn is only possible if its grandparent performed an earlier $g_\growforest$-supported execution, and so on all the way up to the source amoebot rooting this tree.
    
    So it suffices to analyze the case when $A$ satisfies $g_\growforest$ as a source.
    Each time $A$ performs a $g_\growforest$-supported execution as a source, it adopts all its \idle\ neighbors into its (stable) tree.
    By Invariant~\ref{app:inv:energyrun:stabletrees}, these adopted amoebots will remain children of $A$ throughout this energy run.
    Thus, $A$ can perform a $g_\growforest$-supported execution as a source with an \idle\ neighbor only as many times as the number of its \idle\ neighbors, which is at most six if $A$ is contracted and at most eight if $A$ is expanded.
    
    The remaining possibility is that $A$ performs an infinite number of $g_\growforest$-supported executions as a source with an \asking\ child.
    The predicate $g_\askgrowth$ ensures that every asking signal that reaches $A$ originates at an \xactive\ amoebot with an \idle\ neighbor.
    Again, because there are only a finite number of amoebots in the system, an infinite number of asking signals reaching $A$ implies the existence of an amoebot $B$ in the stable tree rooted at $A$ that performs an infinite number of $g_\askgrowth$-supported executions as an \xactive\ amoebot with an \idle\ neighbor.
    Because $B$ is in a stable tree, the only way it can return to \xactive\ from \asking\ is to become \growing\ during a $g_\growforest$-supported execution by its parent and later reset itself to \xactive\ during its own $g_\growforest$-supported execution.
    During its own $g_\growforest$-supported execution, $B$ adopts any \idle\ neighbors it has.
    But it is not guaranteed that $B$ will have an \idle\ neighbor at the time of its $g_\growforest$-supported execution, even though it had one earlier: some neighbor could be \idle\ at the time $B$ performs its $g_\askgrowth$-supported execution, get adopted by a different amoebot by the time $B$ performs its $g_\growforest$-supported execution, and then become \idle\ again via pruning before $B$ performs its next $g_\askgrowth$-supported execution.
    However, $B$ can only ask but fail to adopt an \idle\ neighbor a constant number of times by Lemma~\ref{app:lem:constantrejoin}.    
    With any adoptee remaining in the stable tree throughout the energy run by Invariant~\ref{app:inv:energyrun:stabletrees} and at most a constant number of \idle\ neighbors to adopt, $B$ can perform at most a constant total number of $g_\askgrowth$-supported executions before adopting all its \idle\ children, a contradiction.

    Therefore, we conclude that the number of $g_\getpruned$-, $g_\askgrowth$-, and $g_\growforest$-supported executions in any energy run is finite, as desired.
\end{proof}

The next lemma is an analogous result for the \harvestenergy\ and \shareenergy\ blocks that move energy throughout the system. 

\begin{lemma} \label{app:lem:energyblocksfinite}
    Any energy run of $\sched^\demand$ contains at most a finite number of $g_\harvestenergy$- and $g_\shareenergy$-supported executions of $\alpha_\energydist$.
\end{lemma}
\begin{proof}
    Energy is never spent in an energy run (Invariant~\ref{app:inv:energyrun:nospend}).
    Thus, since every $g_\harvestenergy$-supported execution harvests a single unit of energy into the system, there can be at most $\numAmoebots\capacity$ such executions before the total harvested energy exceeds the total capacity of all $\numAmoebots$ amoebots' batteries.
    Analogously, since every $g_\shareenergy$-supported execution transfers one unit of energy from some parent amoebot to one of its children in $\forest$, any amoebot with $d$ descendants in $\forest$ can perform at most $d\capacity$ such executions before exceeding the total capacity of its descendants' batteries.
    None of the other blocks (\getpruned, \askgrowth, and \growforest) transfer energy, so once all amoebots' batteries are full, $g_\harvestenergy$ and $g_\shareenergy$ will be continuously dissatisfied for the remainder of the energy run.
\end{proof}

Combining Lemmas~\ref{app:lem:forestblocksfinite} and~\ref{app:lem:energyblocksfinite} shows that any energy run is finite.
But more importantly, they show that the unfair adversary exhibits weak fairness in an energy run.
Since the total number of $\alpha_\energydist$ executions in an energy run is finite, the unfair adversary will eventually be forced to activate any continuously enabled amoebot.
We formalize this result in the next lemma, concluding our arguments on energy run termination.

\begin{lemma} \label{app:lem:round}
    Consider any amoebot $A$ for which $\alpha_\energydist$ is enabled and would remain so until execution in some energy run of $\sched^\demand$.
    Then within one additional round, either $A$ executes $\alpha_\energydist$ or this energy run is ended by some $\alpha_i^\demand$ execution.
\end{lemma}
\begin{proof}
    Suppose $\alpha_\energydist$ is enabled for amoebot $A$ in round $r$.
    If an $\alpha_i^\demand$ execution ends this energy run by the completion of round $r+1$, we are done.
    Otherwise, this energy run extends through the remainder of round $r$ and---if round $r$ is finite---all of round $r+1$.
    
    Suppose to the contrary that $A$ is not activated in the remainder of round $r$ or at any time in round $r+1$.
    Recall from Section~\ref{subsec:model} that a (sequential) round ends once every amoebot that was enabled at its start has either completed an action execution or become disabled.
    By supposition, $A$ will remain enabled until its $\alpha_\energydist$ action is executed.
    So at least one of rounds $r$ and $r+1$ must never complete; i.e., at least one of them contains an infinite sequence of $\alpha_\energydist$ executions by enabled amoebots other than $A$.    
    There are only finitely many amoebots, so there must exist an amoebot $B \neq A$ that performs an infinite number of $\alpha_\energydist$ executions.
    Moreover, there are only five predicates that could support these executions, so there must exist a predicate $g \in \mathcal{G}$ such that $B$ performs an infinite number of $g$-supported executions of $\alpha_\energydist$.
    But Lemmas~\ref{app:lem:forestblocksfinite} and~\ref{app:lem:energyblocksfinite} show that any predicate can support at most a finite number of $\alpha_\energydist$ executions per energy run of $\sched^\demand$, a contradiction.
\end{proof}

With Lemma~\ref{app:lem:round} in place, we now argue about the progress and runtime of energy runs towards their overall goal of distributing energy to deficient amoebots in the system.
This next series of lemmas proves an $\bigo{\numAmoebots^2}$ upper bound on the number of rounds any energy run can take before all $\numAmoebots$ amoebots belong to stable trees (Lemma~\ref{app:lem:stabletime}).
Of course, an energy run could be ended by an $\alpha_i^\demand$ execution before all amoebots join stable trees, but this only helps our overall progress argument.
In the following lemmas, we prove our upper bound for \textit{uninterrupted energy runs} that continue until $\alpha_\energydist$ is disabled for all amoebots.
We first upper bound the time for any unstable tree to be dissolved by pruning.

\begin{lemma} \label{app:lem:prunetime}
    In an uninterrupted energy run of $\sched^\demand$, any amoebot $A$ at depth $d$ of an unstable tree $\tree$ will be pruned (i.e., set its children to \pruning, reset their $\parent$ pointers, and become \idle) within at most $d + 1$ rounds.\footnote{The \textit{depth} of a amoebot $A$ in a tree $\tree$ rooted at an amoebot $R$ is the number of nodes in the $(R,A)$-path in $\tree$ (i.e., the root $R$ is at depth $1$, and so on). The depth of a tree $\tree$ is $\max_{A \in \tree}\{\text{depth of } A\}$.}
\end{lemma}
\begin{proof}
    Argue by induction on $d$, the depth of $A$ in $\tree$.
    If $d = 1$, $A$ is the root of the unstable tree $\tree$ and thus must be \pruning\ by definition.
    So $A$ continuously satisfies $g_\getpruned$ since only a \pruning\ amoebot can change its own \xstate.
    By Lemma~\ref{app:lem:round}, $A$ will be activated and perform a $g_\getpruned$-supported execution within $d = 1$ additional round.
    Now suppose $d > 1$ and that every amoebot at depth at most $d - 1$ in $\tree$ is pruned within $d$ rounds.
    If $A$ is also pruned by round $d$, we are done.
    Otherwise, $A$ has been \pruning\ since at least the end of round $d$ when its parent in $\tree$ performed its own $g_\getpruned$-supported execution.
    So $A$ again continuously satisfies $g_\getpruned$ and must be activated by the end of round $d + 1$ by Lemma~\ref{app:lem:round}.
    Thus, in all cases, $A$ is pruned in at most $d + 1$ rounds.
\end{proof}

Once all unstable trees are dissolved, the newly \idle\ amoebots need to be adopted into stable trees.
Recall that members of stable trees must become \asking\ and then \growing\ before they can adopt their \idle\ neighbors as \xactive\ children.

\begin{lemma} \label{app:lem:growtime}
    In an uninterrupted energy run of $\sched^\demand$, any \asking\ amoebot $A$ at depth $d$ of a stable tree $\tree$ will become \growing\ within at most $2d - 2$ rounds.
\end{lemma}
\begin{proof}
    Recall that asking signals are propagated to the source root of a stable tree by \xactive\ parents performing $g_\askgrowth$-supported executions when they have \asking\ children.
    In the worst case, all non-source ancestors of $A$ are \xactive; i.e., no progress has been made towards propagating this asking signal.
    Since $A$ is in a stable tree and thus can't become \pruning, $A$ remains \asking\ until it becomes \growing.
    Thus, the \xactive\ parent of $A$ continuously satisfies $g_\askgrowth$ and will become \asking\ within one additional round by Lemma~\ref{app:lem:round}.
    Any \xactive\ ancestor of $A$ with an \asking\ child also continuously satisfies $g_\askgrowth$ and thus will become \asking\ within one additional round by Lemma~\ref{app:lem:round}.
    There are $d - 2$ \xactive\ ancestors strictly between $A$ and the source amoebot rooting this stable tree, so within at most $d - 2$ rounds the source amoebot will have an \asking\ child.
    The source amoebot will continuously satisfy $g_\growforest$ because of its \asking\ child, so it will make all its \asking\ children \growing\ within one additional round by Lemma~\ref{app:lem:round}.
    Similarly, \growing\ amoebots continuously satisfy $g_\growforest$ and pass their \growing\ \xstate\ to their \asking\ children within one additional round by Lemma~\ref{app:lem:round}.
    So $A$ must become \growing\ within another $d - 1$ additional rounds, for a total of at most $(d - 2) + 1 + (d - 1) = 2d - 2$ rounds.
\end{proof}

Combining Lemmas~\ref{app:lem:prunetime} and~\ref{app:lem:growtime} yields an upper bound on the time an uninterrupted energy run requires to organize all amoebots into stable trees.

\begin{lemma} \label{app:lem:stabletime}
    After at most $\bigo{\numAmoebots^2}$ rounds of any uninterrupted energy run of $\sched^\demand$, all $\numAmoebots$ amoebots belong to stable trees.
\end{lemma}
\begin{proof}
    If all amoebots already belong to stable trees, we are done.
    So suppose at least one amoebot is \idle\ or in an unstable tree.
    The system always contains at least one source amoebot (Invariant~\ref{app:inv:reachable:sources}), so the depth of any unstable tree is at most $\numAmoebots - 1$.
    By Lemma~\ref{app:lem:prunetime}, all members of unstable trees will be pruned and become \idle\ within at most $\numAmoebots$ rounds.
    
    Since the system remains connected (Invariant~\ref{app:inv:reachable:connected}) and always contains a source amoebot (Invariant~\ref{app:inv:reachable:sources}), there must exist an \idle\ amoebot $A$ that has at least one neighbor in a stable tree.
    \idle\ amoebots do not execute any actions, so at least one of its \xactive\ neighbors will continuously satisfy $g_\askgrowth$ and become \asking\ within one additional round by Lemma~\ref{app:lem:round}.
    The depth of any of these \asking\ neighbors of $A$ in their respective stable trees can be at most $\numAmoebots - 1$, counting all amoebots except $A$.
    So by Lemma~\ref{app:lem:growtime}, at least one of these \asking\ neighbors of $A$ will become \growing\ within at most $2(\numAmoebots - 1) - 2 \leq 2\numAmoebots$ rounds.
    \growing\ amoebots continuously satisfy $g_\growforest$, so within one additional round a \growing\ neighbor of $A$ will attempt to adopt an \idle\ neighbor by Lemma~\ref{app:lem:round}.
    The first such \growing\ neighbor must succeed in an adoption because $A$ is in its neighborhood.
    
    Thus, at least one \idle\ amoebot is adopted into a stable tree every $\bigo{\numAmoebots}$ rounds.
    There can be at most $\numAmoebots - 1$ amoebots initially outside stable trees, so we conclude that all amoebots are adopted into stable trees within $\numAmoebots + (\numAmoebots - 1) \cdot \bigo{\numAmoebots} = \bigo{\numAmoebots^2}$ rounds.
\end{proof}

Lemma~\ref{app:lem:stabletime} shows that after at most $\bigo{\numAmoebots^2}$ rounds of any energy run, all amoebots will belong to stable trees.
By Invariant~\ref{app:inv:energyrun:stabletrees}, they will remain there throughout the energy run; in particular, no amoebot will execute $g_\getpruned$-, $g_\askgrowth$-, or $g_\growforest$-supported executions after this point of the energy run.
For convenience, we refer to these sub-runs as \textit{stabilized energy runs}.
This next series of lemmas proves an $\bigo{n}$ upper bound on the \textit{recharge time}, i.e., the worst case number of rounds any stabilized energy run can take to fully recharge all $\numAmoebots$ amoebots, i.e., $A.\battery = \capacity$ for all amoebots $A$ (Lemma~\ref{app:lem:rechargetime}).

We make four observations that simplify this analysis, w.l.o.g.
First, we again consider uninterrupted energy runs as it only helps our overall progress argument if some $\alpha_i^\demand$ execution ends an energy run earlier.
Second, we assume all amoebots have initially empty batteries as this can only increase the recharge time.
Third, it suffices to analyze the recharge time of any one stable tree $\tree$ since trees are not reconfigured and do not interact in stabilized energy runs.
Fourth and finally, we show in the following lemma that the recharge time for $\tree$ is at most the recharge time for a simple path of the same number of amoebots.

\begin{lemma} \label{app:lem:pathrecharge}
    Suppose $\tree$ is a (stable) tree of $k$ amoebots rooted at a source amoebot $A_1$.
    If all amoebots in $\tree$ have initially empty batteries, then the recharge time for $\tree$ is at most the recharge time for a simple path $\mathcal{L} = (A_1, \ldots, A_k)$ in which $A_1$ is a source amoebot, $A_i.\parent = A_{i-1}$ for all $1 < i \leq k$, and all $k$ amoebots have initially empty batteries.
\end{lemma}
\begin{proof}
    Consider any tree $\mathcal{U}$ of $k$ amoebots rooted at a source amoebot $A_1$ and any sequence of amoebot activations $S$ representing an uninterrupted, stabilized energy run in which all amoebots' batteries are initially empty.
    Let $t_S(\mathcal{U})$ denote the number of rounds required to fully recharge all amoebots in $\mathcal{U}$ with respect to $S$ and let $t(\mathcal{U}) = \max_S\{t_S(\mathcal{U})\}$ denote the worst-case recharge time for $\mathcal{U}$.
    With this notation, our goal is to show that $t(\tree) \leq t(\mathcal{L})$.

    The \textit{maximum non-branching path} of a tree $\mathcal{U}$ is the longest directed path $(A_1, \ldots, A_\ell)$ starting at the source amoebot such that $A_{i+1}$ is the only child of $A_i$ in $\mathcal{U}$ for all $1 \leq i < \ell$.
    We argue by (reverse) induction on $\ell$, the length of the maximum non-branching path of $\tree$.
    If $\ell = k$, then $\tree$ and $\mathcal{L}$ are both simple paths of $k$ amoebots with initially empty batteries and thus $t(\tree) = t(\mathcal{L})$.
    So suppose $\ell < k$ and $t(\mathcal{U}) \leq t(\mathcal{L})$ for any tree $\mathcal{U}$ that comprises the same $k$ amoebots as $\tree$ with initially empty batteries, is rooted at amoebot $A_1$, and has at least $\ell + 1$ amoebots in its maximum non-branching path.
    Our goal is to modify the $\parent$ pointers in $\tree$ to form another tree $\tree'$ that has exactly one more amoebot in its maximum non-branching path and satisfies $t(\tree) \leq t(\tree')$.
    Since $\tree'$ has exactly $\ell + 1$ amoebots in its maximum non-branching path, the induction hypothesis implies that $t(\tree) \leq t(\tree') \leq t(\mathcal{L})$.

    We construct $\tree'$ from $\tree$ as follows.
    Let $(A_1, \ldots, A_\ell)$ be a maximum non-branching path of $\tree$, where $A_\ell$ is the ``closest'' amoebot to $A_1$ with multiple children, say $B_1, \ldots, B_c$ for some $c \geq 2$.
    Note that such an $A_\ell$ must exist because $\ell < k$.
    We form $\tree'$ by reassigning $B_i.\parent$ from $A_\ell$ to $B_1$ for each $2 \leq i \leq c$.
    Then $B_1$ is the only child of $A_\ell$ in $\tree'$, and thus $(A_1, \ldots, A_\ell, B_1)$ is the maximum non-branching path of $\tree'$ which has length $\ell + 1$.
    By the induction hypothesis, $t(\tree') \leq t(\mathcal{L})$.
    So it suffices to show that $t(\tree) \leq t(\tree')$.

    Consider any activation sequence $S = (s_1, \ldots, s_f)$ representing an uninterrupted, stabilized energy run where $s_f$ is the first amoebot activation after which all amoebots in $\tree$ have fully recharged batteries.
    Note that Lemma~\ref{app:lem:energyblocksfinite} implies $S$ has finite length and hence $s_f$ exists.
    We must show that there exists an activation sequence $S'$ such that $t_S(\tree) \leq t_{S'}(\tree')$.
    We construct $S'$ from $S$ so that the flow of energy through $\tree'$ mimics that of $\tree$.
    For each $s_i \in S$, we append a corresponding subsequence of activations $s_i'$ to the end of $S'$ that activates the same amoebot as $s_i$ and possibly some others as well, if needed.

    In almost all cases, $s_i$ is valid and has the same effect in both $\tree$ and $\tree'$, so we simply add $s_i' = (s_i)$ to $S'$.
    However, any activations $s_i$ in which $A_\ell$ passes energy to a child $B_j$, for $2 \leq j \leq c$, cannot be performed directly in $\tree'$ since $B_j$ is a child of $B_1$---not of $A_\ell$---in $\tree'$.
    We instead add a pair of activations $s_i' = (s_i^1, s_i^2)$ to $S'$ that have the effect of passing energy from $A_\ell$ to $B_j$ but use $B_1$ as an intermediary.
    There are two cases.
    If the battery of $B_1$ is not full (i.e., $B_1.\battery < \capacity$) just before $s_i$, then $s_i^1$ is a $g_\shareenergy$-supported execution of $\alpha_\energydist$ by $A$ passing a unit of energy to $B_1$ and $s_i^2$ is a $g_\shareenergy$-supported execution of $\alpha_\energydist$ by $B_1$ passing a unit of energy to $B_j$.
    Otherwise, these executions are reversed: $B_1$ passes a unit of energy to $B_j$ in $s_i^1$ and $A$ passes a unit of energy to $B_1$ in $s_i^2$.
    In any case, these activations are valid as their respective amoebots satisfy $g_\shareenergy$.
    
    Since all amoebots start with empty batteries and no energy is ever spent in an energy run (Invariant~\ref{app:inv:energyrun:nospend}), this construction of $S'$ ensures all amoebots' battery levels in $\tree$ and $\tree'$ are the same after each $s_i \in S$ and $s_i' \in S'$, respectively, for all $1 \leq i \leq f$.
    Thus, amoebots in $\tree$ and $\tree'$ only finish recharging after $s_f$ and $s_f'$, respectively.
    Each $s_i'$ activates the same amoebot as $s_i$ does and possibly one additional amoebot, so the number of rounds in $S'$ must be at least that in $S$.
    Therefore, we have $t_S(\tree) \leq t_{S'}(\tree')$, and since the choice of $S$ was arbitrary, we have $t(\tree) \leq t(\tree')$, as desired.
\end{proof}

By Lemma~\ref{app:lem:pathrecharge}, it suffices to analyze the case where $\tree$ is a simple path of $k$ amoebots with initially empty batteries.
To bound the recharge time, we use a \textit{dominance argument} between the sequential setting of stabilized energy runs and a parallel setting that is easier to analyze.
First, we prove that for any stabilized energy run, there exists a parallel version that makes at most as much progress towards recharging the system in the same number of rounds (Lemma~\ref{app:lem:dominance}).
We then upper bound the recharge time in parallel rounds (Lemma~\ref{app:lem:paralleltime}).
Combining these results gives an upper bound on the recharge time in sequential rounds.

Let an \textit{energy configuration} $E$ of the path $\mathcal{L} = (A_1, \ldots, A_k)$ encode the battery values of each amoebot $A_i$ as $E(A_i)$.
An \textit{energy schedule} is a sequence of energy configurations $(E_1, \ldots, E_t)$.
Given any sequence of amoebot activations $S$ representing a stabilized energy run, we define a \textit{sequential energy schedule} $(E_1^S, \ldots, E_t^S)$ where $E_r^S$ is the energy configuration of the path $\mathcal{L}$ at the start of sequential round $r$ in $S$.
Our dominance argument compares these schedules to parallel energy schedules, defined below.

\begin{definition} \label{app:def:parallelschedule}
    A \underline{parallel energy schedule} $(E_1, \ldots, E_t)$ is a schedule such that for all energy configurations $E_r$ and amoebots $A_i$ we have $E_r(A_i) \in [0, \capacity]$ and, for every $1 \leq r < t$, $E_{r+1}$ is reached from $E_r$ using the following for each amoebot $A_i$:
    \begin{itemize}
        \item $E_r(A_1) < \capacity$, so the source amoebot $A_1$ harvests energy from the external source with:
        \[E_{r+1}(A_1) = E_r(A_1) + 1\]
        
        \item $E_r(A_i) \geq 1$ and $E_r(A_{i+1}) < \capacity$, so $A_i$ passes energy to its child $A_{i+1}$ with:
        \[E_{r+1}(A_i) = E_r(A_i) - 1, \quad
        E_{r+1}(A_{i+1}) = E_r(A_{i+1}) + 1\]
    \end{itemize}
    Such a schedule is \underline{greedy} if the above actions are taken in parallel whenever possible.
\end{definition}

For an amoebot $A_i$ in an energy configuration $E$, let $\Delta_E(A_i) = \sum_{j=i}^k E(A_j)$ denote the total amount of energy in the batteries of amoebots $A_i, \ldots, A_k$ in $E$.
For any two battery configurations $E$ and $E'$, we say $E$ \textit{dominates} $E'$---denoted $E \succeq E'$---if and only if $\Delta_E(A_i) \geq \Delta_{E'}(A_i)$ for all amoebots $A_i \in \mathcal{L}$.

\begin{lemma} \label{app:lem:dominance}
    Given any activation sequence $S$ representing an uninterrupted, stabilized energy run on a simple path $\mathcal{L}$ of $k$ amoebots starting in an energy configuration $E_1^S$ in which all amoebots have empty batteries, there exists a greedy parallel energy schedule $(E_1, \ldots, E_t)$ with $E_1 = E_1^S$ such that $E_r^S \succeq E_r$ for all $1 \leq r \leq t$.
\end{lemma}
\begin{proof}
    The activation sequence $S$ and initial energy configuration $E_1^S$ yield a unique sequential energy schedule $(E_1^S, \ldots, E_t^S)$.
    Construct a corresponding parallel energy schedule $(E_1, \ldots, E_t)$ as follows.
    First, set $E_1 = E_1^S$.
    Then, for $1 < r \leq t$, obtain $E_r$ from $E_{r-1}$ by performing one \textit{parallel round} in which each amoebot greedily performs the actions of Definition~\ref{app:def:parallelschedule} if possible.
    We will show $E_r^S \succeq E_r$ for all $1 \leq r \leq t$ by induction on $r$.
    
    Since $E_1 = E_1^S$, we trivially have $E_1^S \succeq E_1$.
    So suppose $r \geq 1$ and for all rounds $1 \leq r' \leq r$ we have $E_{r'}^S \succeq E_{r'}$.
    Considering any amoebot $A_i$, we have $\Delta_{E_r^S}(A_i) \geq \Delta_{E_r}(A_i)$ by the induction hypothesis and want to show that $\Delta_{E_{r+1}^S}(A_i) \geq \Delta_{E_{r+1}}(A_i)$.    
    First suppose the inequality from the induction hypothesis is strict---i.e., $\Delta_{E_r^S}(A_i) > \Delta_{E_r}(A_i)$---meaning strictly more energy has been passed into $A_i, \ldots, A_k$ in the sequential setting than in the parallel one by the start of round $r$.
    No energy is spent in an energy run (Invariant~\ref{app:inv:energyrun:nospend}), so we know $\Delta_{E_{r+1}^S}(A_i) \geq \Delta_{E_r^S}(A_i)$.
    Because all energy transfers pass one unit of energy either from the external energy source to the source amoebot $A_1$ or from a parent $A_i$ to its child $A_{i+1}$, we have that $\Delta_{E_r^S}(A_i) \geq \Delta_{E_r}(A_i) + 1$.
    But by Definition~\ref{app:def:parallelschedule}, an amoebot can receive at most one unit of energy per parallel round, so we have:
    \[\Delta_{E_{r+1}^S}(A_i) \geq \Delta_{E_r^S}(A_i) \geq \Delta_{E_r}(A_i) + 1 \geq \Delta_{E_{r+1}}(A_i).\]
    
    Thus, it remains to consider when $\Delta_{E_r^S}(A_i) = \Delta_{E_r}(A_i)$, meaning the amount of energy passed into $A_i, \ldots, A_k$ is exactly the same in the sequential and parallel settings by the start of round $r$.
    It suffices to show that if $A_i$ receives an energy unit in parallel round $r$, then it also does so in the sequential round $r$.
    We first prove that if $A_i$ receives an energy unit in parallel round $r$, then there is at least one unit of energy for $A_i$ to receive in sequential round $r$.
    If $A_i$ is the source amoebot, this is trivial: the external source of energy is its infinite supply.
    Otherwise, $i > 1$ and we must show $E_r^S(A_{i-1}) \geq 1$.
    We have $\Delta_{E_r^S}(A_i) = \Delta_{E_r}(A_i)$ by supposition and $\Delta_{E_r^S}(A_{i-1}) \geq \Delta_{E_r}(A_{i-1})$ by the induction hypothesis, so
    \begin{align*}
        E_r^S(A_{i-1}) &= \sum_{j=i-1}^k E_r^S(A_j) - \sum_{j=i}^k E_r^S(A_j) \\
        &= \Delta_{E_r^S}(A_{i-1}) - \Delta_{E_r^S}(A_i) \\
        &\geq \Delta_{E_r}(A_{i-1}) - \Delta_{E_r}(A_i) \\
        &= \sum_{j=i-1}^k E_r(A_j) - \sum_{j=i}^k E_r(A_j) \\
        &= E_r(A_{i-1}) \geq 1,
    \end{align*}
    where the final inequality follows from the fact that we presumed $A_i$ receives one energy unit in parallel round $r$ which must come from its parent $A_{i-1}$ since $A_i$ is not a source amoebot.
    
    Next, we show that if $A_i$ receives an energy unit in parallel round $r$, then $E_r^S(A_i) \leq \capacity - 1$; i.e., $A_i$ has enough room in its battery to receive an energy unit during sequential round $r$.
    By supposition we have $\Delta_{E_r^S}(A_i) = \Delta_{E_r}(A_i)$ and by the induction hypothesis we have $\Delta_{E_r^S}(A_{i+1}) \geq \Delta_{E_r}(A_{i+1})$.
    Combining these facts, we have
    \begin{align*}
        E_r^S(A_i) &= \sum_{j=i}^k E_r^S(A_j) - \sum_{j=i+1}^k E_r^S(A_j) \\
        &= \Delta_{E_r^S}(A_i) - \Delta_{E_r^S}(A_{i+1}) \\
        &\leq \Delta_{E_r}(A_i) - \Delta_{E_r}(A_{i+1}) \\
        &= \sum_{j=i}^k E_r(A_j) - \sum_{j=i+1}^k E_r(A_j) \\
        &= E_r(A_i) \leq \capacity - 1,
    \end{align*}
    where the final inequality follows from the following observation about how energy is transferred in a parallel schedule.
    It is easy to see from Definition~\ref{app:def:parallelschedule} that if $j > i$, then $E_{r-1}(A_i) \leq E_{r-1}(A_j)$; i.e., an amoebot can only have as much energy as any one of its descendants in a greedy parallel schedule.
    So if $A_i$ is receiving energy, it cannot have a full battery; otherwise, all of its descendants' batteries must also be full, leaving $A_i$ unable to simultaneously transfer energy to make room for the new energy it is receiving.
    Thus, $A_i$ must have capacity for at least one energy unit at the start of sequential round $r$, as desired.

    Thus, we have shown that if $A_i$ receives a unit of energy in parallel round $r$, then (1) either $i = 1$ or $E_r^S(A_{i-1}) \geq 1$, and (2) $E_r^S(A_i) \leq \capacity - 1$, meaning that at the start of sequential round $r$, there is both an energy unit available to pass to $A_i$ and $A_i$ has sufficient capacity to receive it.
    In other words, either $A_i$ is a source and continuously satisfies $g_\harvestenergy$ or its parent $A_{i-1}$ continuously satisfies $g_\shareenergy$.
    Since no energy is spent in an energy run (Invariant~\ref{app:inv:energyrun:nospend}), additional activations in sequential round $r$ can only increase the amount of energy available to pass to $A_i$ and increase the space available in $A_i.\battery$.
    Thus, by Lemma~\ref{app:lem:round}, $A_i$ must receive at least one energy unit in sequential round $r$, proving that $\Delta_{E_{r+1}^S}(A_i) \geq \Delta_{E_{r+1}}(A_i)$ in all cases.
    Since the choice of $A_i$ was arbitrary, we have shown $E_{r+1}^S \succeq E_{r+1}$.
\end{proof}

To conclude the dominance argument, we bound the number of parallel rounds needed to recharge a path of $k$ amoebots.
Combined with Lemma~\ref{app:lem:dominance}, this gives an upper bound on the worst case number of sequential rounds for any stabilized energy run to do the same.

\begin{lemma} \label{app:lem:paralleltime}
    Let $(E_1, \ldots, E_t)$ be the greedy parallel energy schedule on a simple path $\mathcal{L}$ of $k$ amoebots where $E_1(A_i) = 0$ and $E_t(A_i) = \capacity$ for all amoebots $A_i \in \mathcal{L}$.
    Then $t = k\capacity = \bigo{k}$.
\end{lemma}
\begin{proof}
    Argue by induction on $k$, the number of amoebots in path $\mathcal{L}$.
    If $k = 1$, then $A_1 = A_k$ is the source amoebot that harvests one unit of energy per parallel round from the external energy source by Definition~\ref{app:def:parallelschedule}.
    Since $A_1$ has no children to which it may pass energy, it is easy to see that it will harvest $\capacity$ energy in exactly $\capacity = \Theta(1)$ parallel rounds.
    
    Now suppose $k > 1$ and that any path of $j \in \{1, \ldots, k-1\}$ amoebots fully recharges in $j\capacity$ parallel rounds.
    Once an amoebot $A_i$ has received energy for the first time, it follows from Definition~\ref{app:def:parallelschedule} that $A_i$ will receive a unit of energy from $A_{i-1}$ (or the external energy source, in the case that $i = 1$) in every subsequent parallel round until $A_i.\battery = \capacity$.
    Similarly, Definition~\ref{app:def:parallelschedule} ensures that $A_i$ will pass a unit of energy to $A_{i+1}$ in every subsequent parallel round until $A_{i+1}.\battery = \capacity$.
    Thus, once $A_i$ receives energy for the first time, $A_i$ effectively acts as an external energy source for the remaining amoebots $A_{i+1}, \ldots, A_k$.
    
    The source amoebot $A_1$ first harvests energy from the external energy source in parallel round $1$ and thus acts as a continuous energy source for $A_2, \ldots, A_k$ in all subsequent rounds.
    By the induction hypothesis, we know $A_2, \ldots, A_k$ will fully recharge in $(k-1)\capacity$ parallel rounds, after which $A_1$ will no longer pass energy to $A_2$.
    The source amoebot $A_1$ harvests one energy unit from the external energy source per parallel round and already has $A_1.\battery = 1$, so in an additional $\capacity - 1$ parallel rounds we have $A_1.\battery = \capacity$.
    Therefore, the path $A_1, \ldots, A_k$ fully recharges in $1 + (k-1)\capacity + \capacity - 1 = k\capacity = \bigo{k}$ parallel rounds, as required.
\end{proof}

Combining the lemmas of this section yields the following bound on the recharge time.

\begin{lemma} \label{app:lem:rechargetime}
    After at most $\bigo{\numAmoebots}$ rounds of any uninterrupted, stabilized energy run of $\sched^\demand$, all $\numAmoebots$ amoebots have full batteries.
\end{lemma}
\begin{proof}
    Consider any stabilized energy run of $\sched^\demand$.
    By definition, this energy run starts in a configuration where all amoebots belong to stable trees, and by Invariant~\ref{app:inv:energyrun:stabletrees} the structure of $\forest$ will not change throughout this energy run.
    So consider any (stable) tree $\tree \in \forest$ and suppose, in the worst-case, that all amoebots have initially empty batteries.
    By Lemma~\ref{app:lem:pathrecharge}, the recharge time for $\tree$ is at most the recharge time for a path $\mathcal{L}$ of $|\tree|$ amoebots.
    Any activation sequence representing a recharge process for $\mathcal{L}$ runs at least as fast as a greedy parallel energy schedule for $\mathcal{L}$ (Lemma~\ref{app:lem:dominance}), and the latter must fully recharge $\mathcal{L}$ in $\bigo{|\mathcal{L}|} = \bigo{|\tree|}$ rounds (Lemma~\ref{app:lem:paralleltime}).
    Since $\tree$ contains at most $\numAmoebots$ amoebots, the lemma follows.
\end{proof}

We can now prove Theorem~\ref{thm:main}, concluding our analysis.

\begin{proof}[Proof of Theorem~\ref{thm:main}]
    As in the statement of Theorem~\ref{thm:main}, consider any energy-compatible amoebot algorithm $\alg$ and demand function $\demand : \alg \to \{1, 2, \ldots, \capacity\}$, and let $\alg^\demand$ be the algorithm produced from $\alg$ and $\demand$ by the energy distribution framework.
    Let $C_0$ be any (legal) connected initial configuration for $\alg$ and let $C_0^\demand$ be its extension for $\alg^\demand$ that designates at least one source amoebot and adds the energy distribution variables with their initial values (Table~\ref{tab:frameworkvariables}) to all amoebots.
    Finally, consider any sequential execution $\sched^\demand$ of $\alg^\demand$ starting in $C_0^\demand$.
    Let $\sched^\demand_\alpha$ be its subsequence of $\alpha_i^\demand$ action executions and $\sched_\alpha$ be the corresponding sequence of $\alpha_i$ action executions.
    By Lemma~\ref{app:lem:equivalence}, $\sched_\alpha$ is a valid sequential execution of the original algorithm $\alg$.
    Since $\alg$ is assumed to be energy-compatible, its sequential executions always terminate.
    Thus, $\sched_\alpha$ is finite and, by extension, so is $\sched^\demand_\alpha$.
    This implies that the overall execution $\sched^\demand$ contains at most a finite number of distinct energy runs.
    Each of these energy runs is finite by Lemmas~\ref{app:lem:forestblocksfinite} and~\ref{app:lem:energyblocksfinite}, so we conclude that $\sched^\demand$ in total is finite.
    
    Let $C^\demand$ be the terminating configuration of $\sched^\demand$, but suppose to the contrary that there does not exist a sequential execution of $\alg$ starting in $C_0$ that terminates in the configuration $C$ obtained from $C^\demand$ by removing the energy distribution variables.
    We have already shown that $\sched_\alpha$ is a valid sequential execution of $\alg$ starting in $C_0$.
    Moreover, $\alg^\demand$ only moves amoebots and modifies variables of algorithm $\alg$ during $\alpha_i^\demand$ executions, so all amoebot movements and updates to variables of algorithm $\alg$ are identical in $\sched_\alpha$ and $\sched^\demand$.
    Thus, $\sched_\alpha$ must reach configuration $C$ but---for the sake of contradiction---cannot terminate there; i.e., there must exist an amoebot $A$ for which some action $\alpha_i$ is enabled in $C$ but all amoebots are disabled in $C^\demand$; in particular, the corresponding action $\alpha_i^\demand$ is disabled for $A$ in $C^\demand$.

    The guard $g_i^\demand$ of action $\alpha_i^\demand$ requires three properties: $A$ satisfies guard $g_i$ of action $\alpha_i$, $A$ and its neighbors are not \idle\ or \pruning, and $A$ has at least $\demand(\alpha_i)$ energy.
    We know $A$ satisfies $g_i$ in $C^\demand$ because $\alpha_i$ is enabled for $A$ in $C$.
    No amoebot in $C^\demand$ can be \idle, since the connectivity of $C^\demand$ (Invariant~\ref{app:inv:reachable:connected}) implies that some amoebot would satisfy $g_\askgrowth$ or $g_\growforest$ and thus be enabled by $\alpha_\energydist$, contradicting $C^\demand$ as a terminating configuration.
    Similarly, no amoebot can be \pruning\ in $C^\demand$ since this amoebot would satisfy $g_\getpruned$.
    So suppose that in $C^\demand$, $A.\battery < \demand(\alpha_i) \leq \capacity$.
    Then $A$ cannot be a source, since it would satisfy $g_\harvestenergy$.
    So $A$ must be \xactive, \asking, or \growing, all of which imply $A$ has a parent in forest $\forest$.
    The connectivity of $C^\demand$ (Invariant~\ref{app:inv:reachable:connected}) implies that some ancestor of $A$ satisfies $g_\harvestenergy$ or $g_\shareenergy$: either the parent of $A$ satisfies $g_\shareenergy$, or the parent of $A$ has insufficient energy to share but the grandparent of $A$ satisfies $g_\shareenergy$, and so on up to the source root of the tree which, if it does not have sufficient energy to share, must satisfy $g_\harvestenergy$.
    Therefore, we reach a contradiction in all cases, proving that if $C^\demand$ is a terminating configuration for $\sched^\demand$, then $C$ is a terminating configuration for $\sched_\alpha$ and thus there exists a sequential execution of $\alg$ starting in $C_0$ that terminates in $C$.

    We conclude by proving the runtime overhead bound.
    Let $\algruntime$ be the maximum number of action executions in any sequential execution of $\alg$ on $\numAmoebots$ amoebots.
    We know $\algruntime$ is finite because $\alg$ is energy-compatible.
    By Lemma~\ref{app:lem:equivalence}, any sequential execution of $\alg^\demand$ contains at most $\algruntime + 1$ energy runs, and each energy run terminates in at most $\bigo{\numAmoebots^2}$ rounds by Lemmas~\ref{app:lem:stabletime} and~\ref{app:lem:rechargetime}.
    Therefore, we conclude that any sequential execution of $\alg^\demand$ terminates in at most $\bigo{\numAmoebots^2} \cdot (\algruntime + 1) = \bigo{\numAmoebots^2\algruntime}$ rounds.
\end{proof}

\section{Omitted Analysis of Concurrency-Compatibility} \label{app:asyncproofs}

This section contains the technical material omitted from Section~\ref{sec:concurrency} due to space constraints.

\begin{algorithm}[t]
    \caption{Expansion-Robust Variant $\alg^E$ of Algorithm $\alg$ for Amoebot $A$} \label{alg:expandrobust}
    \begin{algorithmic}[1]
        \Statex \textbf{Input}: Algorithm $\alg = \{[\alpha_i : g_i \to ops_i] : i \in \{1, \ldots, m\}\}$ satisfying Conventions~\ref{conv:valid} and~\ref{conv:phases}.
        \State Set $\alpha_0^E : (\exists$ port $p$ of $A : A.\xflag_p = \true) \to$ \Write$(\bot, \xflag_p, \false)$.
        \For {each action $[\alpha_i : g_i \to ops_i] \in \alg$}
            \State Set $g_i^E \gets g_i$ with $N(A)$ replaced by $N^E(A)$ and connections defined w.r.t.\ $N^E(A)$.
            \State Set $ops_i^E \gets$ ``Do:
            \Indent
                \For {each port $p$ of $A$} \Write$(\bot, \xflag_p, \false)$.  \Comment{Reset own expand flags.}  \label{alg:expandrobust:resetown}
                \EndFor
                \For {each unique neighbor $B \in \Connected()$}
                    \For{each port $p$ of $B$} \Write$(B, \xflag_p, \false)$.  \Comment{Reset neighbors' expand flags.}  \label{alg:expandrobust:resetnbr}
                    \EndFor
                \EndFor
                \State Execute each operation of $ops_i$ with connections defined w.r.t.\ $N^E(A)$.
                \If {a \Pull\ or \Push\ operation was executed with neighbor $B$}
                    \For {each new port $p$ of $A$ not connected to $B$} \Write$(\bot, \xflag_p, \true)$.  \label{alg:expandrobust:handover1}
                    \EndFor
                    \For {each new port $p$ of $B$ not connected to $A$} \Write$(B, \xflag_p, \true)$.  \label{alg:expandrobust:handover2}
                    \EndFor
                \ElsIf {an \Expand\ operation was successfully executed}
                    \For {each new port $p$ of $A$} \Write$(\bot, \xflag_p, \true)$.  \label{alg:expandrobust:expand}
                    \EndFor
                \ElsIf {an \Expand\ operation failed in its execution} undo $ops_i$.''
                \label{alg:expandrobust:undo}
                \EndIf
            \EndIndent
        \EndFor
        \State \Return $\alg^E = \{[\alpha_i^E : g_i^E \to ops_i^E] : i \in \{0, \ldots, m\}\}$.
    \end{algorithmic}
\end{algorithm}

\begin{lemma} \label{app:lem:expandcorrespond}
    If amoebot algorithm $\alg$ is expansion-corresponding, it is also expansion-robust.
\end{lemma}
\begin{proof}
    To prove termination, suppose to the contrary that all sequential executions of $\alg$ starting in $C_0$ terminate, but there exists some infinite sequential execution $\sched^E$ of $\alg^E$ starting in $C_0^E$.
    Algorithm $\alg$ is expansion-corresponding, so there is a sequential execution $\sched$ that is identical to $\sched^E$, modulo executions of $\alpha_0^E$.
    Execution $\sched$ terminates by supposition, so $\sched^E$ must contain an infinite number of $\alpha_0^E$ executions after its final $\alpha_{i \neq 0}^E$ execution.
    But $\alpha_0^E$ executions only reset expand flags, and there are only a finite number of amoebots and a constant number of expand flags per amoebot to reset, a contradiction.

    Correctness follows from the same observation.
    Only $\alpha_{i \neq 0}^E$ executions move amoebots and modify variables of $\alg$.
    Since every sequential execution $\sched^E$ of $\alg^E$ starting in $C_0^E$ represents an identical sequential execution $\sched$ of $\alg$ starting in $C_0$ (after removing the $\alpha_0^E$ executions), and since $\sched^E$ terminates whenever $\sched$ terminates by the above argument, we conclude that they must terminate in configurations that are identical, modulo expand flags.
\end{proof}

Before proving that the energy distribution framework preserves expansion-correspondence, we need one helper lemma characterizing established neighbors in $\alg^\demand$.

\begin{lemma} \label{app:lem:established}
    During an execution of $(\alg^\demand)^E$, if an amoebot $A$ has a neighbor $B \in N(A)$ that is \idle, \pruning, or a child of $A$, then $B \in N^E(A)$.
\end{lemma}
\begin{proof}
    Any neighbor $B \in N(A) \setminus N^E(A)$ expanded into $N(A)$ during an \Expand\ operation by $B$, a \Push\ operation by $B$, or a \Pull\ operation by some other amoebot pulling $B$.
    Any movement in $(\alg^\demand)^E$ occurs in an $(\alpha_i^\demand)^E$ execution, whose guard requires that both the executing amoebot and all its established neighbors are not \idle\ or \pruning.
    Thus, regardless of whether $B$ is initiating the movement (an \Expand\ or \Push) or is participating in it (a \Pull), $B$ cannot be \idle\ or \pruning\ when it enters $N(A)$.
    Any subsequent action execution that could make $B$ \idle\ or \pruning\ must also reset its expand flags (Algorithm~\ref{alg:expandrobust}, Line~\ref{alg:expandrobust:resetnbr}).
    So there are never \idle\ or \pruning\ neighbors in $N(A) \setminus N^E(A)$.

    Next consider any child $B$ of $A$.
    Amoebot $B$ became a child of $A$ when $A$ adopted it during a $g_\growforest$-supported execution of $\alpha_\energydist^E$.
    During this execution, $A$ reset all expand flags of $B$ (Algorithm~\ref{alg:expandrobust}, Line~\ref{alg:expandrobust:resetnbr}).
    As long as $B$ is a child of $A$, its expand flags facing $A$ remain reset.
    Thus, $B \in N^E(A)$.
\end{proof}

We can now prove the main lemma of this section.

\begin{lemma} \label{app:lem:edfconcurrency}
    For any energy-compatible, expansion-corresponding algorithm $\alg$ and demand function $\demand : \alg \to \{1, 2, \ldots, \capacity\}$, the algorithm $\alg^\demand$ produced from $\alg$ and $\demand$ by the energy distribution framework is concurrency-compatible.
\end{lemma}
\begin{proof}
    By Theorem~\ref{thm:main}, we know that every sequential execution of $\alg^\demand$ terminates.
    It remains to show that $\alg^\demand$ satisfies the validity, phase structure, and expansion-robustness conventions.
    
    By supposition, every action $\alpha_i \in \alg$ in the original algorithm is valid, i.e., its execution is successful whenever it is enabled and all other amoebots are inactive.
    Since the guard $g_i$ of $\alpha_i$ is a necessary condition for the energy-constrained version $\alpha_i^\demand$ to be enabled, we know this validity carries over to the compute and movement phases of $\alpha_i$.
    The only new operations added by the energy distribution framework in the $\alpha_i^\demand$ and $\alpha_\energydist$ actions are \Connected\ operations (which never fail) and \Read\ and \Write\ operations involving existing neighbors.
    All of these must succeed, so every action of $\alg^\demand$ is valid.

    It is easy to see that $\alg^\demand$ satisfies the phase structure convention.
    Its only movements are in the $\alpha_i^\demand$ actions, each of which has at most one movement operation that it executes last.
    Moreover, the energy distribution framework does not add any \Lock\ or \Unlock\ operations.

    It remains to show $\alg^\demand$ is expansion-robust, and by Lemma~\ref{app:lem:expandcorrespond}, it suffices to show $\alg^\demand$ is expansion-corresponding.
    We first show that if some action of $(\alg^\demand)^E$ is enabled for an amoebot $A$ w.r.t.\ $N^E(A)$, then the corresponding action of $\alg^\demand$ is enabled for $A$ w.r.t.\ $N(A)$.
    We may safely consider only the guard conditions that depend on an amoebot's neighborhood; all others evaluate identically regardless of neighborhood.
    \begin{itemize}
        \item If $(\alpha_i^\demand)^E$ is enabled for an amoebot $A$, then $A$ must satisfy $g_i^E$---i.e., $A$ satisfies the guard $g_i$ of $\alpha_i \in \alg$ w.r.t.\ $N^E(A)$---and neither $A$ nor its established neighbors can be \idle\ or \pruning.
        Algorithm $\alg$ is expansion-corresponding by supposition, so this implies that $A$ must satisfy $g_i$ w.r.t.\ $N(A)$ as well.
        Moreover, Lemma~\ref{app:lem:established} ensures that if there are no \idle\ or \pruning\ neighbors in $N^E(A)$, there are none in $N(A)$ either.

        \item Suppose $\alpha_\energydist^E$ is enabled for an amoebot $A$ because $A$ has an \idle\ neighbor or an \asking\ child $B \in N^E(A)$, a condition in both $g_\askgrowth$ and $g_\growforest$.
        We know $N^E(A) \subseteq N(A)$, so $\alpha_\energydist$ must be enabled for $A$ w.r.t.\ $N(A)$ as well.

        \item Suppose $\alpha_\energydist^E$ is enabled for an amoebot $A$ because $A$ has a child $B \in N^E(A)$ whose battery is not full, a condition in $g_\shareenergy$.
        By the same argument as above, we have $N^E(A) \subseteq N(A)$, so $\alpha_\energydist$ must be enabled for $A$ w.r.t.\ $N(A)$ as well.
    \end{itemize}

    Finally, we show that the executions of any action of $(\alg^\demand)^E$ w.r.t.\ $N^E(A)$ and the corresponding action of $\alg^\demand$ w.r.t.\ $N(A)$ by the same amoebot $A$ are identical.
    We may safely focus only on the parts of action executions that depend on or interact with an amoebot's neighbors; all others execute identically regardless of neighborhood.
    \begin{itemize}
        \item If $A$ executes an $(\alpha_i^\demand)^E$ action, it emulates the operations of $\alpha_i \in \alg$ w.r.t.\ $N^E(A)$.
        But algorithm $\alg$ is expansion-corresponding by supposition, which immediately implies that an execution of $\alpha_i$ w.r.t.\ $N(A)$ is identical.

        \item If $A$ executes an $(\alpha_i^\demand)^E$ action or the \getpruned\ block of $\alpha_\energydist^E$, it may update its children's \xstate\ and \parent\ variables during $\textsc{Prune}(\,)$.
        By Lemma~\ref{app:lem:established}, any child of $A$ in $N(A)$ is also in $N^E(A)$, so the same children are pruned.

        \item If $A$ executes the \growforest\ block of $\alpha_\energydist^E$, it adopts all its \idle\ neighbors as an \xactive\ children.
        Any \idle\ neighbor $B \in N^E(A)$ that $A$ adopts must also be adopted when $A$ executes $\alpha_\energydist$ since $N^E(A) \subseteq N(A)$.
        But if there are no \idle\ neighbors in $N^E(A)$ for $A$ to adopt, there cannot be any in $N(A)$ either by Lemma~\ref{app:lem:established}.
        Thus, either the same \idle\ neighbors or no neighbors are adopted.

        \item If $A$ executes the \growforest\ block of $\alpha_\energydist^E$, it updates any \asking\ children to \growing.
        By Lemma~\ref{app:lem:established}, any child of $A$ in $N(A)$ is also in $N^E(A)$, so the same children are updated in $\alpha_\energydist$.

        \item If $A$ executes the \shareenergy\ block of $\alpha_\energydist^E$, it transfers an energy unit to one of its children $B \in N^E(A)$ whose battery is not full.
        We know $N^E(A) \subseteq N(A)$, so $B$ is also a possible recipient of this energy in $\alpha_\energydist$. \qedhere
    \end{itemize}
\end{proof}

\else\fi

\end{document}